\documentclass[12pt]{article}
\usepackage{setspace}
\usepackage[margin=1in]{geometry}

\usepackage[utf8]{inputenc}
\usepackage{hanging}	
\usepackage{soul}
\usepackage{xcolor}
\usepackage{amsmath}
\usepackage{amsthm}
\usepackage{amssymb}
\usepackage{amsfonts}
\usepackage{setspace}
\usepackage{multirow}
\usepackage{caption}
\DeclareCaptionLabelFormat{AppendixTables}{A.#2}
\usepackage{subcaption}
\usepackage{graphics}
\usepackage{lscape}
\usepackage{changepage}
\usepackage{graphicx}
\usepackage{xcolor,soul}
\usepackage{booktabs}
\usepackage{tikz-cd}
\usepackage[makeroom]{cancel}
\usepackage{mathrsfs}
\usepackage{algorithm}
\usepackage{algpseudocode}
\usepackage{tikz}
\usetikzlibrary{positioning, fit, backgrounds}

\usepackage{natbib}
\usepackage{xr-hyper}
\usepackage{hyperref}
\usepackage{subfiles}
\usepackage[bottom]{footmisc}
\usepackage{booktabs,longtable,pdflscape,array}

\tikzset{
  symbol/.style={
    draw=none,
    every to/.append style={
      edge node={node [sloped, allow upside down, auto=false]{$#1$}}}
  }
}

\theoremstyle{definition}
\newtheorem{theorem}{Theorem}[section]  \newtheorem{assumption}{Assumption}[section]  \newtheorem{lemma}{Lemma}[section] 
\newtheorem{corollary}{Corollary}[section] \newtheorem{proposition}{Proposition}[section]

\title{Risk-Optimal Curvature Selection for Finite-Sample Cressie--Read Moment Estimation}
\author{Jieun Lee\footnote{Corresponding author. Department of Economics, Emory University.  \textit{jieun.lee@emory.edu; \\jieunlee.sophia@gmail.com.}} \; and Anil K. Bera\footnote{Department of Economics, University of Illinois Urbana-Champaign. \textit{abera@illinois.edu.}}}

\begin{document}

\maketitle

\begin{abstract} 
We propose a finite-sample risk-optimal selection criterion for Cressie--Read
power divergence (CRPD) estimation in overidentified moment-based models. The
CRPD family, dual to generalized empirical likelihood, is indexed by the power
parameter \(\gamma\), with empirical likelihood and exponential tilting as
benchmarks. Although \(\gamma\) is conventionally fixed at a
researcher-chosen value, we argue that it should be interpreted as a data-tunable
curvature parameter governing the finite-sample behavior of the CRPD objective.
Through implied probability weights and associated Lagrange multipliers,
\(\gamma\) affects how the empirical distribution is reweighted to enforce the
moment restrictions, even when population identification is unchanged. The
proposed criterion selects \(\gamma\) by minimizing an estimation- and
system-oriented risk measure. It combines a structural component, which measures
finite-sample distortion in the estimate of the structural parameter relative to
a first-order GMM benchmark, with a multiplier-stability component, which
measures the cost of moment enforcement in the full estimator--multiplier
system. A researcher-specified weight determines the relative importance of the
two components, allowing the selection rule to prioritize structural accuracy,
multiplier stability, or a balance between them. The resulting selector is
designed to reduce second-order finite-sample distortion while discouraging
unstable multipliers, concentrated implied weights, and proximity to the
feasible-probability boundary. Monte Carlo simulations show that the selected
CRPD estimator remains approximately centered around the structural parameter
while improving finite-sample stability. An empirical illustration using Owen's
dairy-cow data shows that similar point estimates can correspond to different
implied weighting schemes, highlighting the practical role of \(\gamma\) as a
curvature parameter in moment-based estimation.

\vspace{2mm}


\end{abstract}

\section{Introduction}

First-order asymptotic analysis is widely used because of its simplicity and
analytical tractability. Its usefulness, however, depends on whether the
first-order approximation provides an accurate description of the estimator in
the sample sizes relevant for empirical work. In many applications, this
requirement may not be well satisfied. Observations may be costly, rare, or
constrained by institutional or temporal factors, so that the sample size cannot
be meaningfully increased. Even when the nominal sample size is moderate, the
sample moments may still be noisy or weakly informative. In such environments,
estimation problems may display small-sample features, and first-order
approximations may provide an incomplete guide to finite-sample performance.

These considerations motivate the use of Cressie--Read power divergence (CRPD)
estimators \citep{CressieandRead1984} in moment-based estimation. CRPD
estimators are closely connected to the generalized empirical likelihood (GEL)
class through a well-known duality. The GEL literature has extensively
documented favorable finite-sample properties relative to GMM
(see, e.g., \citet{Owen1988}; \citet{QinandLawless1994};
\citet{HansenHeatonandYaron1996}; \citet{KitamuraandStutzer1997};
\citet{ImbensSpadyandJohnson1998}; \citet{NeweyandSmith2004};
\citet{Schennach2007}), particularly in small to moderate samples. Through this
duality, CRPD inherits a moment-based reweighting structure that incorporates
the finite-sample geometry of the moment restrictions into the estimation
procedure.

A key feature of CRPD is the \textit{power parameter}, denoted by \(\gamma\),
which indexes a family of divergence functions. Different values of
\(\gamma\) correspond to distinct members of the divergence family, with
canonical choices including empirical likelihood, exponential tilting, and
continuous updating. In practice, however, the power parameter is often chosen
by convention. Because fixed-\(\gamma\) estimators are first-order
asymptotically equivalent under correct specification
(see, e.g., \citet{NeweyandSmith2004}), the choice of \(\gamma\) is commonly
viewed as asymptotically irrelevant.

This paper argues that \(\gamma\) should be viewed not merely as a
conventionally fixed index of the CRPD family, but as a data-tunable curvature
parameter. We show that, although \(\gamma\) is first-order irrelevant under
correct moment specification, it governs finite-sample behavior through the
curvature of the divergence objective, implied probability weights, and
Lagrange multipliers. Building on this characterization, we propose a novel
second-order system-risk criterion that selects \(\gamma\) by balancing
structural accuracy against the stability of moment enforcement. The latter is
important because CRPD enforces the sample moment restrictions by reweighting
the empirical distribution; large multipliers or highly concentrated implied
weights indicate that moment restrictions are being satisfied in a potentially
unstable way.

The criterion includes a researcher-specified weight that determines the
relative importance of structural accuracy and moment-enforcement stability. A
larger weight on the structural component prioritizes reducing finite-sample
distortion in the estimator of the structural parameter, whereas a larger weight
on the stability component prioritizes stable Lagrange multipliers, less
concentrated implied weights, and distance from the feasible-probability
boundary. Thus, the selected value of \(\gamma\) can be tailored to the purpose
of the application while preserving the population target and first-order
inference.

The mechanism is as follows. The CRPD estimator enforces moment restrictions by
reweighting the empirical distribution. Starting from uniform probability
weights, it perturbs these weights so that the sample moments satisfy the
model's restrictions, and \(\gamma\) determines how such perturbations are
penalized. Under correct moment specification, this curvature choice does not
change population identification or the leading first-order asymptotic
distribution. Its role is instead higher-order: through the curvature of the
divergence criterion, \(\gamma\) affects the second-order behavior of both the
structural estimator and the associated Lagrange multipliers. These higher-order
effects matter for finite-sample performance because they determine bias,
sensitivity to sampling variation, implied-weight stability, and the cost of
enforcing the moment restrictions (see, e.g., \citet{Rothenberg1984};
\citet{NeweyandSmith2004}). Thus, \(\gamma\) does not change the population
target, but it affects how the estimator behaves around that target in finite
samples.

Our approach is closely related to \citet{NeweyandSmith2004}, who develop
higher-order expansions for GMM and GEL estimators and show that fixed members
of the GEL family can differ in finite samples despite sharing the same
first-order asymptotic distribution. Our contribution differs in two respects.
First, we propose a shift in the interpretation of \(\gamma\): rather than
treating it merely as a fixed index of a prespecified divergence criterion, we
view it as a finite-sample tuning parameter that governs the curvature of the
CRPD criterion. Second, building on this interpretation, we develop a
risk-based selection criterion for the optimal \(\gamma\). In our framework,
higher-order expansions are used to define a pseudo-true curvature target and estimate a risk-optimal value of
\(\gamma\).

The proposed criterion is novel because it evaluates \(\gamma\) at the level of
the full estimator--multiplier system. Specifically, we define a second-order
system-risk criterion that combines two components. The first is a structural
component, which measures finite-sample distortion in the estimate of the
structural parameter relative to a first-order GMM benchmark. The second is a
multiplier-stability component, which measures the finite-sample cost of
enforcing the moment restrictions through the Lagrange multipliers and implied
probability weights. A researcher-specified weight governs the relative
importance of these two components, allowing the selection rule to reflect
whether the application places greater emphasis on structural accuracy,
moment-enforcement stability, or a combination of both. The Lagrange multiplier
expansion is therefore not only an auxiliary step in characterizing the
structural estimator; it becomes central to the selection problem because it
quantifies the stability cost of moment enforcement.

We establish theoretical properties of the selected curvature parameter and of
the resulting CRPD estimator. In particular, we show that the selected estimator
preserves first-order validity under correct moment specification, while
allowing \(\gamma\) to affect higher-order bias, mean squared error, and
coverage distortion. The analysis therefore separates the first-order role of
the moment restrictions from the higher-order role of the divergence curvature.
This distinction is central to our interpretation: \(\gamma\) is not a
structural parameter, but a curvature parameter that governs the finite-sample
behavior of the estimator--multiplier system.

Our framework also shares a similar spirit with robust \(M\)-estimation, while
remaining grounded in a different framework. Classical robust \(M\)-estimators,
such as Huber-type estimators, modify an observation-level loss function to
reduce sensitivity to large residuals or outcome outliers
\citep{Huber1964,HuberandRonchetti2009}. CRPD/GEL, by contrast, is a
moment-based divergence estimator. Its robustness concern is therefore not
residual-level robustness in the classical Huber sense, but finite-sample
stability in the enforcement of moment restrictions. Nevertheless, both
perspectives emphasize that tuning the curvature of the estimating criterion can
matter for finite-sample performance. In robust \(M\)-estimation, tuning
constants govern the curvature of the residual loss. In CRPD, \(\gamma\) serves
as a finite-sample tuning parameter that governs the curvature of the divergence
criterion and thereby affects implied probability weights, Lagrange
multipliers, and sensitivity to moment discrepancies.

Our work is also related to \citet{Schennach2007}, who develops the
exponentially tilted empirical likelihood framework within a divergence-based
class indexed by a power parameter. While our setting differs, we share the
perspective that the power parameter can be viewed as a continuum-valued index
linking different estimators. The key distinction is that
\citet{Schennach2007} focuses on misspecified moment conditions, whereas we
maintain correct specification and show that the power parameter matters
through second-order effects. In our framework, \(\gamma\) shapes finite-sample
sensitivity and stability even when fixed-\(\gamma\) estimators are
first-order asymptotically equivalent.

Monte Carlo simulations show that the selected CRPD estimator remains
approximately centered around the structural parameter while improving
finite-sample stability. In particular, the risk-selected estimator discourages
unstable multipliers, concentrated implied probability weights, and proximity
to the feasible-probability boundary. An empirical illustration using
\citet{Owen2001}'s dairy-cow data further shows that similar point estimates can correspond
to substantially different implied weighting schemes. This finding indicates
that the choice of \(\gamma\) affects not only the structural estimate itself,
but also how the empirical distribution is reweighted to enforce the moment
restrictions. The application highlights the practical value of \(\gamma\) as a
curvature parameter: it can reveal differences in multiplier stability,
implied-weight concentration, and boundary proximity that may be obscured by
point estimates alone.

The rest of the paper is organized as follows. In Section 2, we illustrate how
the power parameter shapes the geometry of the CRPD criterion. In Section 3, we
develop the moment-based CRPD estimator and show that both the structural
parameter estimator and the associated Lagrange multipliers depend on the power
parameter at the second order, thereby affecting finite-sample properties. In
Section 4, we propose a novel finite-sample risk criterion for selecting the optimal
\(\gamma\) and develop the corresponding second-order theory. In Section 5, we
present Monte Carlo simulations. In Section 6, we provide an empirical
application. In Section 7, we conclude.

\section{The power parameter and the loss geometry}
Let $\boldsymbol{\pi} = (\pi_1, \ldots, \pi_n)$ denote the model-implied probability vector, where $\pi_i > 0$ for all $i = 1, \ldots, n$ and $\sum_{i=1}^n \pi_i = 1$.
Let $\mathbf{q} = (q_1, \ldots, q_n)$ denote a reference probability vector against which $\boldsymbol{\pi}$ is compared. For $\gamma\in\mathbb{R}\setminus\{-1,0\}$, the CRPD criterion \citep{CressieandRead1984} is defined as
\[
\mathscr I_\gamma(\boldsymbol{\pi},\mathbf{q})
=
\frac{1}{\gamma(\gamma+1)}
\sum_{i=1}^n \pi_i\left[\left(\frac{\pi_i}{q_i}\right)^\gamma-1\right].
\]

Observe that $\gamma$ governs the curvature of the divergence penalty and therefore shapes how the estimator trades off fidelity to the reference measure against satisfaction of the moment conditions. For large positive values of $\gamma$, the divergence grows rapidly for small deviations, tightly constraining the optimized weights $\pi_i$ around the reference distribution unless the sample moments strongly compel otherwise. As $\gamma$ approaches zero, the divergence converges to the Kullback--Leibler form, allowing moderate deviations at relatively low cost. When $\gamma$ is negative, the loss function becomes flatter around the center of the distribution while penalizing extreme distortions more sharply, effectively down-weighting observations that generate large moment residuals.

These changes in curvature have direct implications for estimation. In particular, the choice of $\gamma$ determines how the estimator responds to empirical features such as heavy tails, skewness, or localized moment violations. Smaller or negative values of $\gamma$ induce robustness by limiting the influence of extreme observations, while larger values enforce tighter adherence to the reference distribution. Consequently, different values of $\gamma$ correspond to different second-order sensitivities of the estimator to local misspecification.

In particular, as the reference distribution, we adopt the \textit{uniform distribution over the observed data points.} That is, we take $\mathbf{q}:=\mathbf{i}/n$ because by the Glivenko--Cantelli theorem, the empirical distribution—which assigns equal mass to each realized sample point—converges uniformly to the true data-generating distribution. Hence, uniform weighting over the sample points
(not over the population support) provides a consistent estimator of the true distribution.

The CRPD between $\boldsymbol{\pi}$ and $\mathbf{i}/n$ is defined as
\begin{equation}
\label{eq:CRPD}
\mathscr{I}_{\gamma}(\boldsymbol{\pi},\mathbf{i}/n)
=
\frac{1}{\gamma(\gamma+1)}
\sum_{i=1}^{n}
\pi_{i}\left[
\left(n\pi_i\right)^{\gamma}-1
\right],
\qquad \gamma \neq 0,-1,
\end{equation}
with continuous extensions at $\gamma=0$ and $\gamma=-1$. Up to an additive constant, $\mathscr{I}_{\gamma}(\boldsymbol{\pi},\mathbf{i}/n)$ admits the equivalent representations
\begin{equation*}
    \begin{cases}
        \mathscr{I}_{\gamma}(\boldsymbol{\pi},\mathbf{i}/n)\big|_{\gamma=0}
        =
        \sum\limits_{i=1}^{n}\pi_i \ln(\pi_i), \\[0.4em]
        \mathscr{I}_{\gamma}(\boldsymbol{\pi},\mathbf{i}/n)\big|_{\gamma=-1}
        =
        -\sum\limits_{i=1}^{n}\ln(\pi_i),
    \end{cases}
\end{equation*}
where the case $\gamma=0$ corresponds to the exponential tilting (ET) divergence, and the case $\gamma=-1$ corresponds to the empirical likelihood (EL) divergence.

Our objective is to formalize the interpretation of $\gamma$ as a parameter governing the geometry of the learning objective, rather than as a fixed modeling choice imposed ex ante. Establishing existence and uniqueness of $\gamma$ within the CRPD family is therefore essential to ensure that the learning objective is well defined and non-arbitrary.

A defining feature of the CRPD framework is that different values of $\gamma$ correspond to distinct loss functions, each characterized by a different curvature and sensitivity to deviations between the empirical distribution and the reference distribution. In this sense, $\gamma$ indexes a family of admissible learning objectives, with each value implying a different robustness profile, rather than acting as a tuning constant chosen independently of the estimation problem.

Observe that $\gamma$ determines the form of the implied loss function. For instance, when the underlying measurements are approximately Gaussian, the squared-error ($\ell_2$) loss is optimal in the sense of maximum likelihood. Within the CRPD family, this case corresponds to $\gamma = 1$, since
\begin{equation*}
    \mathscr{I}_{\gamma}(\pi,\mathbf{i}/n)\big|_{\gamma=1}
    =
    \frac{1}{2}\sum_{i=1}^{n}\pi_i\,[n\pi_i-1]
    =
    \frac{1}{2}\Bigg[
        n\sum_{i=1}^{n}\pi_i^2
        -
        \underbrace{\sum_{i=1}^{n}\pi_i}_{=1}
    \Bigg]
    =
    \frac{1}{2}\Bigg[
        n\sum_{i=1}^{n}\Big(\pi_i-\frac{1}{n}\Big)^2
    \Bigg],
\end{equation*}
which coincides with the squared Euclidean distance between the empirical weights and the uniform distribution. Other values of $\gamma$ induce alternative loss geometries, emphasizing different aspects of deviation and robustness to misspecification as summarized in Table \ref{tab:crpd_gamma_summary_corrected}.

\begin{table}[hbtp!]
\centering
\caption{CRPD loss geometry and special cases}
\label{tab:crpd_gamma_summary_corrected}
\resizebox{\textwidth}{!}{
\begin{tabular}{c c c p{6.5cm}}
\hline
$\gamma$ value 
& Core loss term 
& Geometry 
& Interpretation \\ 
\hline
$\gamma = -1$ (EL)
& $-\sum_i \log(n\pi_i)$
& Log barrier
& Empirical likelihood. The objective diverges as $\pi_i \downarrow 0$, strongly discouraging vanishing weights and enforcing interior solutions. \\

$\gamma = 0$ (ET)
& $\sum_i \pi_i \log(n\pi_i)$
& Logarithmic
& Exponential tilting. Corresponds to the Kullback--Leibler divergence from the uniform distribution; weights are smoothly reweighted and remain close to $1/n$ when the moment conditions are nearly satisfied. \\[0.6em]

$\gamma = 1$ (Pearson $\chi^2$)
& $\sum_i \pi_i^2$
& Quadratic
& Pearson $\chi^2$ divergence. Equivalent to the squared Euclidean distance from the uniform distribution on the probability simplex; in moment space it corresponds to a quadratic (Hotelling-type) distance. \\[0.6em]

\hline
\end{tabular}}
\end{table}

Interpreting the CRPD objective as a loss function clarifies why $\gamma$ should be regarded as a legitimate parameter of the learning problem rather than a fixed modeling convention. While conventional applications treat $\gamma$ as exogenously specified, such a practice implicitly fixes the curvature of the learning objective independently of the estimation problem. In contrast, our framework allows the curvature parameter $\gamma$ to be treated as an object of analysis, whose role can be formally characterized.



\section{Moment-based CRPD estimation}
We consider CRPD estimation under moment conditions.  Let $\{Z_i\}_{i=1}^n$ be an i.i.d.\ sample drawn from a probability distribution $P$
on $(\mathcal X,\mathcal A)$, where $Z_i=(Y_i,X_i)$ collects the observed outcome and covariates for unit $i$.
Let $Z=(Y,X)$ denote a generic draw from $P$.
Let $g:\mathcal X\times\Theta\to\mathbb R^q$ be a measurable moment function.
We assume that the true parameter $\theta_0\in\Theta$ satisfies the population moment condition
$
E\bigl[g(Z,\theta_0)\bigr]=0.
$ Let $g_i(\theta):=g(Z_i,\theta)\in\mathbb{R}^q$ and define
$
\bar g(\theta):=\frac{1}{n}\sum_{i=1}^n g_i(\theta),
\;
\hat\Omega(\theta):=\frac{1}{n}\sum_{i=1}^n g_i(\theta)g_i(\theta)^{\prime}.
$
Let $G_i(\theta):=\partial g_i(\theta)/\partial\theta^{\prime}$, and define
$
\bar G(\theta):=\frac{1}{n}\sum_{i=1}^n G_i(\theta),
\;
G_0:=E[G(Z,\theta_0)],
\;
\Omega_0:=E[g(Z,\theta_0)g(Z,\theta_0)^{\prime}].
$
Let $H_i(\theta)$ denote the collection of second derivatives of $g_i(\theta)$ w.r.t.\ $\theta$
(i.e.\ a $q\times p\times p$ tensor), and assume $E[\sup_{\theta\in\Theta}\|H(Z,\theta)\|]<\infty$. \vspace{2mm}

\subsection{Moment-based CRPD estimator}
Fix $\gamma\in\Gamma$ with $\gamma\neq 0$ and $\gamma\neq -1$. Define the CRPD estimator:
\begin{equation*}
    (\hat{\theta}_{\gamma},\hat{\boldsymbol{\pi}}_{\gamma})=\arg\min_{\theta,\boldsymbol{\pi}}\mathscr{I}_{\gamma}(\boldsymbol{\pi},\mathbf{i}/n) \quad \text{such that }\sum_{i=1}^{n}\pi_{i}=1, \; \sum_{i=1}^{n}\pi_{i}g_{i}(\theta)=0, \; \pi_{i}\geq0.
\end{equation*}
The Lagrange function is
\begin{equation}
\begin{split}
    \mathcal{L}(\boldsymbol{\pi},\theta,\lambda,\delta)
    &=\mathscr{I}_{\gamma}(\boldsymbol{\pi},\mathbf{i}/n)+\lambda^{\prime}\sum_{i=1}^{n}\pi_{i}g_{i}(\theta)+\delta\Big(\sum_{i=1}^{n}\pi_{i}-1\Big)\\
    &=\frac{1}{\gamma(\gamma+1)}\sum\limits_{i=1}^{n}\pi_{i}\left[(n\pi_{i})^{\gamma}-1\right]+\lambda^{\prime}\sum\limits_{i=1}^{n}\pi_{i}g_{i}+\delta\left(\sum\limits_{i=1}^{n}\pi_{i}-1\right). \label{obj_Lagrangian}
\end{split}
\end{equation}

\noindent The first derivative of \eqref{obj_Lagrangian} with respect to $\pi_{i}$ is
\begin{equation}\label{foc_pi_i}
\begin{split}
    \frac{\partial \mathcal{L}}{\partial \pi_{i}}
    &=\frac{(1+\gamma)(n\pi_{i})^{\gamma}-1}{\gamma(\gamma+1)}+\lambda^{\prime} g_{i}(\theta)+\delta \\
    &=\frac{1}{\gamma}(n\pi_{i})^{\gamma}-\frac{1}{\gamma(\gamma+1)}+\lambda^{\prime}g_{i}(\theta)+\delta.
\end{split}
\end{equation}
\noindent The first-order condition (FOC) is obtained by setting \eqref{foc_pi_i} equal to zero. Observe
\begin{equation*}
    \hspace{7mm}0=\frac{1}{\gamma}(n\pi_{i})^{\gamma}-\frac{1}{\gamma(\gamma+1)}+\delta+\lambda^{\prime}g_{i}(\theta).
\end{equation*}
Rearranging terms,
\begin{equation*}
    \frac{1}{\gamma}(n\pi_{i})^{\gamma}=\frac{1}{\gamma(\gamma+1)}-\delta-\lambda^{\prime}g_{i}(\theta).
\end{equation*}
Multiplying $\gamma$,
\begin{equation*}
    (n\pi_{i})^{\gamma}=\frac{1}{\gamma+1}-\gamma\delta-\gamma\lambda^{\prime}g_{i}(\theta).
\end{equation*}
Thus,
\begin{equation}\label{pi_i}
    \pi_{i}=\frac{1}{n}\left(\frac{1}{\gamma+1}-\gamma\delta-\gamma\lambda^{\prime}g_{i}(\theta)\right)^{1/\gamma}.
\end{equation} \vspace{2mm}

\noindent
Before proceeding to the asymptotic analysis, it is useful to characterize the
population solution of the CRPD problem under correct specification. The
following lemma establishes the benchmark values of the parameters and
multipliers around which the subsequent expansions are developed.

\begin{lemma}\label{lem:population_solution}(Population solution under correct specification)
Under correct specification, the solution to the CRPD problem lies in a neighborhood of
\[
\theta=\theta_{0}, 
\qquad 
\lambda=0, 
\qquad 
\pi_{i}=1/n, 
\qquad 
\delta_{0}=-\frac{1}{\gamma+1}.
\]
\end{lemma}
\noindent \textbf{Proof.} See Appendix \ref{subsec:pf_lem-population-solution}. \vspace{2mm}

\noindent
Lemma \ref{lem:population_solution} identifies the population reference point for the CRPD
optimization problem. Under correct specification, the population solution
corresponds to uniform weights $\pi_i = 1/n$ with $\lambda = 0$ and
$\delta_0 = -1/(\gamma + 1)$. This characterization provides the baseline
around which the sample estimators and Lagrange multipliers are expanded in the
subsequent analysis.

\subsection{Asymptotic analysis and finite sample properties}
For each $\theta$, the pair $(\lambda^{\prime},\delta)$ is determined by the two constraints
\begin{equation}\label{eq:constraints}
\sum_{i=1}^n \pi_i(\theta,\lambda,\delta)=1,
\qquad
\sum_{i=1}^n \pi_i(\theta,\lambda,\delta)\,g_i(\theta)=0.
\end{equation}
Define the stacked system
\[
\Psi_n(\theta,\lambda,\delta)
:=
\begin{pmatrix}
\Psi_{n,1}(\theta,\lambda,\delta)\\
\Psi_{n,2}(\theta,\lambda,\delta)
\end{pmatrix}
=
\begin{pmatrix}
\frac{1}{n}\sum_{i=1}^n w_i(\theta,\lambda,\delta)-1\\[4pt]
\frac{1}{n}\sum_{i=1}^n w_i(\theta,\lambda,\delta)\,g_i(\theta)
\end{pmatrix},
\]
where $w_i(\theta,\lambda,\delta):=
\Big(\frac{1}{\gamma+1}-\gamma\delta-\gamma\lambda^{\prime}g_i(\theta)\Big)^{1/\gamma}$, so that \eqref{eq:constraints} is equivalent to $\Psi_n(\theta,\lambda,\delta)=0$.

For each $\theta$, let $(\hat\lambda(\theta),\hat\delta(\theta))$ denote a solution of
$\Psi_n(\theta,\lambda,\delta)=0$ (existence/uniqueness shown below) and define
$\hat{\boldsymbol{\pi}}(\theta)=\boldsymbol{\pi}(\theta,\hat\lambda(\theta),\hat\delta(\theta))$.
Define the profiled objective
\[
L_n(\theta):=\mathscr{I}_{\gamma}\big(\hat{\boldsymbol{\pi}}(\theta),\textbf{i}/n\big),
\qquad
\hat\theta\in\arg\min_{\theta\in\Theta}L_n(\theta).
\]

We assume the followings.

\begin{assumption}\label{assumption:theta_compact}
    (Compactness) $\Theta$ is compact and $\theta_0\in\mathrm{int}(\Theta)$.
\end{assumption} 

\begin{assumption}\label{assumption:smooth_envolope}
    (Smoothness and envelope) $g(z,\theta)$ is continuously differentiable in $\theta$
for $P$-a.e.\ $z$, with Jacobian $G(z,\theta):=\partial g(z,\theta)/\partial\theta^{\prime}$.
Moreover,
\[
E\Big[\sup_{\theta\in\Theta}\|g(Z,\theta)\|^2\Big]<\infty,
\qquad
E\Big[\sup_{\theta\in\Theta}\|G(Z,\theta)\|\Big]<\infty.
\]
\end{assumption}

\begin{assumption}\label{assumption:identification}
(Identification) $E[g(Z,\theta_0)]=0$ and $E[g(Z,\theta)]\neq 0$ for $\theta\neq\theta_0$.
Let
\[
G_0:=E[G(Z,\theta_0)],\qquad
\Omega_0:=E[g(Z,\theta_0)g(Z,\theta_0)^{\prime}],
\]
and assume $G_0$ has full column rank and $\Omega_0$ is positive definite.
\end{assumption}

\begin{assumption}\label{assumption:interior_feasibility} (Interior feasibility / positivity) There exists a neighborhood $\mathcal{N}$ of $\theta_0$
and a constant $\kappa>0$ such that, with probability approaching one, for each $\theta\in\mathcal{N}$ there exist
$(\lambda,\delta)$ with $\|\lambda\|+\|\delta-\delta_0\|\le \kappa$ satisfying
\[
\frac{1}{\gamma+1}-\gamma\delta-\gamma\lambda^{\prime}g_i(\theta)\ge \kappa
\quad\text{for all } i,
\]
and the constraints \eqref{eq:constraints}.
\end{assumption}

Assumption~\ref{assumption:theta_compact} ensures that the parameter space is bounded and closed, which facilitates uniform convergence arguments and guarantees the existence of minimizers of the profiled criterion. The interior condition $\theta_0\in\mathrm{int}(\Theta)$ rules out boundary issues and permits local expansions around the true parameter. Assumption~\ref{assumption:smooth_envolope} provides the regularity required for both uniform laws of large numbers and Taylor expansions. Continuous differentiability of $g(z,\theta)$ in $\theta$ ensures that the multiplier system and the profiled objective are smooth functions of $\theta$. The envelope conditions guarantee integrable bounds that allow uniform convergence of sample moments and control stochastic equicontinuity. Assumption~\ref{assumption:identification} ensures that the population moment condition $E[g(Z,\theta)]=0$ uniquely characterizes $\theta_0$. The full column rank of $G_0$ guarantees local identification and invertibility of the Jacobian in the multiplier system, while the positive definiteness of $\Omega_0$ ensures nondegeneracy of the moment covariance matrix. Together, these conditions imply both consistency and asymptotic normality of the estimator. Assumption~\ref{assumption:interior_feasibility} is a technical regularity condition ensuring that the implied weights are well-defined and that the multiplier system remains smooth in a neighborhood of $\theta_0$. The uniform positivity bound prevents the implied probabilities from approaching the boundary of the simplex. This guarantees interior solutions, avoids singularities in the CRPD power transformation, and ensures that the mapping $(\theta,\lambda,\delta)\mapsto w_i(\theta,\lambda,\delta)$ is continuously differentiable on the relevant region.

\begin{theorem}\label{thm:uniform_conv_profiled}(Uniform convergence of the profiled objective)
Fix the hyperparameter $\gamma$. Under Assumptions 1--4,
there exists a deterministic function $L(\theta)$ such that
\[
\sup_{\theta\in\Theta}\bigl|L_n(\theta)-L(\theta)\bigr|
\xrightarrow{p}0,
\]
and $L(\theta)$ is uniquely minimized at $\theta_0$.
\end{theorem}
\noindent \textbf{Proof.} See Appendix \ref{subsec:pf_thm-uniform-conv-profiled}.

\noindent
Theorem \ref{thm:uniform_conv_profiled} shows that the profiled sample objective $L_n(\theta)$ converges uniformly to a deterministic population criterion $L(\theta)$. This result ensures that the finite-sample objective approximates a well-defined population problem. Because $L(\theta)$ is uniquely minimized at $\theta_0$, the minimizer of the sample objective will concentrate near the true parameter as the sample size increases. This property forms the basis for establishing consistency of the CRPD estimator in the subsequent results.

To analyze the profiled criterion $L_n(\theta)$, we must first ensure that the 
Lagrange multipliers solving the constraint system are well defined. 
In particular, we require that for $\theta$ in a neighborhood of the true value 
$\theta_0$, the system of first-order conditions admits a unique solution for 
the multipliers. The following lemma establishes the local existence and 
uniqueness of $(\hat\lambda(\theta),\hat\delta(\theta))$.
$(\hat\lambda(\theta),\hat\delta(\theta))$.
\begin{lemma}[Local existence and uniqueness of multipliers]
\label{lem:exist_unique}
Under Assumptions \ref{assumption:theta_compact},\ref{assumption:smooth_envolope},\ref{assumption:identification},\ref{assumption:interior_feasibility}, there exists a neighborhood $\mathcal{N}$ of $\theta_0$ such that,
with probability approaching one, for every $\theta\in\mathcal{N}$ the system
$\Psi_n(\theta,\lambda,\delta)=0$ has a unique solution $(\hat\lambda(\theta),\hat\delta(\theta))$
satisfying $\|\hat\lambda(\theta)\|+\|\hat\delta(\theta)-\delta_0\|\le \kappa$.
\end{lemma}
\noindent \textbf{Proof.} See Appendix \ref{subsec:pf_lem-exist-unique}.\vspace{2mm}

\noindent We can now characterize 
the large-sample behavior of the estimator obtained from minimizing $L_n(\theta)$. 
The following theorem establishes the consistency of $\hat{\theta}$.
\begin{theorem}[Consistency of $\hat\theta$]
\label{thm:theta_consistency}
Under Assumptions \ref{assumption:theta_compact},\ref{assumption:smooth_envolope},\ref{assumption:identification}, and Theorem \ref{thm:uniform_conv_profiled}, any measurable selection
$\hat\theta\in\arg\min_{\theta\in\Theta}L_n(\theta)$ satisfies
$
\hat\theta\xrightarrow{p}\theta_0.
$
\end{theorem}
\noindent \textbf{Proof.} See Appendix \ref{subsec:pf_thm-theta-consistency}. \vspace{2mm}

\noindent The consistency of $\hat{\theta}$ further implies the consistency of the associated 
Lagrange multipliers. The following corollary formalizes this result.

\begin{corollary}[Consistency of multipliers]\label{cor:mult_consistency}
Suppose $\hat{\theta} \xrightarrow{p} \theta_0$. 
For each $\theta$ in a neighborhood of $\theta_0$, let $(\lambda(\theta),\delta(\theta))$ denote the unique solution to
$
\Psi(\theta,\lambda,\delta)=0,
$
and assume that the mapping
$
\theta \mapsto (\lambda(\theta),\delta(\theta))
$
is continuous at $\theta_0$. Define
$
(\hat{\lambda},\hat{\delta})
:=
(\lambda(\hat{\theta}),\delta(\hat{\theta})).
$
Then
$
(\hat{\lambda},\hat{\delta})
\xrightarrow{p}
(\lambda_0,\delta_0),
$
where $(\lambda_0,\delta_0) = (\lambda(\theta_0),\delta(\theta_0))$.
\end{corollary}
\noindent \textbf{Proof.} See Appendix \ref{subsec:pf_cor-mult-consistency} \vspace{2mm}

\noindent
Having established the consistency of the estimator and the associated multipliers, we next characterize their asymptotic behavior. The following result provides the first-order expansion of $\hat{\lambda}$ and the convergence rate of $\hat{\delta}$.
\begin{theorem}[First-order expansion for $\hat\lambda(\hat\theta)$ and rate for $\hat\delta(\hat\theta)$]
\label{thm:lambda_delta_rates}
Under Assumptions \ref{assumption:theta_compact},\ref{assumption:smooth_envolope},\ref{assumption:identification},\ref{assumption:interior_feasibility} and Theorem~\ref{thm:theta_consistency}, let
$\hat\lambda:=\hat\lambda(\hat\theta)$ and $\hat\delta:=\hat\delta(\hat\theta)$.
Then, under correct specification,
\[
\hat\lambda
=
\hat\Omega(\theta_0)^{-1}\bar g(\theta_0)
+o_p(n^{-1/2}),
\qquad
\hat\delta-\delta_0
=
O_p(n^{-1}),
\]
and
\[
\sqrt n\,\hat\lambda \ \xrightarrow{d}\ N(0,\Omega_0^{-1}).
\]
\end{theorem}
\noindent \textbf{Proof.} See Appendix \ref{subsec:pf_thm-lambda-delta-rates}.\vspace{2mm}

\noindent
While Theorem \ref{thm:lambda_delta_rates} establishes the convergence rate of $\hat{\delta}$, its limiting distribution requires a finer expansion. The following theorem characterizes the nondegenerate limit of the probability multiplier.
\begin{theorem}[Nondegenerate limit for the probability multiplier]
\label{thm:delta_limit}
Assume the conditions of Theorem~\ref{thm:lambda_delta_rates}. Suppose additionally that a third-order
remainder control holds so that the quadratic approximation in $t_i(\theta)$ is valid at the $n^{-1}$
scale. Then
\[
n(\hat\delta-\delta_0)
=
-\frac{\gamma+1}{2}\,\Big(\sqrt n\,\bar g(\theta_0)\Big)^{\prime}\Omega_0^{-1}\Big(\sqrt n\,\bar g(\theta_0)\Big)
+o_p(1),\]
and hence
\[
n(\hat\delta-\delta_0)\ \xrightarrow{d}\ -\frac{\gamma+1}{2}\,\chi^{2}_{q},
\]
where $\chi^{2}_{q}$ denotes a chi-square random variable with $q$ degrees of freedom.
\end{theorem}
\noindent \textbf{Proof.} See Appendix \ref{subsec:pf_thm-delta-limit}. \vspace{2mm}

\noindent
With the expansions for the multipliers established, we now derive the asymptotic distribution of the estimator.
\begin{theorem}[Asymptotic normality of $\hat\theta$]
\label{thm:theta_CLT}
Under Assumptions \ref{assumption:theta_compact},\ref{assumption:smooth_envolope},\ref{assumption:identification},\ref{assumption:interior_feasibility} and Theorem~\ref{thm:theta_consistency}. Then
\[
\sqrt n\,(\hat\theta-\theta_0)
=
-(G_0^{\prime}\Omega_0^{-1}G_0)^{-1}G_0^{\prime}\Omega_0^{-1}\sqrt n\,\bar g(\theta_0)
+o_p(1),
\]
and hence
\[
\sqrt n\,(\hat\theta-\theta_0)\ \xrightarrow{d}\ N\!\left(0,\ (G_0^{\prime}\Omega_0^{-1}G_0)^{-1}\right).
\]
\end{theorem}
\noindent \textbf{Proof.} See Appendix \ref{subsec:pf_thm-theta-CLT}. \vspace{2mm}

\noindent
To facilitate the second-order analysis, we first derive a second-order expansion of the moment equation implied by the CRPD constraints.

\begin{lemma}[Second-order expansion of the moment equation]
\label{lem:moment_eq_second}
Let $(\hat\theta,\hat\lambda,\hat\delta)$ satisfy the CRPD constraints:
\[
\frac{1}{n}\sum_{i=1}^n w_i(\hat\theta,\hat\lambda,\hat\delta)=1,
\qquad
\frac{1}{n}\sum_{i=1}^n w_i(\hat\theta,\hat\lambda,\hat\delta)\,g_i(\hat\theta)=0.
\]
Then, uniformly on events of probability approaching one,
\begin{align}
0
&=
\bar g(\hat\theta)
-\hat\Omega(\hat\theta)\hat\lambda
+\frac{1-\gamma}{2}\cdot \frac{1}{n}\sum_{i=1}^n \bigl(t_i(\hat\theta,\hat\lambda,\hat\delta)\bigr)^2\,g_i(\hat\theta)
+o_p(n^{-1}). \label{eq:moment_second}
\end{align}

\end{lemma}
\noindent \textbf{Proof.} See Appendix \ref{subsec:pf_lem-moment-eq-second}. \vspace{2mm}

\noindent
Using the second-order expansion of the moment equation in Lemma \ref{lem:moment_eq_second}, we can now refine the first-order result for the multiplier. The following theorem provides the second-order expansion of $\hat{\lambda}$.

\begin{theorem}[Second-order expansion for $\hat\lambda$]
\label{thm:lambda_second}
Under the conditions of Lemma~\ref{lem:moment_eq_second} and the first-order expansion
\[
\hat\lambda
=
\hat\Omega(\theta_0)^{-1}\bar g(\theta_0)
+
o_p(n^{-1/2}),
\]
we have
\begin{equation}\label{eq:lambda_second}
\hat\lambda
=
\hat\Omega(\theta_0)^{-1}\bar g(\theta_0)
\;+\;
\frac{1}{n}\,B_{\lambda,n}(\gamma)
\;+\;
o_p(n^{-1}),
\end{equation}
where the sample-dependent second-order correction is
\begin{equation}\label{eq:Blambda}
B_{\lambda,n}(\gamma)
:=
\frac{1-\gamma}{2}\,
\hat\Omega(\theta_0)^{-1}
\Bigg[
\frac{1}{n}\sum_{i=1}^n
\Bigl(
(\hat\Omega(\theta_0)^{-1}\sqrt n\,\bar g(\theta_0))^{\prime}
g_i(\theta_0)
\Bigr)^2
g_i(\theta_0)
\Bigg],
\end{equation}
and $B_{\lambda,n}(\gamma)=O_p(1)$.
\end{theorem}
\noindent \textbf{Proof.} See Appendix \ref{subsec:pf_thm-lambda-second}. \vspace{2mm}

\noindent
Finally, combining the first-order asymptotic representation of $\hat{\theta}$ with the second-order expansion of the multiplier $\hat{\lambda}$, we can now derive the second-order expansion of the CRPD estimator.

\begin{theorem}\label{thm:theta_second}(Second-order expansion for $\hat\theta$)
Assume the conditions of Theorems~\ref{thm:theta_CLT} and \ref{thm:lambda_second}. 
Then the CRPD estimator $\hat{\theta}$ admits the expansion
\begin{equation}\label{eq:theta_second}
\hat{\theta}-\theta_0
=
\big(G_0^{\prime}\Omega_0^{-1}G_0\big)^{-1}
G_0^{\prime}\Omega_0^{-1}\bar g(\hat{\theta})
+
\frac{1}{n}B_{\theta,n}(\gamma)
+
o_p(n^{-1}),
\end{equation}
where $\bar g(\theta)=\frac{1}{n}\sum_{i=1}^n g_i(\theta)$ and $
B_{\theta,n}(\gamma)
=
\big(G_0^{\prime}\Omega_0^{-1}G_0\big)^{-1}\Delta_n(\gamma)=O_p(1)$, with
$\Delta_n(\gamma)=O_p(1)$ collects all order-$n^{-1}$ terms arising from
the second-order expansion of the Lagrange multiplier,
the curvature term involving $\bar H(\theta_0)$, and plug-in effects
from replacing sample matrices with their population limits.
\end{theorem}
\noindent \textbf{Proof.} See Appendix \ref{subsec:pf_thm-theta-second}. \vspace{2mm}

\noindent
Our asymptotic analysis characterizes the large-sample behavior of the CRPD estimator and its associated multipliers. The estimator $\hat{\theta}$ is consistent and asymptotically normal, sharing the same first-order limit as standard moment-based estimators. The power parameter $\gamma$ does not affect the leading asymptotic distribution, but it appears in the order-$n^{-1}$ bias term arising from the second-order expansions of the multipliers and the estimator. Consequently, the choice of $\gamma$ influences finite-sample performance while leaving first-order asymptotic efficiency unchanged.

\section{Risk-optimal selection of the CRPD power parameter}
\label{subsec:optimal_gamma}

The previous section shows that, under correct moment specification, the CRPD power
parameter $\gamma$ does not affect the first-order asymptotic distribution of the
structural estimator. For each fixed $\gamma$, the structural estimator and the
associated moment Lagrange multiplier admit the expansions in
\eqref{eq:theta_second} and \eqref{eq:lambda_second}, respectively. The leading
terms in both expansions are invariant to $\gamma$. Therefore, different values of
$\gamma$ generate first-order equivalent estimators. In this sense, $\gamma$ is not
identified by the first-order moment problem. Its role is instead higher-order:
$\gamma$ changes the order-$n^{-1}$ correction terms
$B_{\theta,n}(\gamma)$ and $B_{\lambda,n}(\gamma)$, and thereby affects
finite-sample bias, stability, and sensitivity.

This observation suggests that $\gamma$ should not be treated as a structural
parameter. The population moment condition identifies $\theta_{0}$, not $\gamma$.
Instead, $\gamma$ should be interpreted as a curvature parameter that governs the
finite-sample geometry of the CRPD problem. Hence, the relevant question is not
which $\gamma$ is correct in a structural sense, but which $\gamma$ gives a more
stable finite-sample estimating system. To emphasize this dependence, we write the
structural estimator and the moment multiplier as $\hat\theta_{\gamma}$ and
$\hat\lambda_{\gamma}$.

For a fixed $\gamma$, the CRPD solution consists not only of the structural
estimator $\hat\theta_{\gamma}$, but also of the moment multiplier
$\hat\lambda_{\gamma}$ and the implied empirical weights
$\left\{\hat\pi_{i,\gamma}\right\}_{i=1}^{n}$. The multiplier is important because
it measures the shadow cost of imposing the sample moment restrictions. In
particular, the implied weights satisfy
\begin{equation}
\hat\pi_{i,\gamma}
=
\frac{1}{n}
\left[
\frac{1}{\gamma+1}
-
\gamma \hat\delta_{\gamma}
-
\gamma \hat\lambda_{\gamma}^{\prime}
g_{i}(\hat\theta_{\gamma})
\right]^{1/\gamma}.
\label{eq:pi_gamma_weight}
\end{equation}
Thus, the scalar term
$\hat\lambda_{\gamma}^{\prime}g_{i}(\hat\theta_{\gamma})$ determines how much
observation $i$ is reweighted in order to satisfy the moment restrictions. If $\hat\lambda_{\gamma}$ is large or unstable, the bracketed term in
\eqref{eq:pi_gamma_weight} may move close to the boundary of the admissible
region for the implied probabilities. This boundary-related instability is
reflected in the concentration of the implied weights on a small number of
observations. In such cases, the estimator effectively relies on a limited subset
of the sample to satisfy the moment restrictions.

This loss of balance can be summarized by the effective sample size. In the CRPD
formulation, the implied probabilities $\{\hat\pi_i\}_{i=1}^{n}$ define a
reweighted empirical distribution. When these probabilities are close to uniform,
the estimator draws information from many observations. When they are highly
uneven, a small number of observations receive disproportionate influence. The
effective sample size,
\[
    ESS
    =
    \frac{1}{\sum_{i=1}^{n}\hat\pi_i^2},
\]
measures this concentration. If $\hat\pi_i=1/n$ for all $i$, then $ESS=n$. If the weights are concentrated on a smaller subset of observations, then
$ESS<n$, and the estimator behaves as if it were based on a smaller effective
sample. This loss of effective sample size creates both statistical and numerical
instability.

Statistically, a highly concentrated implied distribution makes the estimator
sensitive to a small number of observations. Small perturbations in those
observations, or in the corresponding sample moments, can lead to non-negligible
changes in $\hat\theta_{\gamma}$, $\hat\lambda_{\gamma}$, and the implied
probabilities. Hence, the moment restrictions may be satisfied only by assigning
excessive influence to a limited part of the sample. Such a solution is less
stable because it depends strongly on particular observations rather than on the
overall empirical distribution.

Numerically, large multipliers and uneven implied probabilities indicate that the
solution is close to the boundary of the feasible probability region. This
boundary need not appear only through exactly zero weights. In many
implementations, the weights are constrained to remain positive, so literal zero
weights may be rare. Nevertheless, the solution can still be close to the
boundary when some weights are nearly zero or when other weights become
excessively large. In this case, small changes in $\lambda$ can produce large
changes in the implied weights, making the optimization problem more fragile.

Weight concentration can also distort finite-sample inference. First-order
standard errors rely on regularity conditions under which the implied
distribution remains well behaved. If the estimator effectively relies on a much
smaller subset of observations, the finite-sample distribution of
$\hat\theta_{\gamma}$ may be more variable than the first-order approximation
suggests. This helps explain why coverage can be poor in small samples even when
the estimator remains approximately centered around the true parameter.

This multiplier-stability issue provides a direct reason to include the
multiplier component in the selection criterion. The purpose is not to treat
$\hat\lambda_{\gamma}$ as another structural parameter, nor to minimize its mean
squared error. Rather, the multiplier component measures the finite-sample cost
of enforcing the moment restrictions. A value of $\gamma$ may deliver a
reasonable structural estimate $\hat\theta_{\gamma}$ while requiring a large
multiplier adjustment. Such a solution can be fragile because small changes in
the sample moments may lead to large changes in the implied probabilities, either
through near-zero weights or through excessive upper-tail weight concentration.\footnote{A detailed discussion is provided in Appendix
\ref{subsec:boundary_problem}.}

Accordingly, the multiplier component should be interpreted as a measure of
overall moment-enforcement stability. It is related to boundary concerns, but it
is not equivalent to a lower-bound hit rate based only on $\min_i\hat\pi_i$. A
selector may produce fewer near-zero weights and still be less stable overall if
it generates larger multiplier norms, larger maximum weights, or a smaller
effective sample size. The multiplier component therefore helps select a
curvature parameter that balances structural accuracy with stable moment
enforcement in the full estimator--multiplier system.

\subsection{Second-order system-risk criterion}

We define two unnormalized second-order risk components. The first component
measures the second-order distortion of the structural estimator:
\begin{equation}
C_{\theta}^{\mathrm{raw}}(\gamma)
=
E\left[
B_{\theta,n}(\gamma)^{\prime}
K_{0}
B_{\theta,n}(\gamma)
\right],
\qquad
K_{0}=G_{0}^{\prime}\Omega_{0}^{-1}G_{0}.
\label{eq:Ctheta_raw}
\end{equation}
Although $B_{\theta,n}(\gamma)$ is already expressed in the structural-parameter
scale, we evaluate it using $K_{0}$ to obtain a standardized measure of
second-order distortion. The matrix $K_{0}$ is the first-order precision matrix of
$\hat\theta_{\gamma}$. Therefore,
$B_{\theta,n}(\gamma)^{\prime}K_{0}B_{\theta,n}(\gamma)$ measures the
second-order structural distortion in Wald-type units. Equivalently, if
$B_{\theta,n}(\gamma)=K_{0}^{-1}\Delta_{n}(\gamma)$, then
\begin{equation}
B_{\theta,n}(\gamma)^{\prime}K_{0}B_{\theta,n}(\gamma)
=
\Delta_{n}(\gamma)^{\prime}K_{0}^{-1}\Delta_{n}(\gamma).
\end{equation}
Thus, the role of $K_{0}$ is not to rederive the expansion. It is to measure the
second-order correction in units normalized by the first-order precision of
$\hat\theta_{\gamma}$.

The second component measures the second-order cost of enforcing the moment
restrictions:
\begin{equation}
C_{\lambda}^{\mathrm{raw}}(\gamma)
=
E\left[
B_{\lambda,n}(\gamma)^{\prime}
\Omega_{0}
B_{\lambda,n}(\gamma)
\right].
\label{eq:Clambda_raw}
\end{equation}
The use of $\Omega_{0}$ is natural because $\lambda$ is dual to the moment vector
$g$. Since $\Omega_{0}=E[g(\theta_{0})g(\theta_{0})^{\prime}]$, for any fixed
vector $b_{\lambda}$ with the same dimension as the moment vector,
\begin{equation}
b_{\lambda}^{\prime}\Omega_{0}b_{\lambda}
=
E
\left[
\left(
b_{\lambda}^{\prime}g(\theta_{0})
\right)^{2}
\right].
\label{eq:Clambda_interpretation}
\end{equation}
Applying this metric to the second-order multiplier correction gives
\begin{equation}
B_{\lambda,n}(\gamma)^{\prime}
\Omega_{0}
B_{\lambda,n}(\gamma),
\end{equation}
which measures the size of the adjustment pressure induced by
$B_{\lambda,n}(\gamma)$ in the natural scale of the moment restrictions.

Because the two components may have different numerical scales, we normalize them
over the candidate set $\Gamma$ before combining them. For
$j\in{\theta,\lambda}$, define
\begin{equation}
\widetilde C_{j}(\gamma)
=
\frac{
C_{j}^{\mathrm{raw}}(\gamma)
-
\min_{\tilde\gamma\in\Gamma}
C_{j}^{\mathrm{raw}}(\tilde\gamma)
}{
\max_{\tilde\gamma\in\Gamma}
C_{j}^{\mathrm{raw}}(\tilde\gamma)
-
\min_{\tilde\gamma\in\Gamma}
C_{j}^{\mathrm{raw}}(\tilde\gamma)
},
\label{eq:population_minmax_normalization}
\end{equation}
whenever the denominator is positive. If a component is constant over $\Gamma$,
then that component contains no information for selecting $\gamma$ and is set to
zero after normalization. This normalization places the structural and multiplier
components on a common $[0,1]$ scale.

The normalized second-order system-risk criterion is
\begin{equation}
C^{\mathrm{sys}}(\gamma;\tau)
=
\tau \widetilde C_{\theta}(\gamma)
+
(1-\tau)\widetilde C_{\lambda}(\gamma),
\qquad
\tau\in[0,1].
\label{eq:system_risk}
\end{equation}
The parameter $\tau$ controls how the researcher values the two parts of the CRPD
system. When $\tau=1$, the criterion focuses only on the second-order structural
distortion of $\hat\theta_{\gamma}$. When $\tau=0$, the criterion focuses only on
the multiplier-stabilization cost. Intermediate values $0<\tau<1$ balance the two.
Thus, $\tau$ does not change the population moment target. It only specifies how
the finite-sample performance of the CRPD system is evaluated.

We define the risk-optimal curvature parameter as
\begin{equation}
\gamma^{\star}(\tau)
=
\arg\min_{\gamma\in\Gamma}
C^{\mathrm{sys}}(\gamma;\tau).
\label{eq:gamma_star}
\end{equation}
This target is not a structural parameter. Rather, $\gamma^{\star}(\tau)$ is a
\textit{pseudo}-true hyperparameter target: it is the value of $\gamma$ that
minimizes the normalized second-order system risk. In this sense, $\gamma$ is
first-order unidentified in the original CRPD moment problem, but
$\gamma^{\star}(\tau)$ is identified by the second-order risk criterion.

\subsection{Consistency of the selected curvature parameter}

To connect the feasible selector to the population target, define
\begin{equation}
Q(\gamma;\tau)
=
C^{\mathrm{sys}}(\gamma;\tau),
\qquad
\widehat Q_{n}(\gamma;\tau)
=
\widehat C_{n}^{\mathrm{sys}}(\gamma;\tau).
\end{equation}
The population target is
\begin{equation}
\gamma^{\star}(\tau)
=
\arg\min_{\gamma\in\Gamma}
Q(\gamma;\tau),
\end{equation}
and the feasible selector is
\begin{equation}
\hat\gamma(\tau)
\in
\arg\min_{\gamma\in\Gamma}
\widehat Q_{n}(\gamma;\tau).
\end{equation}

\begin{assumption}[Identification of the outer risk target]
\label{ass:gamma_identification}
Fix $\tau\in[0,1]$. The candidate set $\Gamma$ is compact, the population
criterion $Q(\gamma;\tau)$ is continuous in $\gamma$, and
$Q(\gamma;\tau)$ has a unique minimizer
\begin{equation}
\gamma^{\star}(\tau)
=
\arg\min_{\gamma\in\Gamma}
Q(\gamma;\tau).
\end{equation}
\end{assumption}

\begin{assumption}[Uniform convergence of raw risk components]
\label{ass:raw_risk_uniform_convergence}
For $j\in{\theta,\lambda}$,
\begin{equation}
\sup_{\gamma\in\Gamma}
\left|
\widehat C_{j,\mathrm{raw}}(\gamma)
-
C_{j}^{\mathrm{raw}}(\gamma)
\right|
\xrightarrow{p}0.
\label{eq:raw_component_uniform_convergence}
\end{equation}
Moreover, the population ranges are nondegenerate:
\begin{equation}
R_{j}
:=
\max_{\gamma\in\Gamma} C_{j}^{\mathrm{raw}}(\gamma)
-
\min_{\gamma\in\Gamma} C_{j}^{\mathrm{raw}}(\gamma)
>
0,
\qquad
j\in{\theta,\lambda}.
\label{eq:nondegenerate_population_range}
\end{equation}
\end{assumption}

\begin{lemma}[Uniform convergence of the normalized system risk]
\label{lem:system_risk_uniform_convergence}
Under Assumption \ref{ass:raw_risk_uniform_convergence},
\begin{equation}
\sup_{\gamma\in\Gamma}
\left|
\widehat Q_{n}(\gamma;\tau)
-
Q(\gamma;\tau)
\right|
\xrightarrow{p}0.
\label{eq:uniform_convergence_gamma}
\end{equation}
\end{lemma}
\noindent \textbf{Proof.} See Appendix \ref{subsec:pf_uniform-conv-system-risk}.

\begin{theorem}[Consistency of the selected curvature parameter]
\label{thm:gamma_hat_consistency}
Under Assumptions \ref{ass:gamma_identification} and
\ref{ass:raw_risk_uniform_convergence}, any measurable selection
\begin{equation}
\hat\gamma(\tau)
\in
\arg\min_{\gamma\in\Gamma}
\widehat Q_{n}(\gamma;\tau)
\end{equation}
satisfies
\begin{equation}
\hat\gamma(\tau)
\xrightarrow{p}
\gamma^{\star}(\tau).
\label{eq:gamma_hat_consistency}
\end{equation}
\end{theorem}
\noindent \textbf{Proof.} See Appendix \ref{subsec:thm_consistency_gammahat}.

This consistency result is an outer-risk result. It does not mean that $\gamma$
is identified by the first-order moment condition. Rather, it means that the
estimated second-order system-risk criterion is uniformly close to its population
counterpart, so the feasible selector chooses the same risk-optimal curvature
parameter asymptotically.

If the estimated criterion has a random limiting objective rather than a
deterministic population target, then the replacement of $\hat\gamma(\tau)$ by
$\gamma^{\star}(\tau)$ is not justified. In that case, $\hat\gamma(\tau)$ should
be interpreted as optimizing the realized finite-sample geometry, and
$B_{\theta,n}(\hat\gamma(\tau))$ should remain inside the second-order expansion.

\subsection{Second-order expansion of $\hat\theta_{\hat\gamma}$}

We now characterize how the data-selected curvature parameter $\hat\gamma(\tau)$
enters the expansion of the structural estimator. The key point is that
$\gamma$ does not affect the first-order term. It enters only through the
second-order correction.

We first impose a local stability condition on the second-order correction
$B_{\theta,n}(\gamma)$.

\begin{assumption}[Local stability of the second-order correction]
\label{ass:Btheta_stochastic_equicontinuity}
For fixed $\tau\in[0,1]$, the process $B_{\theta,n}(\gamma)$ is locally
stochastically equicontinuous at $\gamma^{\star}(\tau)$. That is, for every
$\varepsilon>0$,
\begin{equation}
\lim_{\zeta\downarrow 0}
\limsup_{n\to\infty}
P\left(
\sup_{\gamma\in\Gamma:\,
|\gamma-\gamma^{\star}(\tau)|\leq \zeta}
\left|
B_{\theta,n}(\gamma)
-
B_{\theta,n}(\gamma^{\star}(\tau))
\right|
>
\varepsilon
\right)
=
0.
\label{eq:Btheta_stochastic_equicontinuity}
\end{equation}
\end{assumption}

This condition requires the second-order correction to vary smoothly, in probability,
with the curvature parameter $\gamma$ around the risk-optimal value
$\gamma^{\star}(\tau)$. A sufficient condition is that
$B_{\theta,n}(\gamma)$ is differentiable in $\gamma$ and satisfies
\begin{equation}
\sup_{\gamma\in\Gamma}
\left|
\frac{\partial B_{\theta,n}(\gamma)}{\partial\gamma}
\right|
=
O_{p}(1).
\end{equation}

Under a uniform version of the fixed-$\gamma$ expansion, the selected estimator
satisfies
\begin{equation}
\hat\theta_{\hat\gamma(\tau)}-\theta_{0}
=
A_{0}\bar g(\theta_{0})
+
\frac{1}{n}
B_{\theta,n}(\hat\gamma(\tau))
+
o_{p}(n^{-1}),
\label{eq:theta_gammahat_expansion}
\end{equation}
where
\begin{equation}
A_{0}
=
(G_{0}^{\prime}\Omega_{0}^{-1}G_{0})^{-1}
G_{0}^{\prime}\Omega_{0}^{-1}.
\end{equation}
By Theorem \ref{thm:gamma_hat_consistency},
\begin{equation}
\hat\gamma(\tau)
\xrightarrow{p}
\gamma^{\star}(\tau).
\end{equation}
Together with Assumption \ref{ass:Btheta_stochastic_equicontinuity}, this implies
\begin{equation}
B_{\theta,n}(\hat\gamma(\tau))
=
B_{\theta,n}(\gamma^{\star}(\tau))
+
o_{p}(1).
\label{eq:Btheta_gammahat_replace}
\end{equation}
Substituting \eqref{eq:Btheta_gammahat_replace} into
\eqref{eq:theta_gammahat_expansion} yields
\begin{equation}
\hat\theta_{\hat\gamma(\tau)}-\theta_{0}
=
A_{0}\bar g(\theta_{0})
+
\frac{1}{n}
B_{\theta,n}(\gamma^{\star}(\tau))
+
o_{p}(n^{-1}).
\label{eq:theta_gammahat_consistent_expansion}
\end{equation}

This representation shows that the selected curvature parameter affects
$\hat\theta_{\hat\gamma(\tau)}$ only through the second-order correction term.
The first-order term $A_{0}\bar g(\theta_{0})$ is invariant to $\gamma$. Therefore,
$\hat\gamma(\tau)$ changes the finite-sample bias and sensitivity of the estimator,
but it does not alter the first-order asymptotic distribution. In particular,
$\gamma^{\star}(\tau)$ determines the order-$n^{-1}$ correction of the selected
estimator, while first-order identification remains governed by the original
moment condition for $\theta_{0}$.

\subsection{Post-selection inference for $\hat\theta_{\hat\gamma(\tau)}$}
\label{subsec:post_selection_inference}

The post-selection expansion implies that selecting $\hat\gamma(\tau)$ does not
affect the first-order asymptotic distribution of the structural estimator. The
selected curvature parameter enters only through the second-order correction term.
This distinction is important. First-order inference remains valid, but
finite-sample bias, mean squared error, and coverage distortion may still depend
on the selected second-order risk target $\gamma^{\star}(\tau)$.

Let
\begin{equation}
V_{\theta}
=
A_{0}\Omega_{0}A_{0}^{\prime}
=
(G_{0}^{\prime}\Omega_{0}^{-1}G_{0})^{-1},
\qquad
A_{0}
=
(G_{0}^{\prime}\Omega_{0}^{-1}G_{0})^{-1}
G_{0}^{\prime}\Omega_{0}^{-1}.
\end{equation}
Under the post-selection expansion,
\begin{equation}
\hat\theta_{\hat\gamma(\tau)}-\theta_{0}
=
A_{0}\bar g(\theta_{0})
+
\frac{1}{n}B_{\theta,n}(\gamma^{\star}(\tau))
+
o_{p}(n^{-1}).
\label{eq:post_selection_theta_expansion}
\end{equation}
Multiplying by $\sqrt n$ gives
\begin{equation}
\sqrt n
\left(
\hat\theta_{\hat\gamma(\tau)}-\theta_{0}
\right)
=
\sqrt n A_{0}\bar g(\theta_{0})
+
\frac{1}{\sqrt n}B_{\theta,n}(\gamma^{\star}(\tau))
+
o_{p}(n^{-1/2}).
\label{eq:post_selection_rootn_expansion}
\end{equation}
The first term on the right-hand side is invariant to $\gamma$, while the second
term is of smaller order. Hence,
\begin{equation}
\sqrt n
\left(
\hat\theta_{\hat\gamma(\tau)}-\theta_{0}
\right)
\xrightarrow{d}
N(0,V_{\theta}).
\label{eq:post_selection_first_order_normality}
\end{equation}

\subsection{First-order validity}

\begin{proposition}[First-order validity after selecting $\hat\gamma(\tau)$]
\label{prop:post_selection_first_order_validity}
Suppose that \eqref{eq:post_selection_theta_expansion} holds and that
$\widehat V_{\theta}$ is a consistent estimator of $V_{\theta}$. Then the usual
first-order confidence interval based on $\hat\theta_{\hat\gamma(\tau)}$ has
asymptotically correct coverage. In the scalar case,
\begin{equation}
CI_{1-\alpha}
=
\left[
\hat\theta_{\hat\gamma(\tau)}
-
z_{1-\alpha/2}
\sqrt{\frac{\widehat V_{\theta}}{n}},
\;
\hat\theta_{\hat\gamma(\tau)}
+
z_{1-\alpha/2}
\sqrt{\frac{\widehat V_{\theta}}{n}}
\right]
\end{equation}
satisfies
\begin{equation}
P
\left(
\theta_{0}\in CI_{1-\alpha}
\right)
\to
1-\alpha.
\end{equation}
\end{proposition}
\noindent \textbf{Proof.} See Appendix \ref{subsec:prop_post-selection-first-order-validity}.

\subsection{Higher-order bias and mean squared error}

The first-order validity result does not imply that the choice of $\gamma$ is
irrelevant in finite samples. The selected curvature target
$\gamma^{\star}(\tau)$ enters the post-selection expansion through the
order-$n^{-1}$ term. Hence, it can affect higher-order bias, mean squared error,
and coverage distortion even though it does not change the first-order asymptotic
distribution.

\begin{proposition}[Higher-order bias and mean squared error after selecting $\hat\gamma(\tau)$]
\label{prop:post_selection_bias_mse}
Suppose that the post-selection expansion
\begin{equation}
\hat\theta_{\hat\gamma(\tau)}-\theta_{0}
=
A_{0}\bar g(\theta_{0})
+
\frac{1}{n}B_{\theta,n}(\gamma^{\star}(\tau))
+
r_{\theta,n}
\label{eq:post_selection_expansion_with_remainder}
\end{equation}
holds, where $r_{\theta,n}=o_{p}(n^{-1})$. Let
\begin{equation}
Z_{n}
=
\sqrt n A_{0}\bar g(\theta_{0}).
\end{equation}
Assume that $E[\bar g(\theta_{0})]=0$ and that the remainder terms are uniformly
integrable so that
\begin{equation}
E[r_{\theta,n}]=o(n^{-1})
\end{equation}
and
\begin{equation}
E
\left[
r_{\theta,n}r_{\theta,n}^{\prime}
\right]
=
o(n^{-2}).
\end{equation}
Then the bias satisfies
\begin{equation}
E
\left[
\hat\theta_{\hat\gamma(\tau)}-\theta_{0}
\right]
=
\frac{1}{n}
E
\left[
B_{\theta,n}(\gamma^{\star}(\tau))
\right]
+
o(n^{-1}).
\label{eq:post_selection_bias}
\end{equation}
Moreover, the mean squared error matrix satisfies
\begin{align}
&E
\left[
\left(
\hat\theta_{\hat\gamma(\tau)}-\theta_{0}
\right)
\left(
\hat\theta_{\hat\gamma(\tau)}-\theta_{0}
\right)^{\prime}
\right]
\nonumber\\
&\quad =
\frac{1}{n}
E
\left[
Z_{n}Z_{n}^{\prime}
\right]
+
\frac{1}{n^{3/2}}
E
\left[
Z_{n}B_{\theta,n}(\gamma^{\star}(\tau))^{\prime}
+
B_{\theta,n}(\gamma^{\star}(\tau))Z_{n}^{\prime}
\right]
\nonumber\\
&\qquad
+
\frac{1}{n^{2}}
E
\left[
B_{\theta,n}(\gamma^{\star}(\tau))
B_{\theta,n}(\gamma^{\star}(\tau))^{\prime}
\right]
+
o(n^{-2}).
\label{eq:post_selection_mse_matrix}
\end{align}
If $E[Z_{n}Z_{n}^{\prime}]=V_{\theta}+o(1)$, then the leading variance component is
$V_{\theta}/n$, which is invariant to $\gamma$.
\end{proposition}
\noindent \textbf{Proof.} See Appendix \ref{subsec:pf_prop-post-selection-bias-mse}.

Proposition \ref{prop:post_selection_bias_mse} shows that
$\gamma^{\star}(\tau)$ determines the leading higher-order bias through
$E[B_{\theta,n}(\gamma^{\star}(\tau))]/n$. It also shows that the leading variance
term is invariant to $\gamma$, while the $\gamma$-dependent contribution to mean
squared error appears only through higher-order terms. If the cross term in
\eqref{eq:post_selection_mse_matrix} is zero or negligible, then the leading
$\gamma$-dependent MSE component is
\begin{equation}
\frac{1}{n^{2}}
E
\left[
B_{\theta,n}(\gamma^{\star}(\tau))
B_{\theta,n}(\gamma^{\star}(\tau))^{\prime}
\right].
\end{equation}
This is the sense in which the selected curvature parameter governs finite-sample
risk without changing first-order efficiency.

\subsection{Coverage distortion}

The same second-order correction can affect coverage distortion. Consider the
scalar studentized statistic
\begin{equation}
T_{n}(\hat\gamma(\tau))
=
\frac{
\sqrt n
\left(
\hat\theta_{\hat\gamma(\tau)}-\theta_{0}
\right)
}{
\sqrt{\widehat V_{\theta}}
}.
\end{equation}
Using \eqref{eq:post_selection_rootn_expansion},
\begin{equation}
T_{n}(\hat\gamma(\tau))
=
Z_{n}
+
\frac{1}{\sqrt n}
\frac{
B_{\theta,n}(\gamma^{\star}(\tau))
}{
\sqrt{V_{\theta}}
}
+
o_{p}(n^{-1/2}),
\label{eq:studentized_second_order_expansion}
\end{equation}
where $Z_{n}$ denotes the first-order standardized component. The leading
component $Z_{n}$ determines first-order coverage, while the second term affects
higher-order finite-sample behavior.

\begin{assumption}[Edgeworth expansion of the post-selection statistic]
\label{ass:edgeworth_post_selection}
For fixed $\tau\in[0,1]$, the distribution of
$T_{n}(\hat\gamma(\tau))$ admits the expansion
\begin{equation}
P
\left(
T_{n}(\hat\gamma(\tau))\leq t
\right)
=
\Phi(t)
+
n^{-1/2}
p_{1}(t;\gamma^{\star}(\tau))
\phi(t)
+
o(n^{-1/2}),
\label{eq:edgeworth_generic}
\end{equation}
uniformly for $t$ in a neighborhood of
${-z_{1-\alpha/2},z_{1-\alpha/2}}$. Here $\Phi(\cdot)$ and $\phi(\cdot)$ denote
the standard normal distribution function and density, respectively. The function
$p_{1}(t;\gamma^{\star}(\tau))$ is the first Edgeworth correction term and
collects the higher-order distributional effects of the studentized statistic,
including the second-order correction associated with the selected curvature
target $\gamma^{\star}(\tau)$.
\end{assumption}

\begin{proposition}[Coverage-error representation]
\label{prop:coverage_error_representation}
Under Assumption \ref{ass:edgeworth_post_selection}, the coverage error of the
usual two-sided first-order confidence interval satisfies
\begin{align}
&P
\left(
|T_{n}(\hat\gamma(\tau))|\leq z_{1-\alpha/2}
\right)
-
(1-\alpha)
\nonumber\\
&\quad =
n^{-1/2}
\left[
p_{1}(z_{1-\alpha/2};\gamma^{\star}(\tau))
\phi(z_{1-\alpha/2})
-
p_{1}(-z_{1-\alpha/2};\gamma^{\star}(\tau))
\phi(-z_{1-\alpha/2})
\right]
+
o(n^{-1/2}),
\label{eq:coverage_error_generic}
\end{align}
where $z_{1-\alpha/2}$ is the $(1-\alpha/2)$ quantile of the standard normal
distribution.
\end{proposition}
\noindent \textbf{Proof.} See Appendix \ref{subsec:pf_prop-coverage-error-representation}.

Proposition \ref{prop:coverage_error_representation} shows that
$\gamma^{\star}(\tau)$ can affect finite-sample coverage through the Edgeworth
correction term $p_{1}(\cdot;\gamma^{\star}(\tau))$. Thus, although
$\hat\gamma(\tau)$ does not alter the first-order distribution, its risk target
$\gamma^{\star}(\tau)$ determines the second-order correction to the studentized
statistic and can therefore affect coverage distortion.

\section{Monte Carlo simulations}

We conduct Monte Carlo simulations with 10,000 replications and sample size $n\in\{30,50,100\}$ to evaluate the finite-sample behavior of the proposed data-selected CRPD estimator. The simulations are designed to illustrate three implications of the theory. First, the CRPD power parameter $\gamma$ does not change first-order identification, but it can affect finite-sample performance through second-order terms. Second, the selected value $\hat\gamma(\tau)$ should be understood as a system-risk choice, rather than as a structural parameter. Third, the multiplier component of the selection criterion is informative because it captures the finite-sample cost of enforcing the moment restrictions.

\subsection{Simulation design}

We consider
\[
    Y_{i}
    =
    W_i^{\prime}\beta^{0}
    +
    u_{i}
    =
    \beta^{0}_{0}
    +
    \beta^{0}_{1}X_i
    +
    u_{i},
\]
where $Y_{i}$ is the outcome, $W_i=(1,X_i)^{\prime}$, and $u_{i}$ is the error term.
Let the true target parameter $\beta^{0}$ be
\[
    \beta_0=(\beta^{0}_{0},\beta^{0}_{1})^{\prime}=(1,2)^{\prime}.
\]

The CRPD estimator is based on the overidentified moment condition
\[
    E
    \left[
        Z_i
        \left(
            Y_i-W_i^{\prime}\beta_0
        \right)
    \right]
    =
    0,
\]
where $Z_i$ contains the regressors and additional nonlinear instruments. Specifically,
we set
\[
    Z_i
    =
    \left(
        1,\,
        X_i,\,
        \widetilde H_2(X_i),\,
        \widetilde H_3(X_i)
    \right)^{\prime},
\]
where
\[
    H_2(X_i)=X_i^2-1,
    \qquad
    H_3(X_i)=X_i^3-3X_i.
\]
The variables $\widetilde H_2(X_i)$ and $\widetilde H_3(X_i)$ denote centered and
scaled versions of $H_2(X_i)$ and $H_3(X_i)$ in the realized fixed sample. That is,
for $j\in\{2,3\}$,
\[
    \widetilde H_j(X_i)
    =
    \frac{
        H_j(X_i)
        -
        n^{-1}\sum_{\ell=1}^{n}H_j(X_{\ell})
    }{
        \left[
            (n-1)^{-1}
            \sum_{\ell=1}^{n}
            \left\{
                H_j(X_{\ell})
                -
                n^{-1}\sum_{m=1}^{n}H_j(X_m)
            \right\}^{2}
        \right]^{1/2}
    }.
\]
These Hermite-type polynomial instruments are orthogonal to lower-order terms under
a standard normal design and reduce collinearity relative to raw polynomial
instruments. Since $q=4$ and $p=2$, the moment system is overidentified. This
overidentification is important because, in just-identified settings, the CRPD power
parameter has little scope to affect the implied weights.

All error designs satisfy
\[
    E[u_{ir}\mid X_i]=0,
\]
so the instrumental moment conditions are valid. 

We consider five DGPs, as summarized in Table \ref{tab:mc_dgp_design}. The first two
are homoskedastic benchmark cases: standard normal errors and standardized
$t_{10}$ errors. The remaining three introduce stronger finite-sample features while
preserving conditional mean zero. 

Specifically, the five DGPs are as follows: first, under the normal DGP,
\[
    u_{ir}\sim N(0,1).
\]
Second, under the standardized $t_{10}$ DGP,
\[
    u_{ir}
    =
    \frac{e_{ir}}{\sqrt{10/(10-2)}},
    \qquad
    e_{ir}\sim t_{10},
\]
so that the variance is normalized to one. Third, under the contaminated-normal DGP,
\[
    u_{ir}
    =
    \frac{
        (1-\xi_{ir})e_{ir,0}
        +
        \xi_{ir}e_{ir,1}
    }{
        \sqrt{0.90\cdot 1^2+0.10\cdot 5^2}
    },
\]
where
\[
    \xi_{ir}\sim \mathrm{Bernoulli}(0.10),
    \qquad
    e_{ir,0}\sim N(0,1),
    \qquad
    e_{ir,1}\sim N(0,5^2).
\]
This design introduces occasional large shocks while keeping the unconditional
variance normalized.

Fourth, under the heteroskedastic $t_{10}$ DGP,
\[
    u_{ir}
    =
    \sigma(X_i)e_{ir},
    \qquad
    e_{ir}
    =
    \frac{\tilde e_{ir}}{\sqrt{10/(10-2)}},
    \qquad
    \tilde e_{ir}\sim t_{10},
\]
where
\[
    \sigma(X_i)
    =
    \frac{
        0.50+0.75|X_i|
    }{
        \left[
            n^{-1}\sum_{\ell=1}^{n}
            \left(
                0.50+0.75|X_{\ell}|
            \right)^2
        \right]^{1/2}
    }.
\]
This design makes the conditional variance larger for high-$|X_i|$ observations while
maintaining average variance equal to one.

Finally, under the $X$-dependent contamination DGP, high-leverage observations are
more likely to receive large shocks. Let
\[
    p_i
    =
    0.04
    +
    0.28\cdot
    1\left\{
        |X_i|
        \geq
        Q_{0.80}(|X|)
    \right\},
\]
where $Q_{0.80}(|X|)$ is the $80$th percentile of $\{|X_i|\}_{i=1}^{n}$. Then
\[
    u_{ir}
    =
    \frac{
        (1-\xi_{ir})e_{ir,0}
        +
        \xi_{ir}e_{ir,1}
    }{
        \left[
            n^{-1}\sum_{\ell=1}^{n}
            \left\{
                (1-p_{\ell})\cdot 1^2
                +
                p_{\ell}\cdot 6^2
            \right\}
        \right]^{1/2}
    },
\]
where
\[
    \xi_{ir}\mid X_i\sim \mathrm{Bernoulli}(p_i),
    \qquad
    e_{ir,0}\sim N(0,1),
    \qquad
    e_{ir,1}\sim N(0,6^2).
\]
This design preserves $E[u_{ir}\mid X_i]=0$ but creates stronger moment tension
because large shocks are more likely to occur at observations with high leverage.

Notably, the more difficult DGPs, especially the heteroskedastic and $X$-dependent
contamination designs, are included so that the selected curvature parameter
$\hat\gamma(\tau)$ can respond visibly to changes in tail behavior, heteroskedasticity,
and leverage-dependent shocks.

\small
\begin{longtable}{p{0.19\linewidth}p{0.43\linewidth}p{0.30\linewidth}}
\caption{Data-generating processes in the fixed-design regression simulation}
\label{tab:mc_dgp_design}\\
\toprule
DGP & Error design & Purpose \\
\midrule
\endfirsthead
\multicolumn{3}{c}{\tablename\ \thetable{} -- continued from previous page}\\
\toprule
DGP & Error design & Purpose \\
\midrule
\endhead
\midrule
\multicolumn{3}{r}{Continued on next page}\\
\endfoot

\bottomrule
\multicolumn{3}{p{0.95\linewidth}}{\footnotesize \textit{Notes:} All DGPs satisfy $E[u_{ir}\mid X_i]=0$, so the instrumental moment restrictions remain valid. The contaminated and heteroskedastic designs create stronger finite-sample tension in the moment system.}\\
\endlastfoot
Normal & $u_{ir}\sim N(0,1)$ & Homoskedastic benchmark \\
$t_{10}$ & $u_{ir}=e_{ir}/\sqrt{10/(10-2)},\ e_{ir}\sim t_{10}$ & Heavy-tailed benchmark with unit variance \\
Contaminated normal & $0.90N(0,1)+0.10N(0,5^2)$, scaled to unit variance & Independent tail contamination \\
Heteroskedastic $t_{10}$ & $u_{ir}=\sigma(X_i)e_{ir}$, $e_{ir}$ standardized $t_{10}$ & Variance increases with $|X_i|$ \\
$X$-dependent contamination & High-$|X_i|$ observations have larger contamination probability & Leverage-dependent tail shocks \\
\end{longtable}
\normalsize

Because the simulation focuses on higher-order finite-sample properties, a relatively
large number of Monte Carlo replications is required. Quantities such as bias,
coverage distortion, multiplier stability, and the distribution of the selected
curvature parameter are more sensitive to simulation noise than first-order averages.
We therefore use $10{,}000$ Monte Carlo replications and consider sample sizes
\[
    n\in\{30,50,100\}.
\]

The candidate set for the CRPD power parameter is
\[
    \Gamma=[-2,2],
\]
with grid spacing $0.1$. The special values $\gamma=-1$ and $\gamma=0$ are handled
numerically by small perturbations, since the implementation uses the general
power-divergence formula. We report results for
\[
    \tau\in\{1,0.5,0\}.
\]
The parameter $\tau$ controls how the finite-sample performance of the CRPD system
is evaluated. Larger values place more weight on the second-order behavior of the
structural estimator, while smaller values place more weight on the stability of
moment enforcement through the Lagrange multipliers and implied weights. The case
$\tau=1$ uses only the structural second-order risk component, $\tau=0$ uses only
the multiplier-stabilization component, and $\tau=0.5$ gives equal weight to the two.
Thus, reporting $\tau\in\{1,0.5,0\}$ allows us to examine how the selected curvature
parameter changes as the criterion moves from a purely structural-risk objective to a
purely multiplier-stability objective.

\subsection{Implementation of the system-risk selector}
\label{subsec:simulation_system_risk_selector}

We implement the proposed system-risk selector within each Monte Carlo replication. For each replication, we first compute a common first-order GMM-style benchmark. Let $\hat\beta_{\mathrm{GMM}}$ denote this benchmark. Let
\begin{equation}
\hat\lambda^{(1)}
=
\widehat\Omega^{-1}\bar g(\hat\beta_{\mathrm{GMM}})
\end{equation}
denote the common first-order multiplier benchmark, where $\widehat\Omega$ is the sample covariance matrix of the moments evaluated at the benchmark. This benchmark corresponds to the leading term in the multiplier expansion and is therefore independent of $\gamma$.

For each candidate value $\gamma\in\Gamma$, the CRPD estimator gives the structural estimate $\hat\beta_{\gamma}$ and the associated moment multiplier $\hat\lambda_{\gamma}$. We approximate the second-order corrections by comparing each CRPD solution with the common first-order benchmark:
\begin{equation}
\widehat B_{\theta,n}(\gamma)
=
n
\left\{
\hat\beta_{\gamma}
-
\hat\beta_{\mathrm{GMM}}
\right\},
\qquad
\widehat B_{\lambda,n}(\gamma)
=
n
\left\{
\hat\lambda_{\gamma}
-
\hat\lambda^{(1)}
\right\}.
\label{eq:feasible_B_definitions}
\end{equation}
The first object is a feasible plug-in proxy for the second-order structural correction. The second object is a feasible plug-in proxy for the second-order multiplier correction.

Using these quantities, we compute the raw sample analogues of the two system-risk components:
\begin{equation}
\widehat C_{\theta,\mathrm{raw}}(\gamma)
=
\widehat B_{\theta,n}(\gamma)^{\prime}
\widehat K
\widehat B_{\theta,n}(\gamma),
\label{eq:Ctheta_sample_raw}
\end{equation}
and
\begin{equation}
\widehat C_{\lambda,\mathrm{raw}}(\gamma)
=
\widehat B_{\lambda,n}(\gamma)^{\prime}
\widehat\Omega
\widehat B_{\lambda,n}(\gamma).
\label{eq:Clambda_sample_raw}
\end{equation}
The first component measures standardized second-order structural distortion. The second component measures the second-order cost of moment stabilization, expressed in the natural covariance scale of the moment vector.

Because the two components may have different numerical scales, we normalize them over the candidate set $\Gamma$ before combining them. For $j\in{\theta,\lambda}$, define
\begin{equation}
\widehat{\widetilde C}_{j}(\gamma)
=
\frac{
\widehat C_{j,\mathrm{raw}}(\gamma)
-
\min_{\tilde\gamma\in\Gamma}
\widehat C_{j,\mathrm{raw}}(\tilde\gamma)
}{
\max_{\tilde\gamma\in\Gamma}
\widehat C_{j,\mathrm{raw}}(\tilde\gamma)
-
\min_{\tilde\gamma\in\Gamma}
\widehat C_{j,\mathrm{raw}}(\tilde\gamma)
},
\qquad
j\in{\theta,\lambda}.
\label{eq:sample_minmax_normalization}
\end{equation}
If the denominator is zero, the corresponding component is constant over $\Gamma$ and is set to zero after normalization. This places the structural and multiplier components on a common $[0,1]$ scale.

The feasible system-risk criterion is then
\begin{equation}
\widehat C_{n}^{\mathrm{sys}}(\gamma;\tau)
=
\tau
\widehat{\widetilde C}_{\theta}(\gamma)
+
(1-\tau)
\widehat{\widetilde C}_{\lambda}(\gamma),
\qquad
\tau\in[0,1].
\label{eq:sample_system_risk}
\end{equation}
The selected CRPD power parameter is
\begin{equation}
\hat\gamma(\tau)
=
\arg\min_{\gamma\in\Gamma}
\widehat C_{n}^{\mathrm{sys}}(\gamma;\tau).
\label{eq:gamma_hat}
\end{equation}

This implementation differs from a conventional structural-loss selector. A structural-loss criterion evaluates $\gamma$ only through the behavior of $\hat\beta_{\gamma}$. In contrast, the proposed criterion evaluates the finite-sample behavior of the CRPD estimating system. The multiplier component is included because it measures how much adjustment is required to enforce the moment restrictions. A value of $\gamma$ may produce a reasonable structural estimate while requiring an unstable multiplier and highly distorted weights. The system-risk criterion therefore selects $\gamma$ by balancing standardized structural distortion and multiplier stability.

To monitor numerical stability and the implied reweighting mechanism, we also report multiplier and weight diagnostics. These include the multiplier norm $|\hat\lambda_{\hat\gamma}|$, the quadratic multiplier measure $\hat\lambda_{\hat\gamma}^{\prime}\widehat\Omega\hat\lambda_{\hat\gamma}$, and the effective sample size
\begin{equation}
ESS(\hat\gamma)
=
\frac{1}{
\sum_{i=1}^{n}
\hat\pi_{i,\hat\gamma}^{2}
}.
\end{equation}
We also report the normalized effective sample size $ESS(\hat\gamma)/n$ and boundary-hit rates based on whether the minimum estimated weight falls below a small absolute or relative threshold. These diagnostics are directly connected to the multiplier-stability motivation of the system-risk criterion.

\subsection{Simulation results}

The simulation results are summarized in Tables \ref{tab:mc_selected_summary},
\ref{tab:mc_convergence_by_n_tau},
\ref{tab:mc_gamma_distribution},
\ref{tab:mc_system_diagnostics},
\ref{tab:mc_raw_risk_weight_diagnostics},
\ref{tab:mc_tau05_comparison}, and
\ref{tab:mc_n100_dgp_sensitivity}. The parameter of interest for inference is the slope coefficient $\beta_1$. The tables are organized to show, respectively, the main finite-sample performance of the selected estimator, the average convergence pattern across sample sizes, the distribution of the selected curvature parameter, the multiplier-stability diagnostics, the raw risk and weight diagnostics, the behavior of the balanced selector relative to benchmark selectors, and the DGP sensitivity of the results at $n=100$.

Table \ref{tab:mc_selected_summary} reports the main Monte Carlo performance of the selected CRPD estimator for each DGP, sample size, and value of $\tau$. Across all designs, the bias of $\hat\beta_1$ is small relative to the RMSE. The largest reported biases are close to zero in magnitude, indicating that the selected curvature parameter does not change the structural target of estimation. This is consistent with the theory: $\gamma$ affects the higher-order behavior of the estimator--multiplier system, but it does not alter the first-order identification of $\beta_1$.

Table \ref{tab:mc_selected_summary} also shows that the finite-sample dispersion of the estimator declines as $n$ increases. Across DGPs and values of $\tau$, the RMSE of $\hat\beta_1$ ranges approximately from $0.19$ to $0.33$ when $n=30$, from $0.14$ to $0.22$ when $n=50$, and from $0.09$ to $0.15$ when $n=100$. The improvement is observed not only under the normal and $t_{10}$ designs, but also under the contaminated, heteroskedastic, and $X$-dependent contamination designs. Thus, the selected CRPD estimator becomes more stable as the sample size grows.

The coverage and SD/SE columns in Table \ref{tab:mc_selected_summary} show the expected finite-sample pattern. First-order confidence intervals under-cover substantially when $n=30$, with empirical coverage between roughly $0.60$ and $0.69$. Coverage improves when $n=50$, rising to approximately $0.80$--$0.84$, and improves further when $n=100$, reaching approximately $0.88$--$0.90$. The SD/SE ratio follows the same pattern: it is well above one in small samples and moves closer to one as $n$ increases. Hence, the first-order standard error approximation improves with the sample size, while the remaining under-coverage reflects the higher-order distortions emphasized by the theory.

Table \ref{tab:mc_convergence_by_n_tau} averages the results across the five DGPs and highlights the joint role of $n$ and $\tau$. The convergence pattern is clear. Average RMSE falls from about $0.25$--$0.26$ when $n=30$ to about $0.11$--$0.12$ when $n=100$. Average coverage rises from about $0.62$--$0.67$ when $n=30$ to about $0.89$ when $n=100$, and the SD/SE ratio falls toward one. These averages summarize the improvement in the first-order approximation as the sample size grows.

Table \ref{tab:mc_convergence_by_n_tau} also illustrates the effect of the system-risk weight $\tau$. When $\tau=1$, the selector uses only the structural second-order component and chooses smaller values of $\gamma$. Averaged across DGPs, $\bar{\hat\gamma}$ is $-0.849$ when $n=30$, $-0.090$ when $n=50$, and $0.346$ when $n=100$. When the multiplier-stability component is included, the selected values move toward more positive values. For $\tau=0.5$, the corresponding averages are $0.459$, $0.673$, and $0.811$. This shows that adding multiplier stability pulls the selected curvature parameter toward a region of the CRPD family associated with smoother moment enforcement.

Table \ref{tab:mc_gamma_distribution} reports the Monte Carlo distribution of the selected curvature parameter $\hat\gamma(\tau)$. The distributional evidence shows that $\hat\gamma(\tau)$ is genuinely data-dependent. When $\tau=1$, the selected values are more dispersed and more sensitive to the DGP. In small samples, the structural-only selector often selects negative values of $\gamma$, especially under the contaminated, heteroskedastic, and $X$-dependent contamination designs. For example, when $n=30$, the median selected value is $-1.4$ under independent contamination, $-1.5$ under heteroskedastic $t_{10}$ errors, and $-1.5$ under $X$-dependent contamination.

By contrast, Table \ref{tab:mc_gamma_distribution} shows that the selectors with $\tau=0.5$ and $\tau=0$ choose more positive values of $\gamma$ and exhibit less dispersion. The difference between $\tau=0.5$ and $\tau=0$ is modest but visible when $n=30$ or $n=50$: the balanced selector often chooses slightly smaller values of $\gamma$ than the purely multiplier-based selector. This indicates that the structural component still contributes to the selected curvature parameter when finite-sample distortions are relatively large. As $n$ increases to $100$, the selected values under $\tau=0.5$ and $\tau=0$ become much closer, while the structural-only selector with $\tau=1$ remains distinct.

Table \ref{tab:mc_system_diagnostics} reports normalized system-risk components and multiplier-stability diagnostics. The table shows that moving from $\tau=1$ to $\tau=0.5$ generally lowers the multiplier norm $\|\hat\lambda\|$, lowers the quadratic multiplier measure $\hat\lambda^{\prime}\widehat\Omega\hat\lambda$, and raises the normalized effective sample size $ESS/n$. These improvements are especially visible in the more difficult DGPs. For example, under $X$-dependent contamination with $n=30$, $\|\hat\lambda\|$ falls from $10.019$ under $\tau=1$ to $3.277$ under $\tau=0.5$, while $ESS/n$ rises from $0.759$ to $0.900$.

The implication of Table \ref{tab:mc_system_diagnostics} is that the structural-only criterion can select curvature values that perform well according to structural second-order risk but require stronger moment-enforcement pressure. Adding the multiplier component reduces this pressure and improves the stability of the implied weighting system. The table also clarifies why $\tau=0.5$ is useful: it does not simply minimize the multiplier component, but it substantially improves multiplier stability relative to $\tau=1$ while retaining the structural component of the criterion.

Table \ref{tab:mc_raw_risk_weight_diagnostics} reports the unnormalized risk components and selected-weight diagnostics. This table shows the scale of the multiplier-stability gains before normalization. In the more difficult designs, $C_{\lambda,\mathrm{raw}}$ can be very large under the structural-only selector. For example, under $X$-dependent contamination with $n=30$, $C_{\lambda,\mathrm{raw}}$ is $2.58\times 10^{6}$ under $\tau=1$, but falls to $3.06\times 10^{5}$ under $\tau=0.5$ and to $3.45\times 10^{4}$ under $\tau=0$. At the same time, the maximum implied weight falls and the effective sample size rises.

Table \ref{tab:mc_raw_risk_weight_diagnostics} therefore clarifies the role of the multiplier-stability component. The purpose of $C_{\lambda}$ is to reduce overall instability in moment enforcement, including large multiplier adjustments and upper-tail concentration of implied weights. The decline in $C_{\lambda,\mathrm{raw}}$, the reduction in $\max_i\hat\pi_i$, and the increase in $ESS$ show that $\tau<1$ produces a more stable weighting system, especially in small samples and under high-stress DGPs. The table also shows the tradeoff: the raw structural component $C_{\theta,\mathrm{raw}}$ can increase when the selector places weight on multiplier stability.

Table \ref{tab:mc_tau05_comparison} compares the balanced selector $\tau=0.5$ with the two benchmark selectors, $\tau=1$ and $\tau=0$. Relative to the structural-only selector, the balanced selector usually selects a larger value of $\gamma$, lowers the multiplier norm, and often delivers similar or better RMSE and coverage. The reductions in $\|\hat\lambda\|$ relative to $\tau=1$ are especially large when $n=30$, where moment enforcement is most fragile. For example, under $X$-dependent contamination with $n=30$, the balanced selector lowers the multiplier norm by $6.742$ relative to the structural-only selector.

Relative to the multiplier-only selector, Table \ref{tab:mc_tau05_comparison} shows that the balanced selector is much closer, especially when $n=100$. However, the two selectors are not identical in smaller samples. The differences in $\Delta\gamma_{.5-0}$ are modest but visible for $n=30$ and $n=50$, and the balanced selector often delivers higher coverage than the multiplier-only selector. This supports the interpretation of $\tau=0.5$ as a compromise: it retains structural second-order information while substantially improving multiplier stability relative to $\tau=1$.

Finally, Table \ref{tab:mc_n100_dgp_sensitivity} isolates the larger-sample case. Even when $n=100$, the structural-only selector remains sensitive to the DGP. Under $\tau=1$, the mean selected value of $\hat\gamma$ is approximately $0.49$ under normal errors and $0.46$ under $t_{10}$ errors, but falls to $0.28$ under independent contamination and to $0.13$ under $X$-dependent contamination. Thus, the more difficult DGPs continue to push the structural-risk selector toward smaller values of $\gamma$.

Table \ref{tab:mc_n100_dgp_sensitivity} also shows that the selected values under $\tau=0.5$ and $\tau=0$ are much closer at $n=100$ than in smaller samples. For these two selectors, $\bar{\hat\gamma}$ is generally between $0.74$ and $0.86$, and $ESS/n$ remains close to one across DGPs. This does not imply that the structural and multiplier objectives are fully aligned, because the structural-only selector remains clearly different. Rather, it indicates that once the multiplier-stability component is included, even with equal weight, the selected curvature parameter is pulled toward the multiplier-stable region of the CRPD family.

Overall, the Monte Carlo evidence is consistent with the theoretical results. The CRPD power parameter does not change first-order identification, and the selected estimator remains centered around the true structural parameter. However, the choice of $\gamma$ affects finite-sample performance through second-order structural distortion and multiplier stability. The selected curvature parameter responds to the finite-sample geometry of the moment system: it is more sensitive to the DGP when the structural component dominates, and it moves toward more stable positive values when the multiplier component receives greater weight. These results support the proposed interpretation of $\gamma$ as a data-selected curvature parameter governing the higher-order behavior of the CRPD estimator--multiplier system.

\footnotesize
\begin{longtable}{lrrrrrrrrr}
\caption{Monte Carlo performance of the selected CRPD estimator}
\label{tab:mc_selected_summary}\\
\toprule
DGP & $n$ & $\tau$ & $\bar{\hat\gamma}$ & Med.$(\hat\gamma)$ & Bias & RMSE & Coverage & Cov. dist. & SD/SE \\
\midrule
\endfirsthead
\multicolumn{10}{c}{\tablename\ \thetable{} -- continued from previous page}\\
\toprule
DGP & $n$ & $\tau$ & $\bar{\hat\gamma}$ & Med.$(\hat\gamma)$ & Bias & RMSE & Coverage & Cov. dist. & SD/SE \\
\midrule
\endhead
\midrule
\multicolumn{10}{r}{Continued on next page}\\
\endfoot

\bottomrule
\endlastfoot
Normal & 30 & 1.0 & -0.623 & -1.000 & -0.0025 & 0.264 & 0.649 & -0.301 & 1.970 \\
Normal & 30 & 0.5 & 0.518 & 0.800 & -0.0013 & 0.251 & 0.661 & -0.289 & 1.922 \\
Normal & 30 & 0.0 & 0.527 & 1.000 & -0.0016 & 0.243 & 0.600 & -0.350 & 2.136 \\
$t_{10}$ & 30 & 1.0 & -0.675 & -1.000 & 0.0027 & 0.256 & 0.653 & -0.297 & 1.974 \\
$t_{10}$ & 30 & 0.5 & 0.502 & 0.700 & 0.0030 & 0.245 & 0.662 & -0.288 & 1.940 \\
$t_{10}$ & 30 & 0.0 & 0.527 & 0.900 & 0.0041 & 0.238 & 0.603 & -0.347 & 2.154 \\
Contam. & 30 & 1.0 & -0.984 & -1.400 & -0.0038 & 0.224 & 0.655 & -0.295 & 2.128 \\
Contam. & 30 & 0.5 & 0.425 & 0.400 & -0.0034 & 0.208 & 0.670 & -0.280 & 2.063 \\
Contam. & 30 & 0.0 & 0.474 & 0.700 & -0.0018 & 0.191 & 0.634 & -0.316 & 2.189 \\
Hetero. $t_{10}$ & 30 & 1.0 & -0.910 & -1.500 & 0.0051 & 0.312 & 0.687 & -0.263 & 1.764 \\
Hetero. $t_{10}$ & 30 & 0.5 & 0.426 & 0.500 & 0.0054 & 0.308 & 0.690 & -0.260 & 1.754 \\
Hetero. $t_{10}$ & 30 & 0.0 & 0.522 & 0.700 & 0.0051 & 0.330 & 0.633 & -0.317 & 1.967 \\
$X$-contam. & 30 & 1.0 & -1.054 & -1.500 & 0.0032 & 0.259 & 0.675 & -0.275 & 2.091 \\
$X$-contam. & 30 & 0.5 & 0.422 & 0.400 & 0.0025 & 0.251 & 0.682 & -0.268 & 2.078 \\
$X$-contam. & 30 & 0.0 & 0.509 & 0.700 & 0.0016 & 0.250 & 0.642 & -0.308 & 2.201 \\
\hline
Normal & 50 & 1.0 & 0.150 & 1.000 & 0.0005 & 0.173 & 0.808 & -0.142 & 1.461 \\
Normal & 50 & 0.5 & 0.727 & 1.000 & -0.0005 & 0.167 & 0.815 & -0.135 & 1.428 \\
Normal & 50 & 0.0 & 0.725 & 1.000 & -0.0018 & 0.167 & 0.802 & -0.147 & 1.455 \\
$t_{10}$ & 50 & 1.0 & 0.126 & 1.000 & -0.0008 & 0.171 & 0.799 & -0.151 & 1.488 \\
$t_{10}$ & 50 & 0.5 & 0.732 & 1.000 & -0.0018 & 0.164 & 0.809 & -0.141 & 1.449 \\
$t_{10}$ & 50 & 0.0 & 0.735 & 1.000 & -0.0016 & 0.165 & 0.798 & -0.151 & 1.478 \\
Contam. & 50 & 1.0 & -0.249 & 0.300 & 0.0018 & 0.147 & 0.798 & -0.151 & 1.608 \\
Contam. & 50 & 0.5 & 0.673 & 1.000 & 0.0021 & 0.137 & 0.814 & -0.137 & 1.525 \\
Contam. & 50 & 0.0 & 0.727 & 1.000 & 0.0023 & 0.136 & 0.803 & -0.147 & 1.542 \\
Hetero. $t_{10}$ & 50 & 1.0 & -0.095 & 0.500 & 0.0019 & 0.205 & 0.840 & -0.110 & 1.334 \\
Hetero. $t_{10}$ & 50 & 0.5 & 0.624 & 0.900 & 0.0008 & 0.206 & 0.835 & -0.115 & 1.345 \\
Hetero. $t_{10}$ & 50 & 0.0 & 0.648 & 1.000 & 0.0006 & 0.220 & 0.806 & -0.144 & 1.452 \\
$X$-contam. & 50 & 1.0 & -0.383 & -0.400 & 0.0013 & 0.185 & 0.825 & -0.126 & 1.529 \\
$X$-contam. & 50 & 0.5 & 0.609 & 1.000 & 0.0016 & 0.178 & 0.833 & -0.117 & 1.485 \\
$X$-contam. & 50 & 0.0 & 0.667 & 1.000 & 0.0030 & 0.188 & 0.805 & -0.145 & 1.577 \\
\hline
Normal & 100 & 1.0 & 0.486 & 1.000 & 0.0012 & 0.111 & 0.886 & -0.064 & 1.212 \\
Normal & 100 & 0.5 & 0.858 & 1.000 & 0.0014 & 0.109 & 0.888 & -0.062 & 1.202 \\
Normal & 100 & 0.0 & 0.852 & 1.000 & 0.0017 & 0.109 & 0.887 & -0.063 & 1.205 \\
$t_{10}$ & 100 & 1.0 & 0.461 & 1.000 & -0.0021 & 0.108 & 0.890 & -0.060 & 1.205 \\
$t_{10}$ & 100 & 0.5 & 0.862 & 1.000 & -0.0022 & 0.106 & 0.893 & -0.057 & 1.193 \\
$t_{10}$ & 100 & 0.0 & 0.845 & 1.000 & -0.0018 & 0.108 & 0.891 & -0.059 & 1.207 \\
Contam. & 100 & 1.0 & 0.278 & 1.000 & 0.0007 & 0.093 & 0.886 & -0.064 & 1.247 \\
Contam. & 100 & 0.5 & 0.828 & 1.000 & 0.0009 & 0.091 & 0.889 & -0.061 & 1.222 \\
Contam. & 100 & 0.0 & 0.829 & 1.000 & 0.0008 & 0.093 & 0.884 & -0.066 & 1.250 \\
Hetero. $t_{10}$ & 100 & 1.0 & 0.375 & 1.000 & 0.0016 & 0.137 & 0.899 & -0.051 & 1.156 \\
Hetero. $t_{10}$ & 100 & 0.5 & 0.766 & 1.000 & 0.0014 & 0.138 & 0.898 & -0.052 & 1.164 \\
Hetero. $t_{10}$ & 100 & 0.0 & 0.740 & 1.000 & 0.0012 & 0.145 & 0.883 & -0.067 & 1.220 \\
$X$-contam. & 100 & 1.0 & 0.130 & 1.000 & -0.0001 & 0.124 & 0.901 & -0.049 & 1.218 \\
$X$-contam. & 100 & 0.5 & 0.740 & 1.000 & -0.0003 & 0.124 & 0.901 & -0.049 & 1.217 \\
$X$-contam. & 100 & 0.0 & 0.736 & 1.000 & -0.0012 & 0.131 & 0.883 & -0.067 & 1.287 \\
\end{longtable}
\vspace{-2.5em}
\begin{center}
\begin{minipage}{0.85\linewidth}
\footnotesize Note: Results are based on 10,000 Monte Carlo replications. The target parameter is the slope coefficient $\beta_1$. Coverage is for nominal 95\% first-order confidence intervals. SD/SE is the ratio of the empirical standard deviation of $\hat\beta_1$ to the average estimated standard error.
\end{minipage}
\end{center}
\vspace{0.5em}
\normalsize

\small
\begin{longtable}{rrrrrrrrr}
\caption{Average Monte Carlo performance by sample size and system-risk weight}
\label{tab:mc_convergence_by_n_tau}\\
\toprule
$n$ & $\tau$ & $\bar{\hat\gamma}$ & RMSE & Coverage & Cov. dist. & SD/SE & $\|\hat\lambda\|$ & $ESS/n$ \\
\midrule
\endfirsthead

\multicolumn{9}{c}{\tablename\ \thetable{} -- continued from previous page}\\
\toprule
$n$ & $\tau$ & $\bar{\hat\gamma}$ & RMSE & Coverage & Cov. dist. & SD/SE & $\|\hat\lambda\|$ & $ESS/n$ \\
\midrule
\endhead

\midrule
\multicolumn{9}{r}{Continued on next page}\\
\endfoot

\bottomrule
\endlastfoot

30 & 1.0 & -0.849 & 0.263 & 0.664 & -0.286 & 1.986 & 6.988 & 0.789 \\
30 & 0.5 & 0.459 & 0.253 & 0.673 & -0.277 & 1.951 & 2.474 & 0.906 \\
30 & 0.0 & 0.512 & 0.250 & 0.623 & -0.327 & 2.129 & 1.350 & 0.908 \\
\midrule
50 & 1.0 & -0.090 & 0.176 & 0.814 & -0.136 & 1.484 & 0.848 & 0.906 \\
50 & 0.5 & 0.673 & 0.170 & 0.821 & -0.129 & 1.446 & 0.441 & 0.951 \\
50 & 0.0 & 0.700 & 0.175 & 0.803 & -0.147 & 1.501 & 0.360 & 0.948 \\
\midrule
100 & 1.0 & 0.346 & 0.114 & 0.893 & -0.057 & 1.207 & 0.237 & 0.965 \\
100 & 0.5 & 0.811 & 0.114 & 0.894 & -0.056 & 1.199 & 0.185 & 0.977 \\
100 & 0.0 & 0.800 & 0.117 & 0.886 & -0.064 & 1.234 & 0.179 & 0.975 \\
\end{longtable}
\vspace{-2em}
\begin{center}
\begin{minipage}{0.85\linewidth}
\footnotesize Note: Each entry averages the corresponding statistic across the five DGPs. This table summarizes the convergence pattern as $n$ increases and the tradeoff as $\tau$ changes.
\end{minipage}
\end{center}
\vspace{0.5em}
\normalsize

\footnotesize
\begin{longtable}{lrrrrrrr}
\caption{Distribution of the selected CRPD curvature parameter}
\label{tab:mc_gamma_distribution}\\
\toprule
DGP & $n$ & $\tau$ & Mean & SD & Q1 & Median & Q3 \\
\midrule
\endfirsthead
\multicolumn{8}{c}{\tablename\ \thetable{} -- continued from previous page}\\
\toprule
DGP & $n$ & $\tau$ & Mean & SD & Q1 & Median & Q3 \\
\midrule
\endhead
\midrule
\multicolumn{8}{r}{Continued on next page}\\
\endfoot

\bottomrule
\endlastfoot
Normal & 30 & 1.0 & -0.623 & 1.215 & -1.800 & -1.000 & 1.000 \\
Normal & 30 & 0.5 & 0.518 & 0.604 & 0.000 & 0.800 & 1.000 \\
Normal & 30 & 0.0 & 0.527 & 0.716 & 0.000 & 1.000 & 1.000 \\
$t_{10}$ & 30 & 1.0 & -0.675 & 1.202 & -1.800 & -1.000 & 1.000 \\
$t_{10}$ & 30 & 0.5 & 0.502 & 0.616 & 0.000 & 0.700 & 1.000 \\
$t_{10}$ & 30 & 0.0 & 0.527 & 0.703 & 0.000 & 0.900 & 1.000 \\
Contam. & 30 & 1.0 & -0.984 & 1.062 & -1.900 & -1.400 & -0.100 \\
Contam. & 30 & 0.5 & 0.425 & 0.635 & 0.000 & 0.400 & 1.000 \\
Contam. & 30 & 0.0 & 0.474 & 0.669 & 0.000 & 0.700 & 1.000 \\
Hetero. $t_{10}$ & 30 & 1.0 & -0.910 & 1.183 & -1.900 & -1.500 & 0.300 \\
Hetero. $t_{10}$ & 30 & 0.5 & 0.426 & 0.648 & 0.000 & 0.500 & 1.000 \\
Hetero. $t_{10}$ & 30 & 0.0 & 0.522 & 0.592 & 0.000 & 0.700 & 1.000 \\
$X$-contam. & 30 & 1.0 & -1.054 & 1.071 & -1.900 & -1.500 & -0.600 \\
$X$-contam. & 30 & 0.5 & 0.422 & 0.630 & 0.000 & 0.400 & 1.000 \\
$X$-contam. & 30 & 0.0 & 0.509 & 0.610 & 0.000 & 0.700 & 1.000 \\
\hline
Normal & 50 & 1.0 & 0.150 & 1.151 & -1.000 & 1.000 & 1.000 \\
Normal & 50 & 0.5 & 0.727 & 0.528 & 0.400 & 1.000 & 1.000 \\
Normal & 50 & 0.0 & 0.725 & 0.655 & 0.500 & 1.000 & 1.100 \\
$t_{10}$ & 50 & 1.0 & 0.126 & 1.172 & -1.000 & 1.000 & 1.000 \\
$t_{10}$ & 50 & 0.5 & 0.732 & 0.521 & 0.400 & 1.000 & 1.000 \\
$t_{10}$ & 50 & 0.0 & 0.735 & 0.652 & 0.500 & 1.000 & 1.100 \\
Contam. & 50 & 1.0 & -0.249 & 1.277 & -1.600 & 0.300 & 1.000 \\
Contam. & 50 & 0.5 & 0.673 & 0.550 & 0.300 & 1.000 & 1.000 \\
Contam. & 50 & 0.0 & 0.727 & 0.569 & 0.400 & 1.000 & 1.100 \\
Hetero. $t_{10}$ & 50 & 1.0 & -0.095 & 1.230 & -1.500 & 0.500 & 1.000 \\
Hetero. $t_{10}$ & 50 & 0.5 & 0.624 & 0.502 & 0.200 & 0.900 & 1.000 \\
Hetero. $t_{10}$ & 50 & 0.0 & 0.648 & 0.566 & 0.300 & 1.000 & 1.000 \\
$X$-contam. & 50 & 1.0 & -0.383 & 1.276 & -1.700 & -0.400 & 1.000 \\
$X$-contam. & 50 & 0.5 & 0.609 & 0.535 & 0.200 & 1.000 & 1.000 \\
$X$-contam. & 50 & 0.0 & 0.667 & 0.550 & 0.300 & 1.000 & 1.000 \\
\hline
Normal & 100 & 1.0 & 0.486 & 0.924 & 0.500 & 1.000 & 1.000 \\
Normal & 100 & 0.5 & 0.858 & 0.439 & 0.900 & 1.000 & 1.000 \\
Normal & 100 & 0.0 & 0.852 & 0.565 & 0.800 & 1.000 & 1.100 \\
$t_{10}$ & 100 & 1.0 & 0.461 & 0.948 & 0.400 & 1.000 & 1.000 \\
$t_{10}$ & 100 & 0.5 & 0.862 & 0.420 & 0.900 & 1.000 & 1.000 \\
$t_{10}$ & 100 & 0.0 & 0.845 & 0.578 & 0.900 & 1.000 & 1.100 \\
Contam. & 100 & 1.0 & 0.278 & 1.117 & -1.000 & 1.000 & 1.000 \\
Contam. & 100 & 0.5 & 0.828 & 0.434 & 0.800 & 1.000 & 1.000 \\
Contam. & 100 & 0.0 & 0.829 & 0.538 & 0.800 & 1.000 & 1.100 \\
Hetero. $t_{10}$ & 100 & 1.0 & 0.375 & 0.956 & 0.200 & 1.000 & 1.000 \\
Hetero. $t_{10}$ & 100 & 0.5 & 0.766 & 0.451 & 0.600 & 1.000 & 1.000 \\
Hetero. $t_{10}$ & 100 & 0.0 & 0.740 & 0.583 & 0.600 & 1.000 & 1.000 \\
$X$-contam. & 100 & 1.0 & 0.130 & 1.169 & -1.100 & 1.000 & 1.000 \\
$X$-contam. & 100 & 0.5 & 0.740 & 0.449 & 0.500 & 1.000 & 1.000 \\
$X$-contam. & 100 & 0.0 & 0.736 & 0.553 & 0.600 & 1.000 & 1.000 \\
\end{longtable}
\vspace{-2em}
\begin{center}
\begin{minipage}{0.6\linewidth}
\footnotesize Note: The table summarizes the Monte Carlo distribution of the selected curvature parameter $\hat\gamma(\tau)$. The candidate set is $\Gamma=[-2,2]$ with grid spacing $0.1$.
\end{minipage}
\end{center}
\vspace{0.5em}
\normalsize

\footnotesize
\begin{longtable}{lrrrrrrrr}
\caption{System-risk, multiplier, and weight-stability diagnostics}
\label{tab:mc_system_diagnostics}\\
\toprule
DGP & $n$ & $\tau$ & $\widehat C_{\theta}$ & $\widehat C_{\lambda}$ & $\widehat C^{\mathrm{sys}}$ & $\|\hat\lambda\|$ & $\hat\lambda^{\prime}\widehat\Omega\hat\lambda$ & $ESS/n$ \\
\midrule
\endfirsthead

\multicolumn{9}{c}{\tablename\ \thetable{} -- continued from previous page}\\
\toprule
DGP & $n$ & $\tau$ & $\widehat C_{\theta}$ & $\widehat C_{\lambda}$ & $\widehat C^{\mathrm{sys}}$ & $\|\hat\lambda\|$ & $\hat\lambda^{\prime}\widehat\Omega\hat\lambda$ & $ESS/n$ \\
\midrule
\endhead

\midrule
\multicolumn{9}{r}{Continued on next page}\\
\endfoot

\bottomrule
\endlastfoot

Normal & 30 & 1.0 & 0.000 & 0.519 & 0.000 & 4.546 & 480.914 & 0.822 \\
Normal & 30 & 0.5 & 0.038 & 0.046 & 0.042 & 1.745 & 71.801 & 0.914 \\
Normal & 30 & 0.0 & 0.309 & 0.000 & 0.000 & 0.992 & 10.102 & 0.915 \\
$t_{10}$ & 30 & 1.0 & 0.000 & 0.539 & 0.000 & 4.939 & 522.315 & 0.818 \\
$t_{10}$ & 30 & 0.5 & 0.037 & 0.047 & 0.042 & 1.865 & 105.250 & 0.913 \\
$t_{10}$ & 30 & 0.0 & 0.310 & 0.000 & 0.000 & 1.045 & 18.154 & 0.915 \\
Contam. & 30 & 1.0 & 0.000 & 0.638 & 0.000 & 9.705 & 2924.035 & 0.784 \\
Contam. & 30 & 0.5 & 0.043 & 0.053 & 0.048 & 3.544 & 1016.443 & 0.900 \\
Contam. & 30 & 0.0 & 0.341 & 0.000 & 0.000 & 1.828 & 35.943 & 0.907 \\
Hetero. $t_{10}$ & 30 & 1.0 & 0.000 & 0.607 & 0.000 & 5.729 & 1596.393 & 0.763 \\
Hetero. $t_{10}$ & 30 & 0.5 & 0.031 & 0.037 & 0.034 & 1.940 & 257.593 & 0.904 \\
Hetero. $t_{10}$ & 30 & 0.0 & 0.263 & 0.000 & 0.000 & 1.123 & 29.782 & 0.905 \\
$X$-contam. & 30 & 1.0 & 0.000 & 0.662 & 0.000 & 10.019 & 2872.326 & 0.759 \\
$X$-contam. & 30 & 0.5 & 0.032 & 0.043 & 0.038 & 3.277 & 342.657 & 0.900 \\
$X$-contam. & 30 & 0.0 & 0.308 & 0.000 & 0.000 & 1.763 & 39.063 & 0.899 \\
\midrule
Normal & 50 & 1.0 & 0.000 & 0.278 & 0.000 & 0.600 & 2.444 & 0.923 \\
Normal & 50 & 0.5 & 0.038 & 0.031 & 0.034 & 0.364 & 0.489 & 0.954 \\
Normal & 50 & 0.0 & 0.170 & 0.000 & 0.000 & 0.303 & 0.218 & 0.951 \\
$t_{10}$ & 50 & 1.0 & 0.000 & 0.289 & 0.000 & 0.624 & 2.610 & 0.922 \\
$t_{10}$ & 50 & 0.5 & 0.037 & 0.032 & 0.035 & 0.372 & 0.503 & 0.954 \\
$t_{10}$ & 50 & 0.0 & 0.174 & 0.000 & 0.000 & 0.307 & 0.223 & 0.951 \\
Contam. & 50 & 1.0 & 0.000 & 0.409 & 0.000 & 1.081 & 10.622 & 0.903 \\
Contam. & 50 & 0.5 & 0.043 & 0.036 & 0.040 & 0.560 & 1.537 & 0.950 \\
Contam. & 50 & 0.0 & 0.191 & 0.000 & 0.000 & 0.425 & 0.525 & 0.950 \\
Hetero. $t_{10}$ & 50 & 1.0 & 0.000 & 0.351 & 0.000 & 0.718 & 8.016 & 0.896 \\
Hetero. $t_{10}$ & 50 & 0.5 & 0.039 & 0.024 & 0.031 & 0.360 & 1.441 & 0.950 \\
Hetero. $t_{10}$ & 50 & 0.0 & 0.164 & 0.000 & 0.000 & 0.315 & 0.651 & 0.945 \\
$X$-contam. & 50 & 1.0 & 0.000 & 0.459 & 0.000 & 1.217 & 26.448 & 0.884 \\
$X$-contam. & 50 & 0.5 & 0.042 & 0.031 & 0.036 & 0.546 & 3.212 & 0.949 \\
$X$-contam. & 50 & 0.0 & 0.209 & 0.000 & 0.000 & 0.447 & 1.016 & 0.942 \\
\midrule
Normal & 100 & 1.0 & 0.000 & 0.157 & 0.000 & 0.196 & 0.090 & 0.972 \\
Normal & 100 & 0.5 & 0.018 & 0.016 & 0.017 & 0.174 & 0.048 & 0.977 \\
Normal & 100 & 0.0 & 0.076 & 0.000 & 0.000 & 0.168 & 0.040 & 0.976 \\
$t_{10}$ & 100 & 1.0 & 0.000 & 0.169 & 0.000 & 0.206 & 0.109 & 0.971 \\
$t_{10}$ & 100 & 0.5 & 0.019 & 0.016 & 0.017 & 0.178 & 0.053 & 0.978 \\
$t_{10}$ & 100 & 0.0 & 0.082 & 0.000 & 0.000 & 0.171 & 0.041 & 0.976 \\
Contam. & 100 & 1.0 & 0.000 & 0.231 & 0.000 & 0.292 & 0.360 & 0.963 \\
Contam. & 100 & 0.5 & 0.025 & 0.017 & 0.021 & 0.224 & 0.094 & 0.976 \\
Contam. & 100 & 0.0 & 0.096 & 0.000 & 0.000 & 0.213 & 0.056 & 0.975 \\
Hetero. $t_{10}$ & 100 & 1.0 & 0.000 & 0.173 & 0.000 & 0.183 & 0.289 & 0.965 \\
Hetero. $t_{10}$ & 100 & 0.5 & 0.017 & 0.015 & 0.016 & 0.147 & 0.102 & 0.976 \\
Hetero. $t_{10}$ & 100 & 0.0 & 0.082 & 0.000 & 0.000 & 0.143 & 0.080 & 0.974 \\
$X$-contam. & 100 & 1.0 & 0.000 & 0.285 & 0.000 & 0.309 & 1.080 & 0.953 \\
$X$-contam. & 100 & 0.5 & 0.029 & 0.020 & 0.025 & 0.203 & 0.185 & 0.976 \\
$X$-contam. & 100 & 0.0 & 0.134 & 0.000 & 0.000 & 0.200 & 0.102 & 0.972 \\
\end{longtable}
\vspace{-2.5em}
\begin{center}
\begin{minipage}{0.7\linewidth}
\footnotesize Note: The reported $\widehat C_{\theta}$, $\widehat C_{\lambda}$, and $\widehat C^{\mathrm{sys}}$ are normalized components evaluated at the selected $\hat\gamma(\tau)$. $ESS/n$ is the effective sample size divided by the nominal sample size.
\end{minipage}
\end{center}
\vspace{0.5em}
\normalsize

\scriptsize
\begingroup
\footnotesize
\setlength{\tabcolsep}{3.5pt}
\renewcommand{\arraystretch}{1.05}

\begin{longtable}{lrrrrrrrr}
\caption{Raw risk components and selected-weight diagnostics}
\label{tab:mc_raw_risk_weight_diagnostics}\\
\toprule
DGP & $n$ & $\tau$ & $C_{\theta,\mathrm{raw}}$ & $C_{\lambda,\mathrm{raw}}$ & $\min_i\hat\pi_i$ & $\mathrm{Med.}(\hat\pi_i)$ & $\max_i\hat\pi_i$ & $ESS$ \\
\midrule
\endfirsthead

\multicolumn{9}{c}{\tablename\ \thetable{} -- continued from previous page}\\
\toprule
DGP & $n$ & $\tau$ & $C_{\theta,\mathrm{raw}}$ & $C_{\lambda,\mathrm{raw}}$ & $\min_i\hat\pi_i$ & $\mathrm{Med.}(\hat\pi_i)$ & $\max_i\hat\pi_i$ & $ESS$ \\
\midrule
\endhead

\midrule
\multicolumn{9}{r}{Continued on next page}\\
\endfoot

\bottomrule
\endlastfoot

Normal & 30 & 1.0 & $3.09\times 10^{0}$ & $4.30\times 10^{5}$ & $7.42\times 10^{-3}$ & $3.11\times 10^{-2}$ & $8.29\times 10^{-2}$ & 24.67 \\
Normal & 30 & 0.5 & $9.72\times 10^{0}$ & $6.36\times 10^{4}$ & $5.58\times 10^{-3}$ & $3.44\times 10^{-2}$ & $5.30\times 10^{-2}$ & 27.42 \\
Normal & 30 & 0.0 & $5.45\times 10^{1}$ & $8.71\times 10^{3}$ & $8.22\times 10^{-3}$ & $3.42\times 10^{-2}$ & $5.47\times 10^{-2}$ & 27.45 \\
$t_{10}$ & 30 & 1.0 & $3.52\times 10^{0}$ & $4.67\times 10^{5}$ & $7.59\times 10^{-3}$ & $3.10\times 10^{-2}$ & $8.50\times 10^{-2}$ & 24.54 \\
$t_{10}$ & 30 & 0.5 & $9.37\times 10^{0}$ & $9.37\times 10^{4}$ & $5.70\times 10^{-3}$ & $3.44\times 10^{-2}$ & $5.34\times 10^{-2}$ & 27.40 \\
$t_{10}$ & 30 & 0.0 & $5.71\times 10^{1}$ & $1.59\times 10^{4}$ & $8.20\times 10^{-3}$ & $3.42\times 10^{-2}$ & $5.47\times 10^{-2}$ & 27.45 \\
Contam. & 30 & 1.0 & $5.12\times 10^{0}$ & $2.63\times 10^{6}$ & $7.50\times 10^{-3}$ & $3.04\times 10^{-2}$ & $9.66\times 10^{-2}$ & 23.53 \\
Contam. & 30 & 0.5 & $1.32\times 10^{1}$ & $9.13\times 10^{5}$ & $5.39\times 10^{-3}$ & $3.44\times 10^{-2}$ & $5.66\times 10^{-2}$ & 26.99 \\
Contam. & 30 & 0.0 & $7.59\times 10^{1}$ & $3.17\times 10^{4}$ & $7.54\times 10^{-3}$ & $3.43\times 10^{-2}$ & $5.62\times 10^{-2}$ & 27.20 \\
Hetero. $t_{10}$ & 30 & 1.0 & $1.48\times 10^{0}$ & $1.43\times 10^{6}$ & $4.87\times 10^{-3}$ & $2.98\times 10^{-2}$ & $1.00\times 10^{-1}$ & 22.88 \\
Hetero. $t_{10}$ & 30 & 0.5 & $6.18\times 10^{0}$ & $2.30\times 10^{5}$ & $3.18\times 10^{-3}$ & $3.47\times 10^{-2}$ & $5.49\times 10^{-2}$ & 27.11 \\
Hetero. $t_{10}$ & 30 & 0.0 & $5.78\times 10^{1}$ & $2.61\times 10^{4}$ & $4.12\times 10^{-3}$ & $3.48\times 10^{-2}$ & $5.44\times 10^{-2}$ & 27.16 \\
$X$-contam. & 30 & 1.0 & $2.51\times 10^{0}$ & $2.58\times 10^{6}$ & $6.03\times 10^{-3}$ & $2.98\times 10^{-2}$ & $1.04\times 10^{-1}$ & 22.76 \\
$X$-contam. & 30 & 0.5 & $1.01\times 10^{1}$ & $3.06\times 10^{5}$ & $4.06\times 10^{-3}$ & $3.45\times 10^{-2}$ & $5.59\times 10^{-2}$ & 27.00 \\
$X$-contam. & 30 & 0.0 & $1.05\times 10^{2}$ & $3.45\times 10^{4}$ & $5.32\times 10^{-3}$ & $3.44\times 10^{-2}$ & $5.63\times 10^{-2}$ & 26.96 \\

\midrule

Normal & 50 & 1.0 & $5.34\times 10^{0}$ & $5.39\times 10^{3}$ & $6.53\times 10^{-3}$ & $1.99\times 10^{-2}$ & $3.78\times 10^{-2}$ & 46.14 \\
Normal & 50 & 0.5 & $1.32\times 10^{1}$ & $8.66\times 10^{2}$ & $5.66\times 10^{-3}$ & $2.04\times 10^{-2}$ & $2.94\times 10^{-2}$ & 47.69 \\
Normal & 50 & 0.0 & $4.47\times 10^{1}$ & $3.10\times 10^{2}$ & $6.12\times 10^{-3}$ & $2.04\times 10^{-2}$ & $3.07\times 10^{-2}$ & 47.53 \\
$t_{10}$ & 50 & 1.0 & $5.17\times 10^{0}$ & $5.79\times 10^{3}$ & $6.58\times 10^{-3}$ & $1.99\times 10^{-2}$ & $3.81\times 10^{-2}$ & 46.11 \\
$t_{10}$ & 50 & 0.5 & $1.29\times 10^{1}$ & $9.02\times 10^{2}$ & $5.66\times 10^{-3}$ & $2.04\times 10^{-2}$ & $2.95\times 10^{-2}$ & 47.71 \\
$t_{10}$ & 50 & 0.0 & $4.43\times 10^{1}$ & $3.28\times 10^{2}$ & $6.13\times 10^{-3}$ & $2.04\times 10^{-2}$ & $3.07\times 10^{-2}$ & 47.56 \\
Contam. & 50 & 1.0 & $6.50\times 10^{0}$ & $2.53\times 10^{4}$ & $6.63\times 10^{-3}$ & $1.96\times 10^{-2}$ & $4.40\times 10^{-2}$ & 45.17 \\
Contam. & 50 & 0.5 & $1.59\times 10^{1}$ & $3.33\times 10^{3}$ & $5.16\times 10^{-3}$ & $2.05\times 10^{-2}$ & $3.04\times 10^{-2}$ & 47.51 \\
Contam. & 50 & 0.0 & $5.50\times 10^{1}$ & $1.04\times 10^{3}$ & $5.45\times 10^{-3}$ & $2.05\times 10^{-2}$ & $3.08\times 10^{-2}$ & 47.50 \\
Hetero. $t_{10}$ & 50 & 1.0 & $2.76\times 10^{0}$ & $1.87\times 10^{4}$ & $4.69\times 10^{-3}$ & $1.97\times 10^{-2}$ & $4.39\times 10^{-2}$ & 44.82 \\
Hetero. $t_{10}$ & 50 & 0.5 & $9.28\times 10^{0}$ & $3.06\times 10^{3}$ & $3.63\times 10^{-3}$ & $2.06\times 10^{-2}$ & $2.97\times 10^{-2}$ & 47.50 \\
Hetero. $t_{10}$ & 50 & 0.0 & $3.71\times 10^{1}$ & $1.22\times 10^{3}$ & $3.86\times 10^{-3}$ & $2.06\times 10^{-2}$ & $3.09\times 10^{-2}$ & 47.26 \\
$X$-contam. & 50 & 1.0 & $3.99\times 10^{0}$ & $6.41\times 10^{4}$ & $5.31\times 10^{-3}$ & $1.95\times 10^{-2}$ & $4.90\times 10^{-2}$ & 44.18 \\
$X$-contam. & 50 & 0.5 & $1.27\times 10^{1}$ & $7.34\times 10^{3}$ & $3.94\times 10^{-3}$ & $2.05\times 10^{-2}$ & $3.04\times 10^{-2}$ & 47.46 \\
$X$-contam. & 50 & 0.0 & $6.36\times 10^{1}$ & $2.16\times 10^{3}$ & $4.18\times 10^{-3}$ & $2.05\times 10^{-2}$ & $3.17\times 10^{-2}$ & 47.11 \\

\midrule

Normal & 100 & 1.0 & $5.65\times 10^{0}$ & $4.53\times 10^{2}$ & $3.85\times 10^{-3}$ & $1.01\times 10^{-2}$ & $1.57\times 10^{-2}$ & 97.17 \\
Normal & 100 & 0.5 & $1.13\times 10^{1}$ & $1.33\times 10^{2}$ & $3.59\times 10^{-3}$ & $1.01\times 10^{-2}$ & $1.43\times 10^{-2}$ & 97.73 \\
Normal & 100 & 0.0 & $2.90\times 10^{1}$ & $8.87\times 10^{1}$ & $3.65\times 10^{-3}$ & $1.01\times 10^{-2}$ & $1.46\times 10^{-2}$ & 97.61 \\
$t_{10}$ & 100 & 1.0 & $6.39\times 10^{0}$ & $6.07\times 10^{2}$ & $3.91\times 10^{-3}$ & $1.01\times 10^{-2}$ & $1.58\times 10^{-2}$ & 97.12 \\
$t_{10}$ & 100 & 0.5 & $1.25\times 10^{1}$ & $1.66\times 10^{2}$ & $3.63\times 10^{-3}$ & $1.01\times 10^{-2}$ & $1.42\times 10^{-2}$ & 97.75 \\
$t_{10}$ & 100 & 0.0 & $3.30\times 10^{1}$ & $1.00\times 10^{2}$ & $3.71\times 10^{-3}$ & $1.01\times 10^{-2}$ & $1.47\times 10^{-2}$ & 97.60 \\
Contam. & 100 & 1.0 & $5.87\times 10^{0}$ & $2.86\times 10^{3}$ & $3.77\times 10^{-3}$ & $1.00\times 10^{-2}$ & $1.80\times 10^{-2}$ & 96.33 \\
Contam. & 100 & 0.5 & $1.46\times 10^{1}$ & $4.98\times 10^{2}$ & $3.23\times 10^{-3}$ & $1.01\times 10^{-2}$ & $1.47\times 10^{-2}$ & 97.65 \\
Contam. & 100 & 0.0 & $4.08\times 10^{1}$ & $2.16\times 10^{2}$ & $3.32\times 10^{-3}$ & $1.01\times 10^{-2}$ & $1.50\times 10^{-2}$ & 97.50 \\
Hetero. $t_{10}$ & 100 & 1.0 & $3.69\times 10^{0}$ & $2.14\times 10^{3}$ & $2.88\times 10^{-3}$ & $1.01\times 10^{-2}$ & $1.72\times 10^{-2}$ & 96.49 \\
Hetero. $t_{10}$ & 100 & 0.5 & $8.96\times 10^{0}$ & $5.16\times 10^{2}$ & $2.62\times 10^{-3}$ & $1.02\times 10^{-2}$ & $1.45\times 10^{-2}$ & 97.65 \\
Hetero. $t_{10}$ & 100 & 0.0 & $3.23\times 10^{1}$ & $3.57\times 10^{2}$ & $2.71\times 10^{-3}$ & $1.01\times 10^{-2}$ & $1.51\times 10^{-2}$ & 97.38 \\
$X$-contam. & 100 & 1.0 & $4.31\times 10^{0}$ & $9.56\times 10^{3}$ & $3.01\times 10^{-3}$ & $1.00\times 10^{-2}$ & $2.02\times 10^{-2}$ & 95.35 \\
$X$-contam. & 100 & 0.5 & $1.39\times 10^{1}$ & $1.26\times 10^{3}$ & $2.47\times 10^{-3}$ & $1.02\times 10^{-2}$ & $1.48\times 10^{-2}$ & 97.60 \\
$X$-contam. & 100 & 0.0 & $5.77\times 10^{1}$ & $5.62\times 10^{2}$ & $2.58\times 10^{-3}$ & $1.01\times 10^{-2}$ & $1.55\times 10^{-2}$ & 97.20 \\

\end{longtable}

\endgroup
\vspace{-2.5em}
\begin{center}
\begin{minipage}{0.9\linewidth}
\footnotesize Note: This table reports unnormalized risk components and selected-weight diagnostics. Because raw components can be on different scales, the selection criterion combines their min--max normalized versions.
\end{minipage}
\end{center}
\vspace{0.5em}
\normalsize

\footnotesize
\begin{longtable}{lrrrrrrrrr}
\caption{Balanced selector relative to the two benchmark selectors}
\label{tab:mc_tau05_comparison}\\
\toprule
DGP & $n$ & $\Delta\gamma_{.5-1}$ & $\Delta\gamma_{.5-0}$ & $\Delta RMSE_{.5-1}$ & $\Delta RMSE_{.5-0}$ & $\Delta Cov_{.5-1}$ & $\Delta Cov_{.5-0}$ & $\Delta\|\lambda\|_{.5-1}$ & $\Delta\|\lambda\|_{.5-0}$ \\
\midrule
\endfirsthead
\multicolumn{10}{c}{\tablename\ \thetable{} -- continued from previous page}\\
\toprule
DGP & $n$ & $\Delta\gamma_{.5-1}$ & $\Delta\gamma_{.5-0}$ & $\Delta RMSE_{.5-1}$ & $\Delta RMSE_{.5-0}$ & $\Delta Cov_{.5-1}$ & $\Delta Cov_{.5-0}$ & $\Delta\|\lambda\|_{.5-1}$ & $\Delta\|\lambda\|_{.5-0}$ \\
\midrule
\endhead
\midrule
\multicolumn{10}{r}{Continued on next page}\\
\endfoot

\bottomrule
\endlastfoot
$X$-contam. & 30 & 1.476 & -0.086 & -0.007 & 0.002 & 0.006 & 0.039 & -6.742 & 1.514 \\
$t_{10}$ & 30 & 1.178 & -0.025 & -0.011 & 0.008 & 0.009 & 0.058 & -3.074 & 0.820 \\
Contam. & 30 & 1.409 & -0.049 & -0.016 & 0.018 & 0.014 & 0.036 & -6.162 & 1.716 \\
Hetero. $t_{10}$ & 30 & 1.335 & -0.097 & -0.004 & -0.022 & 0.003 & 0.056 & -3.789 & 0.817 \\
Normal & 30 & 1.141 & -0.008 & -0.013 & 0.008 & 0.012 & 0.061 & -2.801 & 0.753 \\
\hline
$X$-contam. & 50 & 0.992 & -0.058 & -0.007 & -0.010 & 0.009 & 0.028 & -0.671 & 0.099 \\
$t_{10}$ & 50 & 0.606 & -0.003 & -0.006 & -0.001 & 0.010 & 0.011 & -0.252 & 0.065 \\
Contam. & 50 & 0.922 & -0.054 & -0.010 & 0.001 & 0.015 & 0.010 & -0.521 & 0.135 \\
Hetero. $t_{10}$ & 50 & 0.720 & -0.023 & 0.001 & -0.015 & -0.005 & 0.029 & -0.358 & 0.045 \\
Normal & 50 & 0.577 & 0.001 & -0.006 & 0.000 & 0.007 & 0.012 & -0.235 & 0.062 \\
\hline
$X$-contam. & 100 & 0.609 & 0.003 & -0.000 & -0.007 & -0.000 & 0.018 & -0.106 & 0.003 \\
$t_{10}$ & 100 & 0.400 & 0.016 & -0.001 & -0.001 & 0.003 & 0.002 & -0.028 & 0.007 \\
Contam. & 100 & 0.550 & -0.001 & -0.002 & -0.002 & 0.003 & 0.006 & -0.068 & 0.011 \\
Hetero. $t_{10}$ & 100 & 0.390 & 0.026 & 0.001 & -0.007 & -0.001 & 0.015 & -0.036 & 0.004 \\
Normal & 100 & 0.372 & 0.006 & -0.001 & -0.000 & 0.002 & 0.002 & -0.023 & 0.006 \\
\end{longtable}
\vspace{-3.5em}
\begin{center}
\begin{minipage}{1.1\linewidth}
\footnotesize Note: For any statistic $M$, 
$\Delta M_{0.5-b}=M(\tau=0.5)-M(\tau=b)$, where $b\in\{1,0\}$.
``Cov.'' denotes the empirical coverage probability of the nominal 95\% confidence
interval for the slope parameter $\beta_1$; hence a positive value of
$\Delta \mathrm{Cov}_{0.5-b}$ means that the balanced selector has higher coverage
than benchmark $b$. The term $\|\hat\lambda\|$ denotes the Euclidean norm of the
CRPD Lagrange multiplier evaluated at the selected $\hat\gamma(\tau)$. Larger values
of $\|\hat\lambda\|$ indicate stronger moment-enforcement pressure and potentially
less stable implied weights. Therefore, a negative value of
$\Delta\|\hat\lambda\|_{0.5-b}$ means that the balanced selector lowers the multiplier
norm relative to benchmark $b$.
\end{minipage}
\end{center}
\vspace{0.5em}
\normalsize

\small
\begin{longtable}{lrrrrrrr}
\caption{DGP sensitivity of the selected estimator at $n=100$}
\label{tab:mc_n100_dgp_sensitivity}\\
\toprule
DGP & $\tau$ & $\bar{\hat\gamma}$ & RMSE & Coverage & SD/SE & $\|\hat\lambda\|$ & $ESS/n$ \\
\midrule
\endfirsthead
\multicolumn{8}{c}{\tablename\ \thetable{} -- continued from previous page}\\
\toprule
DGP & $\tau$ & $\bar{\hat\gamma}$ & RMSE & Coverage & SD/SE & $\|\hat\lambda\|$ & $ESS/n$ \\
\midrule
\endhead
\midrule
\multicolumn{8}{r}{Continued on next page}\\
\endfoot

\bottomrule
\endlastfoot
Normal & 1.0 & 0.486 & 0.111 & 0.886 & 1.212 & 0.196 & 0.972 \\
Normal & 0.5 & 0.858 & 0.109 & 0.888 & 1.202 & 0.174 & 0.977 \\
Normal & 0.0 & 0.852 & 0.109 & 0.887 & 1.205 & 0.168 & 0.976 \\
$t_{10}$ & 1.0 & 0.461 & 0.108 & 0.890 & 1.205 & 0.206 & 0.971 \\
$t_{10}$ & 0.5 & 0.862 & 0.106 & 0.893 & 1.193 & 0.178 & 0.978 \\
$t_{10}$ & 0.0 & 0.845 & 0.108 & 0.891 & 1.207 & 0.171 & 0.976 \\
Contam. & 1.0 & 0.278 & 0.093 & 0.886 & 1.247 & 0.292 & 0.963 \\
Contam. & 0.5 & 0.828 & 0.091 & 0.889 & 1.222 & 0.224 & 0.976 \\
Contam. & 0.0 & 0.829 & 0.093 & 0.884 & 1.250 & 0.213 & 0.975 \\
Hetero. $t_{10}$ & 1.0 & 0.375 & 0.137 & 0.899 & 1.156 & 0.183 & 0.965 \\
Hetero. $t_{10}$ & 0.5 & 0.766 & 0.138 & 0.898 & 1.164 & 0.147 & 0.976 \\
Hetero. $t_{10}$ & 0.0 & 0.740 & 0.145 & 0.883 & 1.220 & 0.143 & 0.974 \\
$X$-contam. & 1.0 & 0.130 & 0.124 & 0.901 & 1.218 & 0.309 & 0.953 \\
$X$-contam. & 0.5 & 0.740 & 0.124 & 0.901 & 1.217 & 0.203 & 0.976 \\
$X$-contam. & 0.0 & 0.736 & 0.131 & 0.883 & 1.287 & 0.200 & 0.972 \\
\end{longtable}
\vspace{-1.4em}
\begin{center}
\begin{minipage}{0.8\linewidth}
\footnotesize Note: This table isolates the larger-sample case to show how the selected curvature parameter varies with the DGP and with the system-risk weight.
\end{minipage}
\end{center}
\vspace{0.5em}
\normalsize

\section{Empirical Application}

To illustrate the proposed method in a simple empirical setting, we use the
dataset from Owen (2001, Table 3.2, p. 44), which reports the pounds of milk
produced and the number of days milked in 1936 for $n=22$ dairy cows. The
parameter of interest is the mean daily milk yield,
\[
    \theta = \mu = E[X_i],
    \qquad
    X_i = \frac{\textit{Milk}_i}{\textit{Days}_i}.
\]

Table \ref{tab:mpd_desc} reports descriptive statistics for milk production per
day. The sample average is 12.43 lb/day, with a standard deviation of 3.08. The
minimum and maximum values are 7.55 and 18.66, respectively, indicating
substantial heterogeneity in daily milk productivity across cows.

\begin{table}[!htbp] 
\centering 
\caption{Descriptive statistics for milk production per day} 
\label{tab:mpd_desc} 
\begin{tabular}{@{\extracolsep{5pt}}cccccc} 
\toprule
Statistic 
& \multicolumn{1}{c}{$n$} 
& \multicolumn{1}{c}{Mean} 
& \multicolumn{1}{c}{St.\ Dev.} 
& \multicolumn{1}{c}{Minimum} 
& \multicolumn{1}{c}{Maximum} \\ 
\midrule
mpd & 22 & 12.4250 & 3.0750 & 7.5470 & 18.6610 \\ 
\bottomrule
\multicolumn{6}{l}{%
\begin{minipage}{11cm}\vspace{1mm}
    \footnotesize 
    Note: mpd denotes milk production per day (lb/day). $n$ is the number of
    observations. St.\ Dev.\ denotes the standard deviation.
\end{minipage}}     
\end{tabular} 
\end{table} 

\paragraph{Empirical specification.}
Let $X_i=\textit{mpd}_i$ denote milk production per day for cow $i$. The
parameter of interest is the population mean
\[
    \mu = E[X_i],
\]
with the baseline identifying moment condition
\[
    E[X_i-\mu]=0.
\]
With only this single moment condition, however, the empirical mean satisfies the
sample analogue exactly under uniform weights. Hence, the CRPD problem does not
generate a nontrivial overidentified reweighting problem.

To obtain an overidentified moment system, we introduce an auxiliary moment based
on the number of days milked. Let $D_i=\textit{Days}_i$ and define the
standardized variable
\[
    \widetilde D_i = \frac{D_i-\bar D}{s_D},
\]
where $\bar D$ and $s_D$ are the sample mean and sample standard deviation of
$D_i$. The moment vector is
\[
    g_i(\mu)
    =
    Z_i(X_i-\mu),
    \qquad
    Z_i =
    \begin{pmatrix}
        1 \\
        \widetilde D_i
    \end{pmatrix}.
\]
Equivalently,
\[
    g_i(\mu)
    =
    \begin{pmatrix}
        X_i-\mu \\
        (X_i-\mu)\widetilde D_i
    \end{pmatrix}.
\]
The second moment requires deviations of daily milk production from the
population mean not to vary systematically with the number of days milked. A
Pearson correlation test between $X_i$ and $D_i$ yields a $p$-value of 0.32,
suggesting that this auxiliary moment is not strongly contradicted by the sample.
The resulting system has two moments for one structural parameter, so that
$q=2>p=1$.

\paragraph{Data-driven selection of $\gamma$.}
The CRPD estimator depends on the power parameter $\gamma$, which controls the
curvature of the divergence and therefore the implied probability weights. Rather
than selecting $\gamma$ by an out-of-sample prediction criterion, we select
$\gamma$ by the system-risk criterion used in the main simulation design.

For each candidate value of $\gamma$ on the grid
\[
    \mathcal G = \{-2.0,-1.9,\ldots,2.0\},
\]
the CRPD estimator is computed subject to the probability and moment restrictions.
Let $\hat\theta_{\gamma}$ and $\hat\lambda_{\gamma}$ denote the resulting
structural estimate and multiplier. We compare these quantities with a
first-order GMM benchmark. Let $\hat\theta_{\mathrm{GMM}}$ denote the GMM
estimate and let $\hat\lambda_{\mathrm{FO}}$ denote the corresponding first-order
multiplier approximation. Define
\[
    B_{\theta}(\gamma)
    =
    n\bigl(\hat\theta_{\gamma}-\hat\theta_{\mathrm{GMM}}\bigr),
    \qquad
    B_{\lambda}(\gamma)
    =
    n\bigl(\hat\lambda_{\gamma}-\hat\lambda_{\mathrm{FO}}\bigr).
\]
The structural and multiplier components are computed as
\[
    C_{\theta}^{\mathrm{raw}}(\gamma)
    =
    B_{\theta}(\gamma)^{\prime} \widehat K B_{\theta}(\gamma),
    \qquad
    C_{\lambda}^{\mathrm{raw}}(\gamma)
    =
    B_{\lambda}(\gamma)^{\prime} \widehat\Omega B_{\lambda}(\gamma),
\]
where $\widehat K$ and $\widehat\Omega$ are feasible weighting matrices. These
raw components are normalized over the feasible gamma path to obtain
$C_{\theta}(\gamma)$ and $C_{\lambda}(\gamma)$. The selected value of $\gamma$ is
then
\[
    \hat\gamma(\tau)
    =
    \arg\min_{\gamma\in\mathcal G_{\mathrm{feas}}}
    C_{\mathrm{sys}}(\gamma;\tau),
\]
where
\[
    C_{\mathrm{sys}}(\gamma;\tau)
    =
    \tau C_{\theta}(\gamma)
    +
    (1-\tau)C_{\lambda}(\gamma).
\]
The tuning parameter $\tau\in[0,1]$ determines the relative weight placed on
structural accuracy versus multiplier stability. When $\tau=1$, the selection
criterion uses only the structural component. When $\tau=0$, it uses only the
multiplier-stability component. Intermediate values balance the two components.

\paragraph{Estimation results.}
Table \ref{tab:empirics_selected_gamma} reports the selected CRPD estimates for
$\tau\in\{1,0.5,0\}$. The first-order GMM benchmark is
\[
    \hat\theta_{\mathrm{GMM}} = 12.8264,
\]
with a GMM criterion value of 0.0724. When $\tau=1$, the selector chooses
$\hat\gamma=0.7$, yielding $\hat\theta_{\gamma}=12.8264$. This coincides with the
GMM benchmark up to the reported precision, so the normalized structural
component is zero. When $\tau=0.5$ or $\tau=0$, the selector chooses
$\hat\gamma=1.4$, yielding $\hat\theta_{\gamma}=12.8291$. Thus, the point
estimate is highly stable across the selected curvature choices.

\begin{table}[!htbp]
\centering
\caption{Selected CRPD estimates and diagnostics}
\label{tab:empirics_selected_gamma}
\begin{tabular}{ccccccccc}
\toprule
$\tau$
& $\hat\gamma$
& $\hat\theta$
& SE
& $C_{\theta}$
& $C_{\lambda}$
& $\|\hat\lambda\|$
& $ESS/n$
& $(\min_i\hat\pi_i,\max_i\hat\pi_i)$ \\
\midrule
1.0 & 0.7 & 12.8264 & 0.6025 & 0.000000 & 0.000547 & 0.1138 & 0.9271 & $(0.0158,0.0767)$ \\
0.5 & 1.4 & 12.8291 & 0.6025 & 0.000007 & 0.000000 & 0.1047 & 0.9265 & $(0.0072,0.0708)$ \\
0.0 & 1.4 & 12.8291 & 0.6025 & 0.000007 & 0.000000 & 0.1047 & 0.9265 & $(0.0072,0.0708)$ \\
\bottomrule
\multicolumn{9}{l}{%
\begin{minipage}{15cm}\vspace{1mm}
    \footnotesize
    Note: $\tau$ controls the weight placed on the structural component in
    $C_{\mathrm{sys}}(\gamma;\tau)$. $C_{\theta}$ and $C_{\lambda}$ are the
    normalized structural and multiplier-stability components. $ESS/n$ denotes
    the effective sample size divided by the nominal sample size. The uniform
    probability benchmark is $1/n=0.0455$.
\end{minipage}}
\end{tabular}
\end{table}

The main difference across $\tau$ is therefore not the estimated mean but the
criterion used to choose the curvature parameter. The structural-only selector
chooses $\gamma=0.7$ because this value minimizes the structural discrepancy from
the first-order GMM benchmark. In contrast, the multiplier-stability selector
chooses $\gamma=1.4$, because this value minimizes the multiplier component over
the feasible gamma path. The balanced criterion with $\tau=0.5$ also selects
$\gamma=1.4$, indicating that the small increase in structural discrepancy is
more than offset by the gain in multiplier stability.

Figure \ref{fig:owen_gamma_path} plots the normalized gamma-path criterion for
each value of $\tau$. The criterion is large for highly negative values of
$\gamma$, especially near the left end of the grid. This indicates that those
curvature choices require substantially larger structural or multiplier
adjustments to enforce the moment restrictions. For moderate and positive values
of $\gamma$, the criterion is close to zero, and the estimator is much more
stable. The selected value is marked by the vertical dashed line in each panel.

Figure \ref{fig:owen_pi_hist_selected} reports the implied probability weights
for the selected specifications. The dashed vertical line denotes the uniform
benchmark $1/n$. Across the selected values of $\gamma$, most weights remain
close to the uniform benchmark. The effective sample size remains approximately
20.4, or about 93\% of the nominal sample size. This indicates that the selected
CRPD estimators do not rely on a small subset of observations to satisfy the
moment restrictions.

This empirical illustration highlights the role of $\gamma$ as a curvature
parameter in the CRPD estimator. The estimated mean daily milk yield is robust
across the selected values of $\gamma$, indicating that the structural estimate is
not materially affected by the curvature choice in this dataset. However, the
implied multipliers and probability weights vary across selected values of
$\gamma$, showing that similar point estimates can be obtained through different
degrees of moment-enforcement adjustment. Thus, the empirical application shows why the choice of $\gamma$ should not be
evaluated only by the stability of $\hat\theta_{\gamma}$. Even when
$\hat\theta_{\gamma}$ is nearly unchanged, different curvature choices can imply
different finite-sample costs in the full estimator--multiplier system.

\section{Conclusion}

This paper develops a data-tunable CRPD framework for moment-based estimation in
finite samples. The CRPD family is indexed by the power parameter $\gamma$, which
has often been treated as a researcher-chosen tuning index, with empirical
likelihood and exponential tilting arising as familiar benchmark cases. We
interpret $\gamma$ instead as a curvature parameter of the divergence objective.
Through this curvature channel, $\gamma$ governs the sensitivity of the implied
probabilities to sample moment discrepancies and therefore affects the
finite-sample behavior of the full estimator--multiplier system.

The central implication is that the choice of $\gamma$ should be evaluated not
only through the structural estimate $\hat\theta_{\gamma}$, but also through the
multiplier and implied-probability system used to enforce the moment restrictions.
Even when different values of $\gamma$ deliver similar structural estimates, they
can imply different multiplier adjustments, different degrees of weight
concentration, and different proximity to the boundary of the feasible
probability region. These features matter for finite-sample stability because
they determine how the empirical distribution is reweighted to satisfy the moment
conditions.

Motivated by this observation, we propose a system-risk criterion for selecting
$\gamma$. The criterion is both estimation-oriented and system-oriented: a
researcher-specified weight determines the relative emphasis placed on
structural accuracy and multiplier stability, allowing the criterion to target
either objective or a balance between the two. Its
structural component measures the finite-sample distortion of
$\hat\theta_{\gamma}$ relative to a first-order GMM benchmark, while its
multiplier component measures the stability cost associated with enforcing the
moment restrictions. The selected value of $\gamma$ is therefore chosen to reduce
second-order finite-sample distortion in the structural estimate while maintaining
stable moment enforcement in the full estimator--multiplier system.

The simulation evidence supports this interpretation. Across the designs
considered, the selected CRPD estimator remains approximately centered around the
true structural parameter, while the system-risk criterion helps avoid curvature
choices that require large multiplier adjustments or highly distorted implied
weights. The results show that finite-sample performance depends not only on the
point estimate, but also on the stability of the implied probability
distribution. In this sense, the proposed selector uses information from the
entire CRPD system rather than selecting $\gamma$ from the structural estimate
alone.

The empirical application to the \citet{Owen2001}'s cow data provides a simple illustration of
the same mechanism. The estimated mean daily milk yield is robust across the
selected values of $\gamma$, but the implied multipliers and probability weights
reveal differences in moment-enforcement cost. Thus, the application shows why
$\gamma$ should be interpreted as a curvature parameter governing the
finite-sample geometry of the estimator, rather than as a fixed index chosen only
from conventional benchmark values.

For future research, the
role of $\gamma$ under moment misspecification deserves separate analysis. Under
misspecification, the first-order equivalence across members of the CRPD class can
break down, so $\gamma$ may affect not only higher-order behavior but also the
leading asymptotic approximation. Characterizing the resulting pseudo-true
targets, robustness properties, and inference procedures under misspecification
would further clarify the role of CRPD curvature in empirical moment-based
estimation.

\begin{landscape}
\begin{figure}[hbtp]
\centering

\begin{subfigure}{0.48\linewidth}
    \centering
    \includegraphics[width=\linewidth]{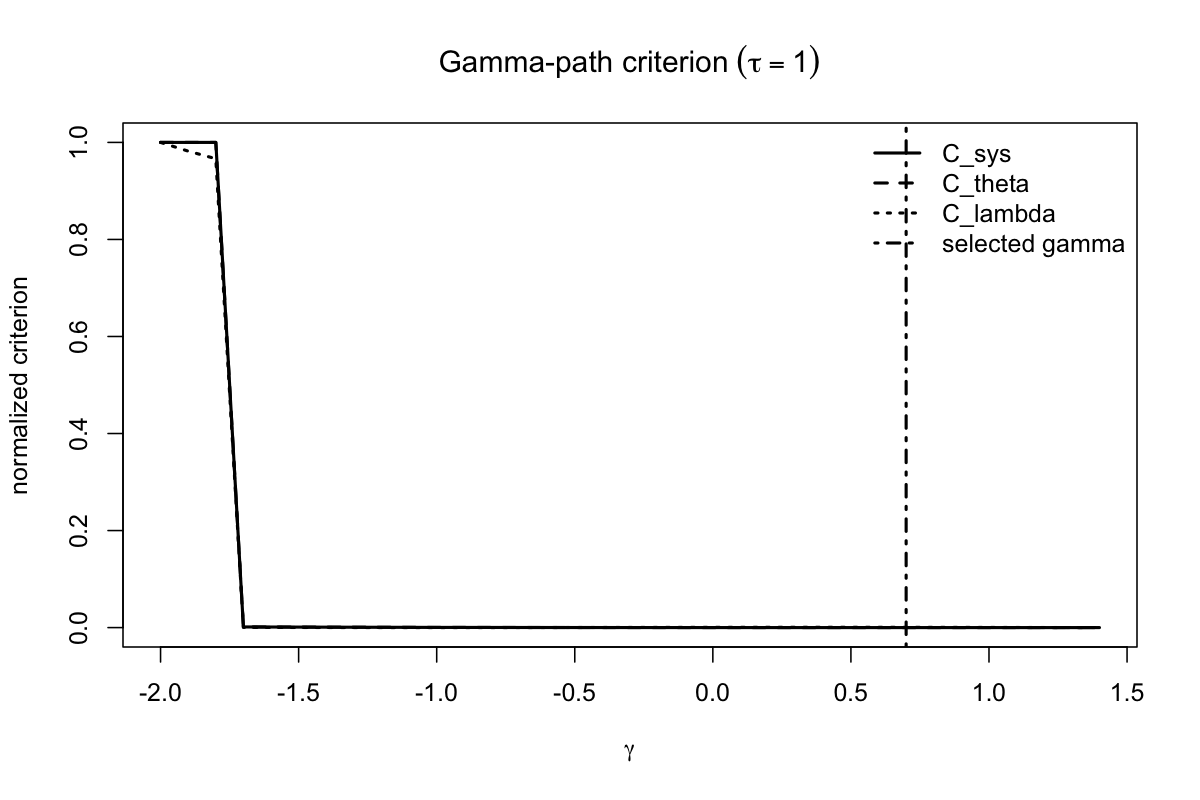}
    \caption{$\tau=1$}
\end{subfigure}
\hfill
\begin{subfigure}{0.48\linewidth}
    \centering
    \includegraphics[width=\linewidth]{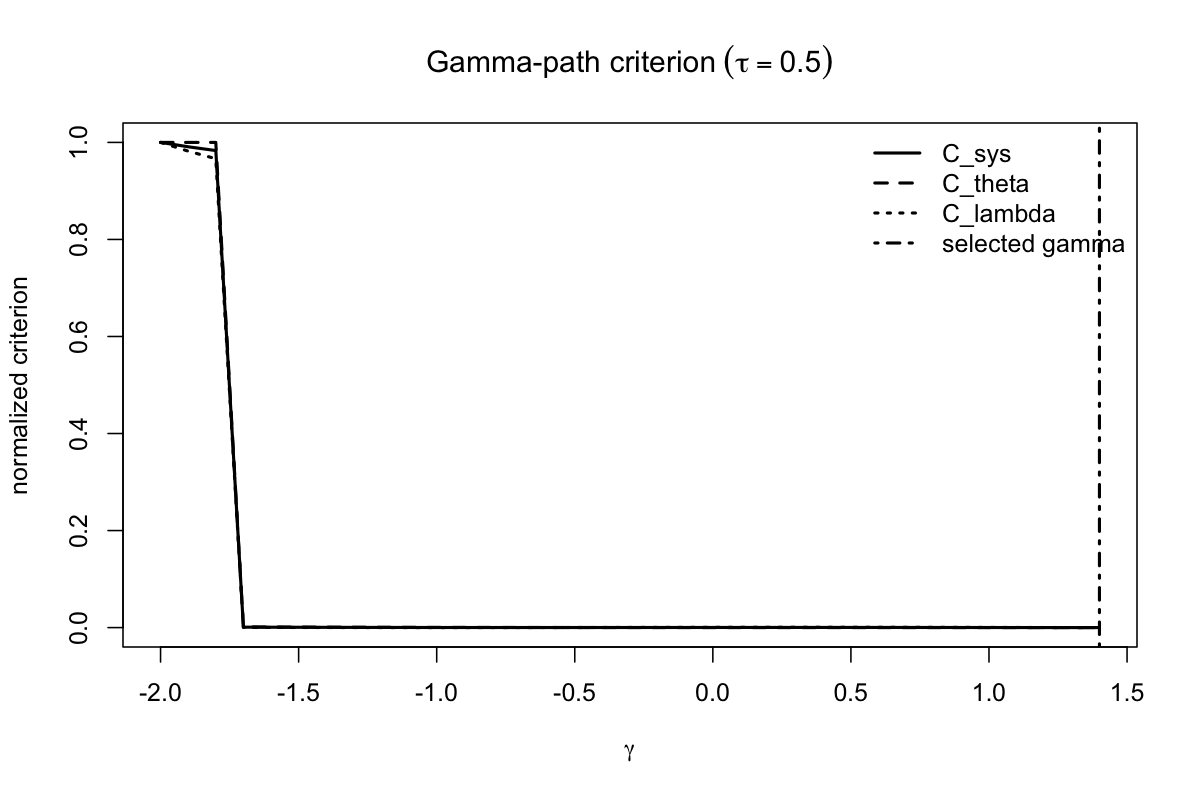}
    \caption{$\tau=0.5$}
\end{subfigure}

\vspace{0.6em}

\begin{subfigure}{0.48\linewidth}
    \centering
    \includegraphics[width=\linewidth]{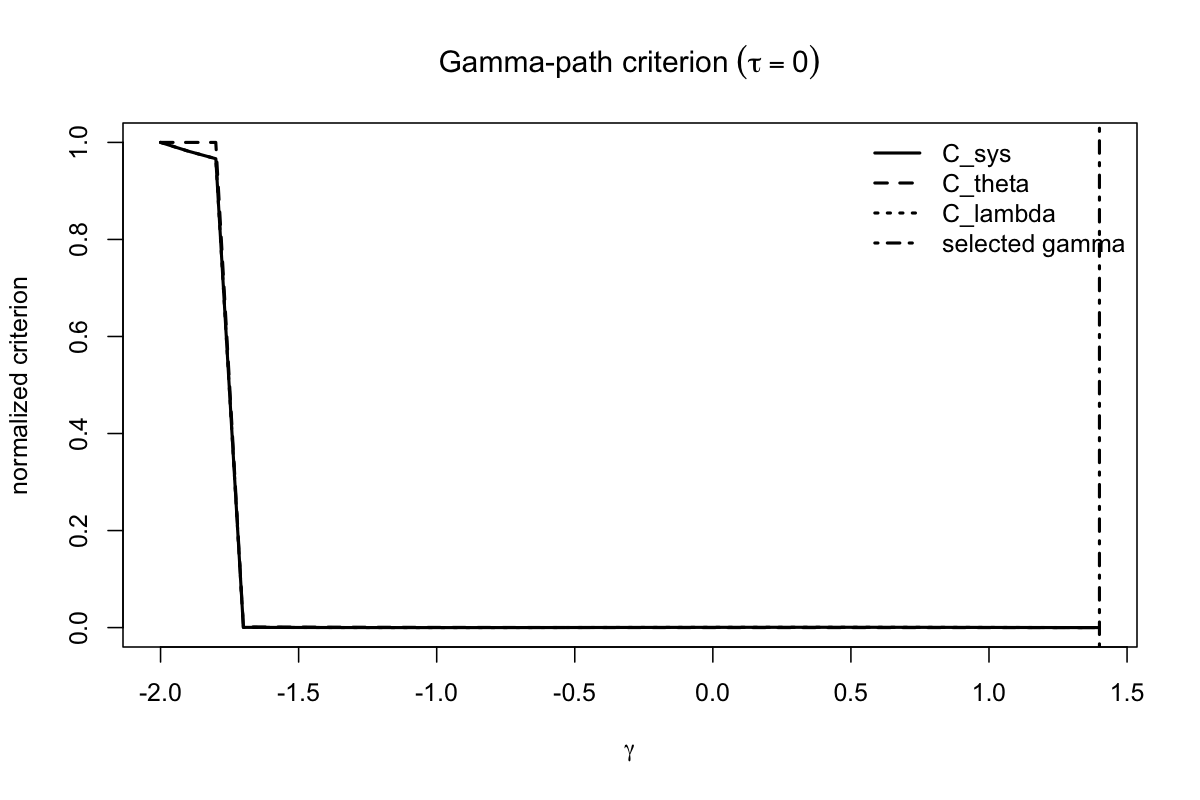}
    \caption{$\tau=0$}
\end{subfigure}

\caption{Gamma-path criterion for the Owen cow data}
\label{fig:owen_gamma_path}
\caption*{\footnotesize
Note: The solid line reports $C_{\mathrm{sys}}(\gamma;\tau)$, the dashed line
reports $C_{\theta}(\gamma)$, and the dotted line reports
$C_{\lambda}(\gamma)$. The vertical dashed line marks the selected value of
$\gamma$.}
\end{figure}
\end{landscape}

\begin{landscape}
\begin{figure}[p]
\centering

\begin{subfigure}{0.48\linewidth}
    \centering
    \includegraphics[width=\linewidth]{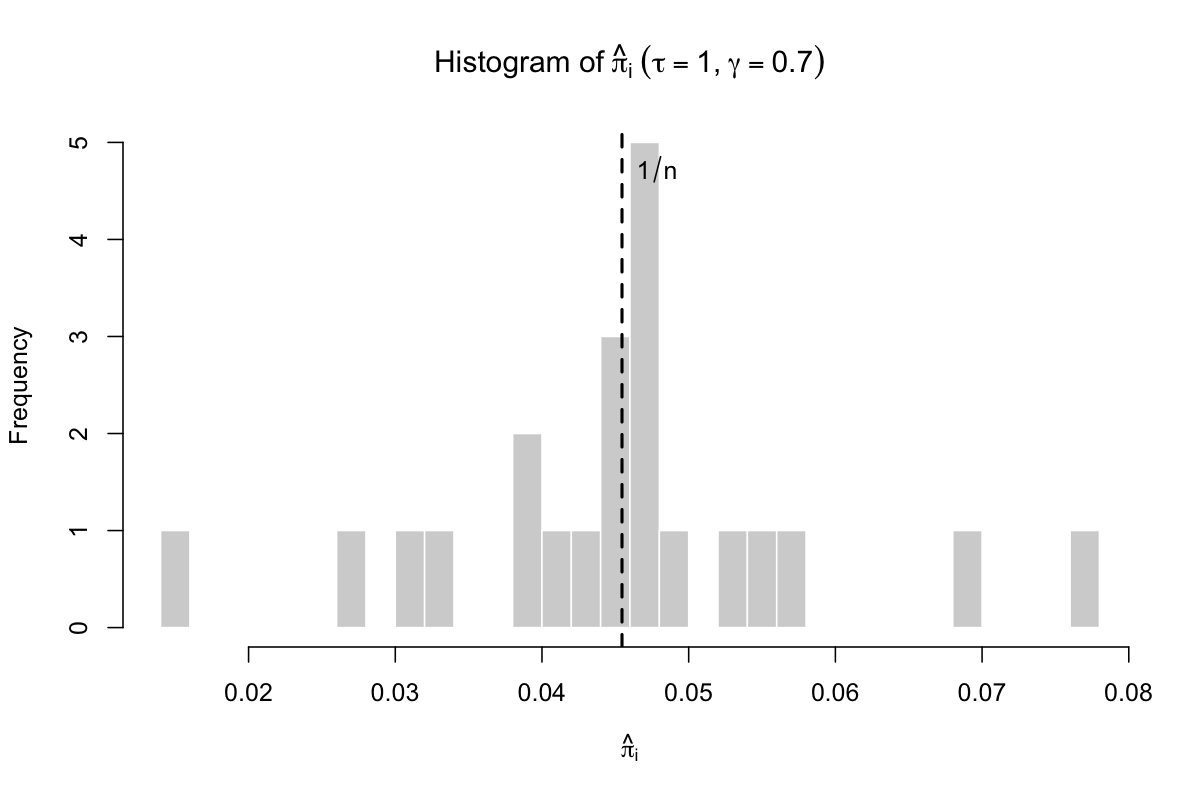}
    \caption{$\tau=1$, $\hat\gamma=0.7$}
\end{subfigure}
\hfill
\begin{subfigure}{0.48\linewidth}
    \centering
    \includegraphics[width=\linewidth]{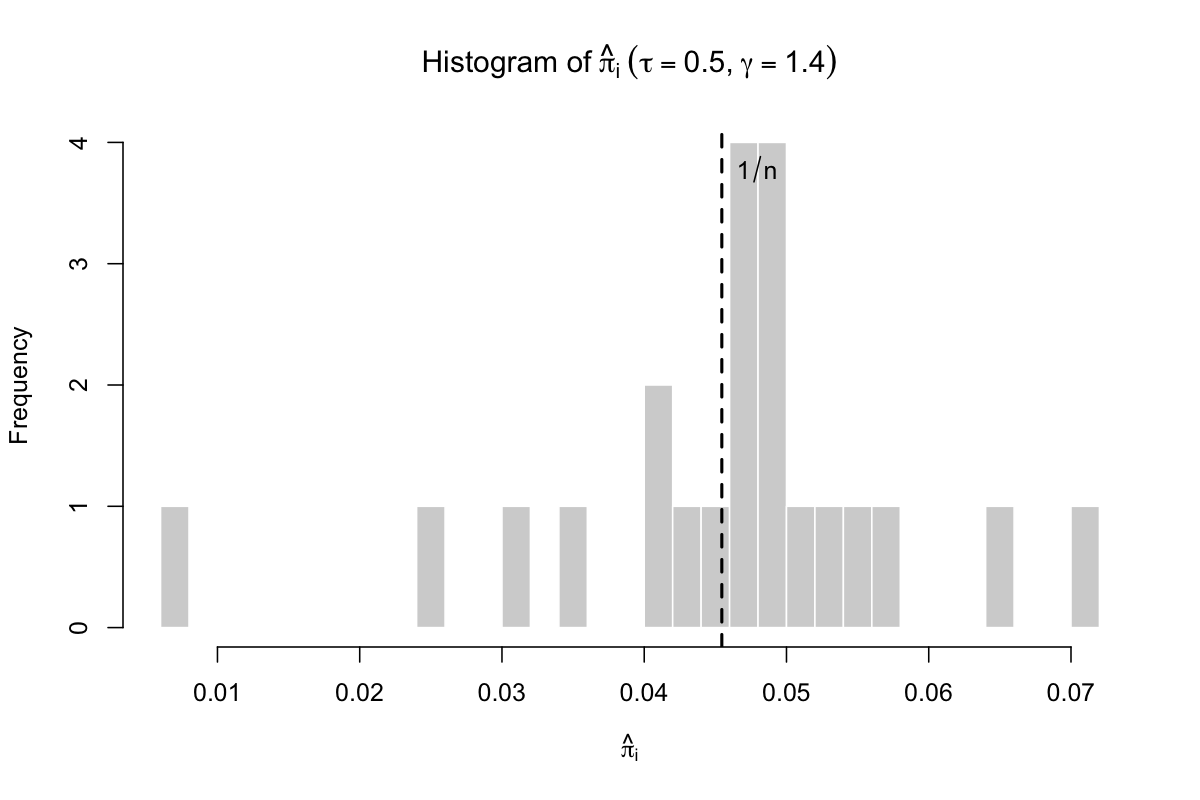}
    \caption{$\tau=0.5$, $\hat\gamma=1.4$}
\end{subfigure}

\vspace{0.6em}

\begin{subfigure}{0.48\linewidth}
    \centering
    \includegraphics[width=\linewidth]{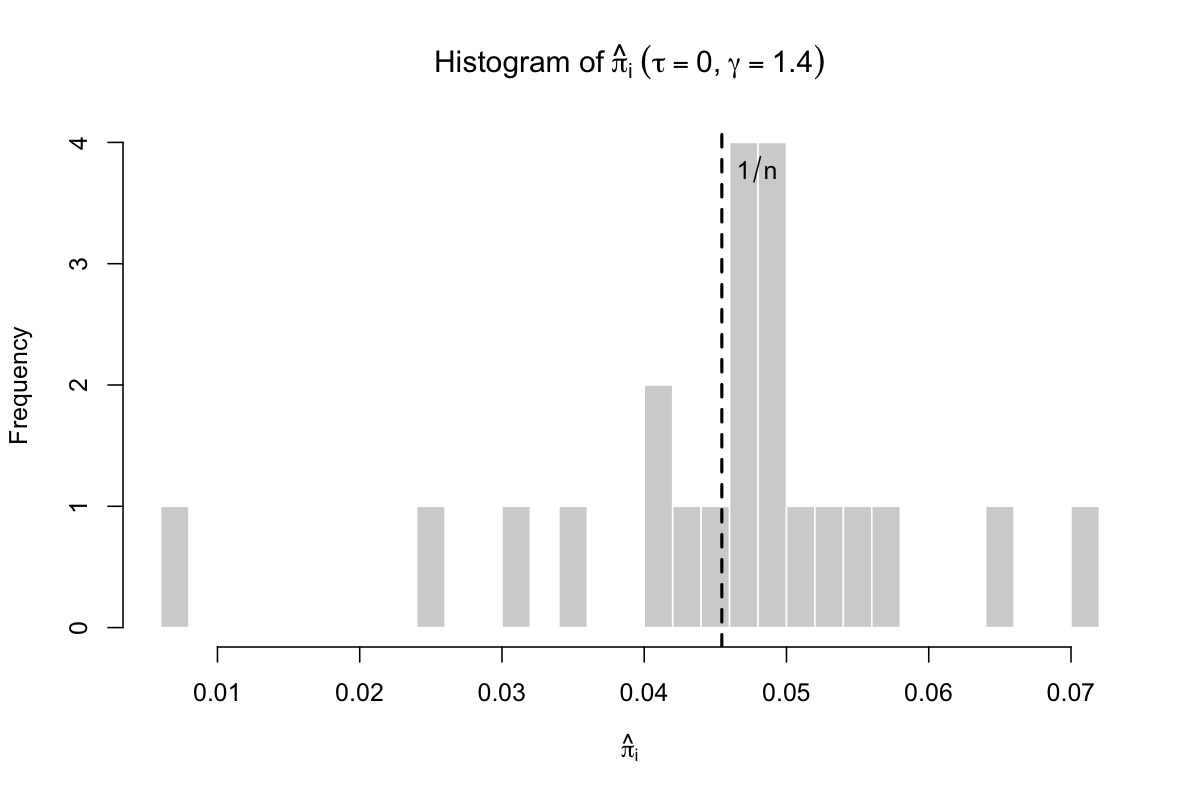}
    \caption{$\tau=0$, $\hat\gamma=1.4$}
\end{subfigure}

\caption{Estimated implied probabilities for selected values of $\gamma$}
\label{fig:owen_pi_hist_selected}
\caption*{\footnotesize
Note: The dashed vertical line denotes the uniform benchmark $1/n$. The selected
specifications have no boundary hits under the absolute threshold $10^{-6}$ or
the relative threshold $0.01/n$.}
\end{figure}
\end{landscape}

\section*{Acknowledgements}
We thank Xiaofeng Shao (Washington University in St. Louis), Seojeong (Jay) Lee (Seoul National University), Yuya Sasaki (Vanderbilt University), and the audience at the Econometrics Brown Bag Seminar at the University of Illinois Urbana–Champaign for helpful comments. All remaining errors are our own.

\newpage
\section*{Appendix}
\setcounter{equation}{0}
\renewcommand{\theequation}{A.\arabic{equation}}
\renewcommand{\thesection}{A}
\renewcommand{\thesubsection}{A.\arabic{subsection}}

\subsection{Proof for Lemma \ref{lem:population_solution}}\label{subsec:pf_lem-population-solution}
To characterize the population solution, consider the population CRPD problem obtained by replacing sample averages with expectations:
$
\min_{P} \mathscr{I}_{\gamma}(P,P_{0})
\quad \text{subject to}$ $ 
\mathbb{E}_{P}[g(X,\theta)] = 0,
$
where $P_{0}$ denotes the true data-generating distribution.

Under correct specification,
$
\mathbb{E}_{P_{0}}[g(X,\theta_{0})]=0,
$
so the true distribution $P_{0}$ already satisfies the moment restriction. Since CRPD divergences are uniquely minimized at the reference distribution, the population minimizer is
$
P=P_{0},
\;
\lambda_{0}=0,
$
that is, no distortion of the distribution is required when the moment condition holds exactly.

At the sample level, the empirical distribution with uniform weights $\pi_i=1/n$ is the natural analogue of $P_{0}$. By the Glivenko–Cantelli theorem, the empirical distribution converges uniformly to $P_{0}$. Hence, consistency of the CRPD estimator implies
$
\hat{\lambda} \xrightarrow{p} 0.
$

More explicitly, under correct specification the population moment condition satisfies
$
\mathbb{E}[g(X,\theta_{0})]=0.
$
By the Law of Large Numbers,
$
\frac{1}{n}\sum_{i=1}^n g_i(\theta_{0}) = O_p(n^{-1/2}),
$
so the sample moment violation at $\theta=\theta_{0}$ is asymptotically negligible.

Substituting $\pi_i=1/n$ into \eqref{pi_i} yields
$
\frac{1}{\gamma+1}-\gamma \delta_0 = 1,
$
which implies
$
\delta_0 = -\frac{1}{\gamma+1}.
$ Therefore, under correct specification the solution lies in a neighborhood of \\
$
(\theta_0,\lambda_0,\pi_i,\delta_0)
=
\left(\theta_0,0,\frac{1}{n},-\frac{1}{\gamma+1}\right),
$
as claimed. \qed

\subsection{Proof for Theorem \ref{thm:uniform_conv_profiled}} \label{subsec:pf_thm-uniform-conv-profiled}

\textbf{Step 1: Uniform LLN for the multiplier system.}

By Assumptions \ref{assumption:theta_compact} and \ref{assumption:smooth_envolope}, $\Theta$ is compact and
$g(z,\theta)$ is continuously differentiable with square-integrable envelope.
By Assumption \ref{assumption:interior_feasibility}, with probability approaching one,
there exists a compact set
$
\mathcal H
=
\bigl\{(\lambda,\delta):
\|\lambda\|+\|\delta-\delta_0\|\le \kappa
\bigr\}
$
such that for all $(\theta,\lambda,\delta)\in \mathcal N\times\mathcal H$,
$
\frac{1}{\gamma+1}
-
\gamma\delta
-
\gamma\lambda^{\prime} g_i(\theta)
\ge \kappa>0.
$ On this interior region,
$w_i(\theta,\lambda,\delta)$
is continuously differentiable in $(\theta,\lambda,\delta)$
and, by the positivity bound in Assumption \ref{assumption:interior_feasibility} together with Assumption \ref{assumption:smooth_envolope},
is uniformly bounded by an integrable envelope function.
Consequently, the empirical averages of
$w(Z;\theta,\lambda,\delta)$
and
$w(Z;\theta,\lambda,\delta)g(Z,\theta)$
satisfy a uniform law of large numbers over
$(\theta,\lambda,\delta)\in\mathcal N\times\mathcal H$; that is,
$
\sup_{(\theta,\lambda,\delta)\in\mathcal N\times\mathcal H}
\left|
\frac{1}{n}\sum_{i=1}^n w_i(\theta,\lambda,\delta)
-
E\!\left[w(Z;\theta,\lambda,\delta)\right]
\right|
\xrightarrow{p}
0,
$
and \\$
\sup_{(\theta,\lambda,\delta)\in\mathcal N\times\mathcal H}
\left\|
\frac{1}{n}\sum_{i=1}^n
w_i(\theta,\lambda,\delta)g_i(\theta)
-
E\!\left[w(Z;\theta,\lambda,\delta)g(Z,\theta)\right]
\right\|
\xrightarrow{p}
0.
$\\
Therefore,
$
\sup_{(\theta,\lambda,\delta)\in\mathcal N\times\mathcal H}
\bigl\|
\Psi_n(\theta,\lambda,\delta)
-
\Psi(\theta,\lambda,\delta)
\bigr\|
\xrightarrow{p}
0.
$

\textbf{Step 2: Uniform convergence of profiled multipliers.}

By Assumption 3,
the Jacobian
$
\partial_{(\lambda,\delta)}
\Psi(\theta_0,\lambda_0,\delta_0)
$
is nonsingular.
By continuity,
it remains nonsingular on a neighborhood $\mathcal N$ of $\theta_0$. Combining the uniform convergence of $\Psi_n$
with the nonsingularity of the Jacobian
and standard uniform $Z$-estimation arguments,
\[
\sup_{\theta\in\mathcal N}
\left\|
\begin{pmatrix}
\hat\lambda(\theta) \\
\hat\delta(\theta)
\end{pmatrix}
-
\begin{pmatrix}
\lambda(\theta) \\
\delta(\theta)
\end{pmatrix}
\right\|
\xrightarrow{p}
0.
\]

\textbf{Step 3: Uniform convergence of the profiled objective.}

On the interior region,
the map
$
(\theta,\lambda,\delta)
\mapsto
\mathscr{I}_{\gamma}\bigl(\pi(\theta,\lambda,\delta),i/n\bigr)
$
is continuous and uniformly continuous on
$\mathcal N\times\mathcal H$. Hence,
$
\sup_{\theta\in\mathcal N}
\bigl|
L_n(\theta)
-
L(\theta)
\bigr|
\xrightarrow{p}
0.
$

\textbf{Step 4: Identification and uniqueness.}

By Assumption \ref{assumption:identification},
$E[g(Z,\theta)]=0$
if and only if $\theta=\theta_0$.
Under correct specification,
the CRPD divergence is minimized
only when the implied weights satisfy the moment restriction,
which occurs uniquely at $\theta_0$. Therefore,
$L(\theta)$ is uniquely minimized at $\theta_0$.
 \qed

\subsection{Proof for Lemma \ref{lem:exist_unique}} \label{subsec:pf_lem-exist-unique}
Fix $\theta\in\mathcal{N}$. By Assumption \ref{assumption:interior_feasibility}, there exists an interior solution in the domain
where all terms $s_i(\theta,\lambda,\delta):=\frac{1}{\gamma+1}-\gamma\delta-\gamma\lambda^{\prime}g_i(\theta)$
are bounded away from zero, ensuring that $w_i=s_i^{1/\gamma}$ is well-defined and smooth.

Consider the mapping $(\lambda,\delta)\mapsto \Psi_n(\theta,\lambda,\delta)$ on the compact set
$\{(\lambda,\delta):\|\lambda\|+\|\delta-\delta_0\|\le \kappa\}$ intersected with the positivity domain
$s_i\ge \kappa$. On this set, $\Psi_n$ is continuously differentiable.
Compute the Jacobian with respect to $(\lambda,\delta)$:
\[
\frac{\partial w_i}{\partial\lambda}
=
-\Big(\frac{1}{\gamma+1}-\gamma\delta-\gamma\lambda^{\prime}g_i(\theta)\Big)^{1/\gamma-1} g_i(\theta), \qquad
\frac{\partial w_i}{\partial\delta}
=
-\Big(\frac{1}{\gamma+1}-\gamma\delta-\gamma\lambda^{\prime}g_i(\theta)\Big)^{1/\gamma-1}
\]
Thus,
\[
\frac{\partial \Psi_{n,1}}{\partial\lambda^{\prime}}
=
-\frac{1}{n}\sum_{i=1}^n s_i^{1/\gamma-1} g_i(\theta)^{\prime},\qquad
\frac{\partial \Psi_{n,1}}{\partial\delta}
=
-\frac{1}{n}\sum_{i=1}^n s_i^{1/\gamma-1},
\]
and
\[
\frac{\partial \Psi_{n,2}}{\partial\lambda^{\prime}}
=
-\frac{1}{n}\sum_{i=1}^n s_i^{1/\gamma-1} g_i(\theta)g_i(\theta)^{\prime},\;
\frac{\partial \Psi_{n,2}}{\partial\delta}
=
-\frac{1}{n}\sum_{i=1}^n s_i^{1/\gamma-1} g_i(\theta).
\]
At $(\lambda,\delta)=(0,\delta_0)$ we have $s_i=1$ and hence
$
\frac{\partial \Psi_{n,1}}{\partial\lambda^{\prime}}\Big|_{(0,\delta_0)}=-\bar g(\theta)^{\prime},\;
\frac{\partial \Psi_{n,1}}{\partial\delta}\Big|_{(0,\delta_0)}=-1,
$
$
\frac{\partial \Psi_{n,2}}{\partial\lambda^{\prime}}\Big|_{(0,\delta_0)}=-\hat\Omega(\theta),\qquad
\frac{\partial \Psi_{n,2}}{\partial\delta}\Big|_{(0,\delta_0)}=-\bar g(\theta),
$
where $\bar g(\theta):=\frac{1}{n}\sum_{i=1}^n g_i(\theta)$ and
$\hat\Omega(\theta):=\frac{1}{n}\sum_{i=1}^n g_i(\theta)g_i(\theta)^{\prime}$. Denote the Jacobian with respect to $(\lambda^{\prime},\delta)$ at $(\lambda,\delta)=(0,\delta_0)$ by
\[
J_n(\theta)
:=
\frac{\partial \Psi_n(\theta,\lambda,\delta)}{\partial(\lambda^{\prime},\delta)}
\Bigg|_{(\lambda,\delta)=(0,\delta_0)}
=
\begin{pmatrix}
-\bar g(\theta)^{\prime} & -1 \\
-\hat\Omega(\theta) & -\bar g(\theta)
\end{pmatrix}.
\]
Since signs do not affect nonsingularity, it suffices to consider
$
\tilde J_n(\theta)
=
\begin{pmatrix}
\bar g(\theta)^{\prime} & 1 \\
\hat\Omega(\theta) & \bar g(\theta)
\end{pmatrix}.
$
Permute rows and columns to place the scalar block first. Let $P$ be the permutation matrix
that swaps the $(q+1)$st coordinate with the first $q$ coordinates, so that
$
P\,\tilde J_n(\theta)\,P^{\prime}
=
\begin{pmatrix}
1 & \bar g(\theta)^{\prime} \\
\bar g(\theta) & \hat\Omega(\theta)
\end{pmatrix}.
$
Because permutation preserves determinant up to sign, $\tilde J_n(\theta)$ is nonsingular
if and only if $P\tilde J_n(\theta)P^{\prime}$ is nonsingular. By Assumption~\ref{assumption:identification} and continuity,
$\hat\Omega(\theta)$ is nonsingular w.p.a.1 for $\theta$ near $\theta_0$. Using the Schur
complement with respect to $\hat\Omega(\theta)$ yields
$
\det\!\bigl(P\tilde J_n(\theta)P^{\prime}\bigr)
=
\det\!\bigl(\hat\Omega(\theta)\bigr)\,
\Bigl(1-\bar g(\theta)^{\prime}\hat\Omega(\theta)^{-1}\bar g(\theta)\Bigr).
$
Since $\hat\Omega(\theta)$ is positive definite w.p.a.1, the quadratic form
$\bar g(\theta)^{\prime}\hat\Omega(\theta)^{-1}\bar g(\theta)\ge 0$. Moreover, for $\theta$ in a
shrinking neighborhood of $\theta_0$, $\bar g(\theta)=O_p(n^{-1/2})$, so
$\bar g(\theta)^{\prime}\hat\Omega(\theta)^{-1}\bar g(\theta)=O_p(n^{-1})$, and hence
$
1-\bar g(\theta)^{\prime}\hat\Omega(\theta)^{-1}\bar g(\theta)=1+o_p(1).
$
Therefore the Schur complement is strictly positive w.p.a.1 for $n$ large, implying that
$J_n(\theta)$ is nonsingular in a neighborhood of $(0,\delta_0)$ w.p.a.1.

Therefore, by the Implicit Function Theorem there exists a neighborhood
$\mathcal{N}$ of $\theta_0$ and a unique continuously differentiable mapping
$
\theta \mapsto \bigl(\lambda(\theta),\delta(\theta)\bigr)
$
defined on $\mathcal{N}$ such that
$
\Psi_n\bigl(\theta,\lambda(\theta),\delta(\theta)\bigr)=0
\; \text{for all } \theta \in \mathcal{N}.
$
In particular, the solution $(\lambda,\delta)$ is locally unique in a neighborhood
of $(0,\delta_0)$, and depends smoothly on $\theta$. \qed

\subsection{Proof for Theorem \ref{thm:theta_consistency}} \label{subsec:pf_thm-theta-consistency}
By Theorem \ref{thm:uniform_conv_profiled}, the profiled criterion $L_n(\theta)$ converges uniformly in probability to a deterministic limit function $L(\theta)$ on $\Theta$, that is,
$
\sup_{\theta \in \Theta} |L_n(\theta) - L(\theta)| \xrightarrow{p} 0.
$
By Assumption \ref{assumption:identification}, the population criterion $L(\theta)$ is uniquely minimized at $\theta_0$. Under Assumption \ref{assumption:theta_compact}, the conditions of the standard argmin theorem for M-estimators 
are satisfied. Therefore, any measurable selection
$
\hat{\theta} \in \arg\min_{\theta \in \Theta} L_n(\theta)
$
satisfies
$
\hat{\theta} \xrightarrow{p} \theta_0.
$
\qed

\subsection{Proof for Corollary \ref{cor:mult_consistency}}\label{subsec:pf_cor-mult-consistency}

By Theorem \ref{thm:theta_consistency}, $\hat{\theta} \xrightarrow{p} \theta_0$. Lemma \ref{lem:exist_unique} implies the mapping 
$\theta \mapsto (\lambda(\theta),\delta(\theta))$ is continuous at $\theta_0$, the Continuous Mapping Theorem implies
$
(\lambda(\hat{\theta}),\delta(\hat{\theta}))
\xrightarrow{p}
(\lambda(\theta_0),\delta(\theta_0)).
$
By definition, $(\hat{\lambda},\hat{\delta}) = (\lambda(\hat{\theta}),\delta(\hat{\theta}))$, which yields the result.
\qed

\subsection{Proof for Theorem \ref{thm:lambda_delta_rates}}\label{subsec:pf_thm-lambda-delta-rates}

Because $(\hat\lambda(\hat{\theta}),\hat\delta(\hat{\theta}))$ satisfies 
$
\Psi_n(\hat\theta,\hat\lambda(\hat{\theta}),\hat\delta(\hat{\theta}))=0
$
and $\hat\theta\xrightarrow{p}\theta_0$, it suffices to expand 
$
\Psi_n(\theta,\lambda(\theta),\delta(\theta))
$
in $(\lambda(\theta),\delta(\theta))$ around $(0,\delta_0)$ with $\theta$ in a shrinking neighborhood of $\theta_0$. 
For notational simplicity, we suppress the explicit profiling dependence on $\theta$ in what follows and write
$(\hat\lambda,\hat\delta)$ for $(\hat\lambda(\hat\theta),\hat\delta(\hat\theta))$, and $(\lambda,\delta)$ for
$(\lambda(\theta),\delta(\theta))$ whenever no confusion arises.

Write
$
s_i(\theta,\lambda,\delta)=\frac{1}{\gamma+1}-\gamma\delta-\gamma\lambda^{\prime}g_i(\theta)$ and $
w_i(\theta,\lambda,\delta)=s_i(\theta,\lambda,\delta)^{1/\gamma}.
$
By Lemma~\ref{lem:exist_unique}, with probability approaching one, $(\hat\lambda,\hat\delta)$ lies in a neighborhood
of $(0,\delta_0)$ on which $\inf_i s_i(\hat\theta,\lambda,\delta)\ge c>0$.
Hence $w_i(\theta,\lambda,\delta)$ is twice continuously differentiable in $(\lambda,\delta)$ on that neighborhood, and
its first derivatives are Lipschitz (uniformly in $i$) on the same set.

\paragraph{Step 1: Mean-value expansion and remainder bound.}
Apply a mean-value expansion in $(\lambda,\delta)$ around $(0,\delta_0)$:
\[
w_i(\hat\theta,\hat\lambda,\hat\delta)
=
w_i(\hat\theta,0,\delta_0)
+
\frac{\partial w_i}{\partial\delta}\Big|_{(\hat\theta,\tilde\lambda,\tilde\delta)}(\hat\delta-\delta_0)
+
\frac{\partial w_i}{\partial\lambda^{\prime}}\Big|_{(\hat\theta,\tilde\lambda,\tilde\delta)}\hat\lambda,
\]
for some intermediate $(\tilde\lambda,\tilde\delta)$ on the line segment between $(0,\delta_0)$ and
$(\hat\lambda,\hat\delta)$. Since $w_i(\hat\theta,0,\delta_0)=1$ by the choice of $\delta_0$, define the remainder
\[
r_{i,n}
:=
\Big(\frac{\partial w_i}{\partial\delta}\Big|_{(\hat\theta,\tilde\lambda,\tilde\delta)}
      -\frac{\partial w_i}{\partial\delta}\Big|_{(\hat\theta,0,\delta_0)}\Big)(\hat\delta-\delta_0)
+
\Big(\frac{\partial w_i}{\partial\lambda^{\prime}}\Big|_{(\hat\theta,\tilde\lambda,\tilde\delta)}
      -\frac{\partial w_i}{\partial\lambda^{\prime}}\Big|_{(\hat\theta,0,\delta_0)}\Big)\hat\lambda.
\]
Because the derivatives are Lipschitz on the neighborhood (uniformly in $i$),
\begin{equation}\label{eq:r_bound_pre}
\max_{1\le i\le n}|r_{i,n}|
=
O_p\bigl(\|\hat\lambda\|^2+|\hat\delta-\delta_0|^2\bigr).
\end{equation}
Moreover, by smoothness in $\theta$ and $\hat\theta\to_p\theta_0$,
\[
\frac{\partial w_i}{\partial\delta}\Big|_{(\hat\theta,0,\delta_0)}=-1+o_p(1),
\qquad
\frac{\partial w_i}{\partial\lambda^{\prime}}\Big|_{(\hat\theta,0,\delta_0)}=-g_i(\hat\theta)^{\prime}+o_p(1),
\]
uniformly in $i$ on events of probability approaching one. Therefore,
\begin{equation}\label{eq:w_linear}
w_i(\hat\theta,\hat\lambda,\hat\delta)
=
1-(\hat\delta-\delta_0)-\hat\lambda^{\prime}g_i(\hat\theta)
+\rho_{i,n}+r_{i,n},
\end{equation}
where $\max_i|\rho_{i,n}|=o_p(1)\bigl(\|\hat\lambda\|+|\hat\delta-\delta_0|\bigr)$ and
$r_{i,n}$ satisfies \eqref{eq:r_bound_pre}.

\paragraph{Step 2: Moment constraint yields the rate and influence function for $\hat\lambda$.}
Next plug \eqref{eq:w_linear} into $\Psi_{n,2}(\hat\theta,\hat\lambda,\hat\delta)=0$:
\begin{equation} \label{eq:moment_constraint}
\begin{split}
0
&=
\frac{1}{n}\sum_{i=1}^n w_i(\hat\theta,\hat\lambda,\hat\delta)\,g_i(\hat\theta) \\
&=
\bar g(\hat\theta)
-(\hat\delta-\delta_0)\bar g(\hat\theta)
-\hat\Omega(\hat\theta)\hat\lambda
+\tilde\rho_n+\tilde r_n,
\end{split}
\end{equation}
where $\hat\Omega(\hat\theta):=n^{-1}\sum_i g_i(\hat\theta)g_i(\hat\theta)^{\prime}$,
$\tilde\rho_n:=n^{-1}\sum_i \rho_{i,n}g_i(\hat\theta)$, and
$\tilde r_n:=n^{-1}\sum_i r_{i,n}g_i(\hat\theta)$.
Under Assumption \ref{assumption:smooth_envolope} and the definitions above,
$
\tilde\rho_n=o_p(1)\bigl(\|\hat\lambda\|+|\hat\delta-\delta_0|\bigr),
\;
\tilde r_n=O_p\bigl(\|\hat\lambda\|^2+|\hat\delta-\delta_0|^2\bigr).
$
Since $\bar g(\hat\theta)=O_p(n^{-1/2})$ and $\hat\delta-\delta_0=o_p(1)$ by Corollary \ref{cor:mult_consistency},
the term $(\hat\delta-\delta_0)\bar g(\hat\theta)=o_p(n^{-1/2})$. Hence,
\begin{equation}\label{eq:lambda_key}
\hat\Omega(\hat\theta)\hat\lambda
=
\bar g(\hat\theta)
+o_p(n^{-1/2})
+o_p(1)\bigl(\|\hat\lambda\|+|\hat\delta-\delta_0|\bigr)
+O_p\bigl(\|\hat\lambda\|^2+|\hat\delta-\delta_0|^2\bigr).
\end{equation}
Since $\hat\Omega(\hat\theta)\xrightarrow{p} \Omega_0\succ0$, we have $\|\hat\Omega(\hat\theta)^{-1}\|=O_p(1)$. Thus,
multiplying \eqref{eq:lambda_key} by $\hat\Omega(\hat\theta)^{-1}$ and taking norms yields
\[
\|\hat\lambda\|
\le 
C\Big(\|\bar g(\hat\theta)\|+o_p(n^{-1/2})
+o_p(1)\bigl(\|\hat\lambda\|+|\hat\delta-\delta_0|\bigr)
+O_p\bigl(\|\hat\lambda\|^2+|\hat\delta-\delta_0|^2\bigr)\Big)
\]
for some constant $C>0$ on an event of probability approaching one.

By Lemma~\ref{lem:exist_unique}, with probability approaching one the solution
$(\hat\lambda,\hat\delta)$ lies in a fixed neighborhood of $(0,\delta_0)$,
so that $\|\hat\lambda\|+|\hat\delta-\delta_0|\le \kappa$ for some $\kappa>0$.
In particular, since $x^2 \le \kappa x$ whenever $0\le x\le \kappa$,
$
\|\hat\lambda\|^2+|\hat\delta-\delta_0|^2
\le
\kappa\bigl(\|\hat\lambda\|+|\hat\delta-\delta_0|\bigr).
$
Hence the quadratic remainder term in \eqref{eq:lambda_key}
is bounded by a constant multiple of
$\|\hat\lambda\|+|\hat\delta-\delta_0|$
on this neighborhood.

Combining the bounds above with \eqref{eq:lambda_key} and multiplying by
$\|\hat\Omega(\hat\theta)^{-1}\|$, which is $O_p(1)$ since
$\hat\Omega(\hat\theta)\to_p\Omega_0\succ0$, we obtain on an event of probability
approaching one
\[
\|\hat\lambda\|
\le
C\|\bar g(\hat\theta)\|
+
C\,o_p(1)\bigl(\|\hat\lambda\|+|\hat\delta-\delta_0|\bigr)
+
C\,\kappa\bigl(\|\hat\lambda\|+|\hat\delta-\delta_0|\bigr),
\]
for some deterministic constant $C>0$. Collecting the coefficients multiplying $\|\hat\lambda\|$ into
$
\eta_n := C\,o_p(1) + C\kappa,
$
we can rewrite the inequality as
$
\|\hat\lambda\|
\le
C\|\bar g(\hat\theta)\|
+
\eta_n \|\hat\lambda\|
+
C^{\prime}|\hat\delta-\delta_0|,
$
where $C^{\prime}>0$ is another finite constant absorbing fixed multiplicative factors.
Because $o_p(1)\to 0$ and $\kappa$ is fixed and can be chosen sufficiently small
(by restricting to the local solution), we have $\eta_n=o_p(1)$.
Hence for $n$ sufficiently large, $\eta_n<1/2$ w.p.a.1.
Rearranging gives
$
(1-\eta_n)\|\hat\lambda\|
\le
C\|\bar g(\hat\theta)\|
+
C^{\prime}|\hat\delta-\delta_0|,
$
and therefore
$
\|\hat\lambda\|
\le
\frac{C}{1-\eta_n}
\bigl(\|\bar g(\hat\theta)\|+|\hat\delta-\delta_0|\bigr).
$
Since $\hat{\delta}-\delta_{0}=o_p(1)$,  $(1-\eta_n)^{-1}=O_p(1)$ and $\bar g(\hat\theta)=O_p(n^{-1/2})$,
this implies $\hat\lambda=O_p(n^{-1/2})$. Specifically, using
$
\bar g(\hat\theta)=\bar g(\theta_0)+G_0(\hat\theta-\theta_0)+o_p(n^{-1/2}),
\;
\hat\Omega(\hat\theta)=\hat\Omega(\theta_0)+o_p(1),
$
we obtain the first-order representation of \eqref{eq:lambda_key}
$
\hat\lambda
=
\hat\Omega(\theta_0)^{-1}\bar g(\theta_0)+o_p(n^{-1/2}),
$
so $\sqrt n\,\hat\lambda \xrightarrow{d} N(0,\Omega_0^{-1})$ by the CLT for $\sqrt n\,\bar g(\theta_0)$.

\paragraph{Step 3: Probability constraint yields the rate for $\hat{\delta}-\delta_{0}$.}
Plug \eqref{eq:w_linear} into $\Psi_{n,1}(\hat\theta,\hat\lambda,\hat\delta)=0$:
\[
0
=
\frac{1}{n}\sum_{i=1}^n \bigl[1-(\hat\delta-\delta_0)-\hat\lambda^{\prime}g_i(\hat\theta)+\rho_{i,n}+r_{i,n}\bigr]-1
=
-(\hat\delta-\delta_0)-\hat\lambda^{\prime}\bar g(\hat\theta)+\bar\rho_n+\bar r_n,
\]
where $\bar\rho_n:=n^{-1}\sum_i\rho_{i,n}$ and $\bar r_n:=n^{-1}\sum_i r_{i,n}$.
Hence,
\begin{equation}\label{eq:delta_from_A_revised}
\hat\delta-\delta_0
=
-\hat\lambda^{\prime}\bar g(\hat\theta)+\bar\rho_n+\bar r_n.
\end{equation}
By construction,
$
\bar\rho_n=o_p(1)\bigl(\|\hat\lambda\|+|\hat\delta-\delta_0|\bigr),
\;
\bar r_n=O_p\bigl(\|\hat\lambda\|^2+|\hat\delta-\delta_0|^2\bigr).
$ By Step~2 we have
$
\hat\lambda = O_p(n^{-1/2}),
\;
\bar g(\hat\theta)=O_p(n^{-1/2}).
$
Therefore the leading term satisfies
$
\hat\lambda^{\prime}\bar g(\hat\theta)
=
O_p(n^{-1/2}) \cdot O_p(n^{-1/2})
=
O_p(n^{-1}).
$ Next consider the remainder terms in \eqref{eq:delta_from_A_revised}:
$
\bar\rho_n
=
o_p(1)\bigl(\|\hat\lambda\|+|\hat\delta-\delta_0|\bigr)
=
o_p(n^{-1/2}) + o_p(1)\,|\hat\delta-\delta_0|,
$
and
$
\bar r_n
=
O_p\bigl(\|\hat\lambda\|^2+|\hat\delta-\delta_0|^2\bigr)
=
O_p(n^{-1}) + O_p\bigl(|\hat\delta-\delta_0|^2\bigr).
$ Substituting these bounds into \eqref{eq:delta_from_A_revised} gives
$
\hat\delta-\delta_0
=
-\hat\lambda^{\prime}\bar g(\hat\theta)
+ O_p(n^{-1})
+ o_p(n^{-1/2})
+ o_p(1)\,|\hat\delta-\delta_0|
+ O_p\bigl(|\hat\delta-\delta_0|^2\bigr).
$

Since $|\hat\delta-\delta_0|=o_p(1)$ by Corollary \ref{cor:mult_consistency}, the terms
$o_p(1)\,|\hat\delta-\delta_0|$ and $O_p(|\hat\delta-\delta_0|^2)$
are of smaller order than $|\hat\delta-\delta_0|$ itself and therefore negligible
relative to the leading terms on the right-hand side.
Consequently,
$
\hat\delta-\delta_0
=
-\hat\lambda^{\prime}\bar g(\hat\theta)
+ O_p(n^{-1})
+ o_p(n^{-1/2}),
$
and
$
\hat\lambda^{\prime}\bar g(\hat\theta)=O_p(n^{-1}).
$
Hence all terms on the right-hand side are of order $O_p(n^{-1})$,
which implies
$
\hat\delta-\delta_0 = O_p(n^{-1}).
$
\qed

\subsection{Proof for Theorem \ref{thm:delta_limit}}\label{subsec:pf_thm-delta-limit}
Define
$
t_i(\theta):=(\hat\delta-\delta_0)+\hat\lambda^{\prime}g_i(\theta),
\;
\text{so that}\;
w_i=(1-\gamma t_i)^{1/\gamma}
$
because $\frac{1}{\gamma+1}-\gamma\hat\delta-\gamma\hat\lambda^{\prime}g_i(\theta)
=1-\gamma\big((\hat\delta-\delta_0)+\hat\lambda^{\prime}g_i(\theta)\big)$ when $\delta_0=-1/(\gamma+1)$. Note that at the population solution under correct specification, we have
$\lambda_0=0$ and $\delta=\delta_0$, and hence
$
t_i(\theta_0;\lambda_0,\delta_0)
=
(\delta_0-\delta_0)+\lambda_0^{\prime} g_i(\theta_0)
=
0,
$
so that $t=0$ is the appropriate expansion point. By Theorem~\ref{thm:lambda_delta_rates}, $t_i(\theta_0)=O_p(n^{-1/2})$, and hence
$t_i(\theta_0)\xrightarrow{p} 0$. Therefore a Taylor expansion of $(1-\gamma t)^{1/\gamma}$
around $t=0$ is justified. Moreover,
$
t_i^3 = O_p(n^{-3/2}) = o_p(n^{-1/2}),
$
so the cubic remainder term is negligible at the $\sqrt{n}$-scale relevant for
first-order asymptotics. Consequently,
\begin{equation}\label{eq:w_i}
w_i=
(1-\gamma t_{i})^{1/\gamma}
=
1 - t_{i} + \frac{1-\gamma}{2}\, t_{i}^2 + O(t_{i}^3),
\end{equation}
where the expansion holds uniformly on events of probability approaching one
under the stated remainder control.

Apply the probability constraint $\Psi_{n,1}(\hat\theta,\hat\lambda,\hat\delta)=0$:
$
0
=
\frac{1}{n}\sum_{i=1}^n w_i-1
=\\
\frac{1}{n}\sum_{i=1}^n
\left(
1-t_i+\frac{1-\gamma}{2}t_i^{2}
\right)
-1
+O\!\left(\frac{1}{n}\sum_{i=1}^n |t_i|^3\right).
$
Cancelling the constants yields\footnote{In the present setting,
$
t_i(\theta_0)
=
(\hat\delta-\delta_0)
+
\hat\lambda^{\prime} g_i(\theta_0),
$
where $\hat\lambda=O_p(n^{-1/2})$ and $\hat\delta-\delta_0=O_p(n^{-1})$ by
Theorem~\ref{thm:lambda_delta_rates}.
Hence
$
|t_i|
\le
|\hat\delta-\delta_0|
+
|\hat\lambda^{\prime} g_i(\theta_0)|.
$
Since $g_i(\theta_0)=O_p(1)$ and $E\|g(Z,\theta_0)\|^3<\infty$ under the stated moment
conditions, we obtain
\[
|t_i|^3
\lesssim
|\hat\delta-\delta_0|^3
+
|\hat\lambda^{\prime} g_i(\theta_0)|^3
=
O_p(n^{-3})
+
O_p(n^{-3/2})\,\|g_i(\theta_0)\|^3.
\]
Averaging over $i$ yields
\[
\frac{1}{n}\sum_{i=1}^n |t_i|^3
=
O_p(n^{-3})
+
O_p(n^{-3/2})
\left(
\frac{1}{n}\sum_{i=1}^n \|g_i(\theta_0)\|^3
\right).
\]
By the law of large numbers,
$
\frac{1}{n}\sum_{i=1}^n \|g_i(\theta_0)\|^3 = O_p(1),
$
and therefore
$
\frac{1}{n}\sum_{i=1}^n |t_i|^3
=
O_p(n^{-3/2}).
$
In particular, this remainder term is $o_p(n^{-1/2})$ and hence negligible
at the $\sqrt{n}$-scale relevant for the first-order asymptotic expansion.}
\begin{equation}\label{eq:expanded_prob_constraint}
\hat\delta-\delta_0
=
-\hat\lambda^{\prime}\bar g(\hat\theta)
+\frac{1-\gamma}{2}\cdot \frac{1}{n}\sum_{i=1}^n t_i(\hat\theta)^2
+o_p(n^{-1}).
\end{equation}
Since $\hat\delta-\delta_0=O_p(n^{-1})$ and $\hat\lambda^{\prime}g_i(\hat\theta)=O_p(n^{-1/2})$,
we have
$
\frac{1}{n}\sum_{i=1}^n t_i(\hat\theta)^2
=
\hat\lambda^{\prime}\hat\Omega(\hat\theta)\hat\lambda
+o_p(n^{-1}).
$ The first term in \ref{eq:expanded_prob_constraint}, recall \ref{eq:lambda_key}, where
$
\tilde\rho_n
=
o_p(1)\bigl(O_p(n^{-1/2})+O_p(n^{-1})\bigr)
=
o_p(n^{-1/2}),
$
and
$
\tilde r_n
=
O_p\bigl(O_p(n^{-1})+O_p(n^{-2})\bigr)
=
O_p(n^{-1})
=
o_p(n^{-1/2}).
$
Hence both $\tilde\rho_n$ and $\tilde r_n$ are negligible at the $\sqrt{n}$-scale in the moment constraint expansion and
$
    \hat{\Omega}(\hat{\theta})\hat{\lambda}=\bar{g}(\hat{\theta})+o_p(n^{-1/2}).
$
Thus, $\hat{\lambda}^{\prime}\bar{g}(\theta)=\hat{\lambda}^{\prime}\hat{\Omega}(\hat{\theta})\hat{\lambda}+o_p(n^{-1})$. Therefore,
$\hat\delta-\delta_0
=
\left(\frac{1-\gamma}{2}-1\right)\hat\lambda^{\prime}\hat\Omega(\hat\theta)\hat\lambda
+o_p(n^{-1})
=
-\frac{\gamma+1}{2}\,\hat\lambda^{\prime}\hat\Omega(\theta_0)\hat\lambda
+o_p(n^{-1}).
$
Using $\hat\lambda=\hat\Omega(\theta_0)^{-1}\bar g(\theta_0)+o_p(n^{-1/2})$ and
$\hat\Omega(\theta_0)\xrightarrow{p}\Omega_0$ yields
\begin{align*}
n(\hat\delta-\delta_0)
&=
-\frac{\gamma+1}{2}\,\Big(\sqrt n\,\bar g(\theta_0)\Big)^{\prime}\Omega_0^{-1}\Big(\sqrt n\,\bar g(\theta_0)\Big)
+o_p(1),\\
&=-\frac{\gamma+1}{2}\Big(\Omega_{0}^{-1/2}\sqrt{n}\bar{g}(\theta_{0})\Big)^{\prime}\Big(\Omega_{0}^{-1/2}\sqrt{n}\bar{g}(\theta_{0})\Big)+o_p(1)
\xrightarrow{d}-\frac{\gamma+1}{2}\chi_{q}^{2},
\end{align*}
because $\sqrt n\,\bar g(\theta_0)\xrightarrow{d} N(0,\Omega_0)$. \qed

\subsection{Proof for Theorem \ref{thm:theta_CLT}}\label{subsec:pf_thm-theta-CLT}
By the envelope theorem, the first-order condition for the profiled problem can be written as the
$\theta$-derivative of the Lagrangian evaluated at the inner optimizer:
\[
0
=
\frac{\partial}{\partial\theta}
\left[
\lambda^{\prime}\sum_{i=1}^n \pi_i g_i(\theta)
\right]_{\ (\theta,\lambda,\delta)=(\hat\theta,\hat\lambda(\hat\theta),\hat\delta(\hat\theta))}
=
\sum_{i=1}^n \hat\pi_i(\hat\theta)\,G_i(\hat\theta)^{\prime}\hat\lambda,
\]
where $G_i(\theta):=\partial g_i(\theta)/\partial\theta^{\prime}$.
Divide by $n$:
$
0
=
\frac{1}{n}\sum_{i=1}^n \hat\pi_i(\hat\theta)\,G_i(\hat\theta)^{\prime}\hat\lambda.
$ Since $\hat\pi_i(\hat\theta)=\frac{1}{n}w_i(\hat\theta,\hat\lambda,\hat\delta)$, we may write
$
0=\frac{1}{n}\sum_{i=1}^n \frac{1}{n}w_i(\hat\theta,\hat\lambda,\hat\delta)\,G_i(\hat\theta)^{\prime}\hat\lambda
=\Big(\frac{1}{n}\sum_{i=1}^n w_i(\hat\theta,\hat\lambda,\hat\delta)\,G_i(\hat\theta)\Big)^{\prime}\hat\lambda\cdot \frac{1}{n}.
$
Multiplying both sides by $n$ (which does not change the equality) yields
$
0=\Big(\frac{1}{n}\sum_{i=1}^n w_i(\hat\theta,\hat\lambda,\hat\delta)\,G_i(\hat\theta)\Big)^{\prime}\hat\lambda.
$
Decompose $w_i=1+(w_i-1)$ to obtain
$
0=\Big(\frac{1}{n}\sum_{i=1}^n G_i(\hat\theta)\Big)^{\prime}\hat\lambda
+\Big(\frac{1}{n}\sum_{i=1}^n (w_i-1)G_i(\hat\theta)\Big)^{\prime}\hat\lambda.
$
By the weight expansion, $w_i-1=O_p(n^{-1/2})$ for each fixed $i$\footnote{From the first-order expansion of the weights,
$
w_i(\hat\theta,\hat\lambda,\hat\delta)
=
1-(\hat\delta-\delta_0)
-\hat\lambda^{\prime} g_i(\hat\theta)
+r_{i,n},
$
so that
$
w_i-1
=
-(\hat\delta-\delta_0)
-\hat\lambda^{\prime} g_i(\hat\theta)
+r_{i,n}.
$
By the previously established rates
$
\hat\lambda = O_p(n^{-1/2}),
\;
\hat\delta-\delta_0 = O_p(n^{-1}),
\;
r_{i,n}=O_p(\|\hat\lambda\|^2+|\hat\delta-\delta_0|^2)=O_p(n^{-1}),
$
and since for each fixed $i$ we have $g_i(\hat\theta)=O_p(1)$,
it follows that
$
w_i-1 = O_p(n^{-1/2})
\; \text{for each fixed } i.
$
Hence the natural pointwise stochastic rate is
$
w_i = 1 + O_p(n^{-1/2})
\; \text{for each fixed } i.
$} and under Assumption \ref{assumption:smooth_envolope}, we have $(1/n)\sum_{i=1}^n\|G_i(\hat\theta)\|=O_p(1)$.
Since $\hat\lambda=O_p(n^{-1/2})$, it follows that
$
\Big\|\Big(\frac{1}{n}\sum_{i=1}^n (w_i-1)G_i(\hat\theta)\Big)^{\prime}\hat\lambda\Big\|
\;\le\;
\|\hat\lambda\|\cdot \frac{1}{n}\sum_{i=1}^n |w_i-1|\,\|G_i(\hat\theta)\|
=
O_p(n^{-1/2})\cdot O_p(n^{-1/2})
=
O_p(n^{-1}),
$
and in particular this term is $o_p(n^{-1/2})$. Therefore,
\begin{equation}\label{eq:env_foc}
0=\Big(\frac{1}{n}\sum_{i=1}^n G_i(\hat\theta)\Big)^{\prime}\hat\lambda + o_p(n^{-1/2}).
\end{equation}

By the uniform law of large numbers (from the envelope condition) and consistency $\hat\theta\xrightarrow{p}\theta_0$,
$
\frac{1}{n}\sum_{i=1}^n G_i(\hat\theta)\xrightarrow{p} G_0:=E[G(Z,\theta_0)].
$
Moreover, by the first-order expansion for the multiplier,
\begin{equation}\label{eq:lambda_linrep}
\hat\lambda=\hat\Omega(\theta_0)^{-1}\,\bar g(\theta_0)+o_p(n^{-1/2}),
\end{equation}
where $\hat\Omega(\theta_0):=\frac{1}{n}\sum_{i=1}^n g_i(\theta_0)g_i(\theta_0)^{\prime}$ and $\hat\Omega(\theta_0)\xrightarrow{p}\Omega_0:=E[g(Z,\theta_0)g(Z,\theta_0)^{\prime}]$.
Substituting \eqref{eq:lambda_linrep} into \eqref{eq:env_foc} and using
$\frac{1}{n}\sum_i G_i(\hat\theta)=G_0+o_p(1)$ yields
\begin{equation}\label{eq:env_foc2}
0
=
\Big(G_0^{\prime}+o_p(1)\Big)
\Big(\Omega_0^{-1}\bar g(\theta_0)+o_p(n^{-1/2})\Big)
+o_p(n^{-1/2})
=
G_0^{\prime}\Omega_0^{-1}\bar g(\theta_0)+o_p(n^{-1/2}).
\end{equation}
To express the result in terms of $\bar g(\hat\theta)$, use the mean-value expansion
$
\bar g(\hat\theta)=\bar g(\theta_0)+G_0(\hat\theta-\theta_0)+o_p(n^{-1/2}),
$
so that $\bar g(\theta_0)=\bar g(\hat\theta)-G_0(\hat\theta-\theta_0)+o_p(n^{-1/2})$.
Plugging this into \eqref{eq:env_foc2} gives
$
0
=
G_0^{\prime}\Omega_0^{-1}\Big(\bar g(\hat\theta)-G_0(\hat\theta-\theta_0)\Big)+o_p(n^{-1/2})
=
G_0^{\prime}\Omega_0^{-1}\bar g(\hat\theta)
-
G_0^{\prime}\Omega_0^{-1}G_0(\hat\theta-\theta_0)
+o_p(n^{-1/2}),
$
or equivalently,
\begin{equation}\label{eq:theta_key}
G_0^{\prime}\Omega_0^{-1}G_0(\hat\theta-\theta_0)
=
G_0^{\prime}\Omega_0^{-1}\bar g(\hat\theta)
+o_p(n^{-1/2}).
\end{equation}
Since $G_0$ has full column rank and $\Omega_0\succ0$, the matrix
$G_0^{\prime}\Omega_0^{-1}G_0$ is nonsingular, and solving \eqref{eq:theta_key} yields the linear representation
$
\hat\theta-\theta_0
=
\big(G_0^{\prime}\Omega_0^{-1}G_0\big)^{-1}G_0^{\prime}\Omega_0^{-1}\bar g(\hat\theta)
+o_p(n^{-1/2}).
$
Using again $\bar g(\hat\theta)=\bar g(\theta_0)+o_p(n^{-1/2})$, we obtain the equivalent form
\begin{equation}\label{eq:theta_if}
\sqrt{n}(\hat\theta-\theta_0)
=
\big(G_0^{\prime}\Omega_0^{-1}G_0\big)^{-1}G_0^{\prime}\Omega_0^{-1}\sqrt{n}\,\bar g(\theta_0)
+o_p(1).
\end{equation}

By the multivariate central limit theorem,
$
\sqrt{n}\,\bar g(\theta_0)\xrightarrow{d}N(0,\Omega_0).
$
Premultiplying by the constant matrix
$
A:=\big(G_0^{\prime}\Omega_0^{-1}G_0\big)^{-1}G_0^{\prime}\Omega_0^{-1},
$
and applying Slutsky's theorem to \eqref{eq:theta_if} yields
$
\sqrt{n}(\hat\theta-\theta_0)\xrightarrow{d} N\!\Big(0,\;A\,\Omega_0\,A^{\prime}\Big)
=
N\!\Big(0,\;\big(G_0^{\prime}\Omega_0^{-1}G_0\big)^{-1}\Big).
$
\qed

\subsection{Proof for Lemma \ref{lem:moment_eq_second}} \label{subsec:pf_lem-moment-eq-second}
Recall \eqref{eq:w_i}:
$
w_i=1-t_i+\frac{1-\gamma}{2} t_i^2 + O(t_i^3),
$
where $t_{i}=O_p(n^{-1/2})$. Plugging into the weighted moment constraint gives
$
0=\frac{1}{n}\sum_{i=1}^n w_i g_i(\hat\theta)
=
\bar g(\hat\theta)
-\frac{1}{n}\sum_{i=1}^n t_i g_i(\hat\theta)
+\frac{1-\gamma}{2} \frac{1}{n}\sum_{i=1}^n t_i^2 g_i(\hat\theta)
+\frac{1}{n}\sum_{i=1}^n O(t_i^3)g_i(\hat\theta).
$
Since $t_i=(\hat\delta-\delta_0)+\hat\lambda^{\prime}g_i(\hat\theta)$ and
$\hat\delta-\delta_0=O_p(n^{-1})$, we have
$
\frac{1}{n}\sum_{i=1}^n t_i g_i(\hat\theta)
=
(\hat\delta-\delta_0)\bar g(\hat\theta)+\hat\Omega(\hat\theta)\hat\lambda.
$
The term $(\hat\delta-\delta_0)\bar g(\hat\theta)=O_p(n^{-3/2})$ is negligible at order $n^{-1}$.
For the remainder, $\frac{1}{n}\sum |t_i|^3=O_p(n^{-3/2})$, so under the moment bounds,
$\frac{1}{n}\sum O(t_i^3)g_i(\hat\theta)=o_p(n^{-1})$, yielding \eqref{eq:moment_second}.
\qed

\subsection{Proof for Theorem \ref{thm:lambda_second}}\label{subsec:pf_thm-lambda-second}

Theorem~\ref{thm:theta_consistency} and \eqref{eq:moment_second} imply
\begin{equation}\label{eq:lambda_eq_star}
\hat\Omega(\theta_0)\hat\lambda
=
\bar g(\theta_0)
+
\frac{1-\gamma}{2}\cdot \frac1n\sum_{i=1}^n
t_i(\theta_0,\hat\lambda,\hat\delta)^2\, g_i(\theta_0)
+ o_p(n^{-1}).
\end{equation}
Moreover,
$
t_i(\theta_0,\hat\lambda,\hat\delta)
=
\hat\lambda^{\prime} g_i(\theta_0)
+ O_p(n^{-1}),
$
so
$
\frac1n\sum_{i=1}^n t_i^2 g_i(\theta_0)
=
\frac1n\sum_{i=1}^n
(\hat\lambda^{\prime} g_i(\theta_0))^2 g_i(\theta_0)
+ o_p(n^{-1}).
$ Substituting the first-order representation
$
\hat\lambda
=
\hat\Omega(\theta_0)^{-1}\bar g(\theta_0)
+ o_p(n^{-1/2})
$
into the quadratic term yields, uniformly in $i$ on events of probability approaching one,
$
(\hat\lambda^{\prime} g_i(\theta_0))^2
=
\big((\hat\Omega(\theta_0)^{-1}\bar g(\theta_0))^{\prime} g_i(\theta_0)\big)^2
+ o_p(n^{-1}),
$
and therefore
\begin{equation}\label{eq:t2g_expand}
\frac1n\sum_{i=1}^n t_i^2 g_i(\theta_0)
=
\frac1n\sum_{i=1}^n
\big((\hat\Omega(\theta_0)^{-1}\bar g(\theta_0))^{\prime} g_i(\theta_0)\big)^2
g_i(\theta_0)
+ o_p(n^{-1}).
\end{equation}

Multiplying \eqref{eq:lambda_eq_star} by $\hat\Omega(\theta_0)^{-1}$ and collecting
all terms of order $n^{-1}$ gives
\begin{equation}\label{eq:lambda_second_hat}
\hat\lambda
=
\hat\Omega(\theta_0)^{-1}\bar g(\theta_0)
+
\frac{1}{n}B_{\lambda,n}(\gamma)
+
o_p(n^{-1}),
\end{equation}
where
\begin{equation}\label{eq:Blambda_hat}
B_{\lambda,n}(\gamma)
=
\frac{1-\gamma}{2}\,
\hat\Omega(\theta_0)^{-1}
\left[
\frac1n\sum_{i=1}^n
\big((\hat\Omega(\theta_0)^{-1}\sqrt n\,\bar g(\theta_0))^{\prime}
g_i(\theta_0)\big)^2
g_i(\theta_0)
\right].
\end{equation}

\noindent To verify the stochastic order of the bracketed term in \eqref{eq:Blambda_hat},
write
$
v_n := \hat\Omega(\theta_0)^{-1}\sqrt n\,\bar g(\theta_0),
$
so that the bracket equals
$
\frac{1}{n}\sum_{i=1}^n
\big(v_n^{\prime} g_i(\theta_0)\big)^2 g_i(\theta_0).
$
Using the inequality
$
|v_n^{\prime} g_i(\theta_0)|
\le \|v_n\|\,\|g_i(\theta_0)\|,
$
we obtain
$
\left\|
\frac{1}{n}\sum_{i=1}^n
\big(v_n^{\prime} g_i(\theta_0)\big)^2 g_i(\theta_0)
\right\|
\le
\|v_n\|^2
\cdot
\frac{1}{n}\sum_{i=1}^n \|g_i(\theta_0)\|^3.
$
Since $\hat\Omega(\theta_0)^{-1}=O_p(1)$ as $\hat{\Omega}\xrightarrow{p}\Omega_{0}\succ 0$ and
$\sqrt n\,\bar g(\theta_0)=O_p(1)$ by the central limit theorem,
it follows that $\|v_n\|=O_p(1)$ and hence $\|v_n\|^2=O_p(1)$.
Moreover, under Assumption \ref{assumption:smooth_envolope}, the law of large numbers yields
$
\frac{1}{n}\sum_{i=1}^n \|g_i(\theta_0)\|^3 = O_p(1).
$
Combining these bounds gives
$
\frac{1}{n}\sum_{i=1}^n
\big(v_n^{\prime} g_i(\theta_0)\big)^2 g_i(\theta_0)
=
O_p(1).
$
Hence the bracketed term in \eqref{eq:Blambda_hat} is $O_p(1)$. Thus, $B_{\lambda,n}(\gamma)=O_p(1)$.
\qed

\subsection{Proof for Theorem \ref{thm:theta_second}}\label{subsec:pf_thm-theta-second}
\textbf{Step 1: First-order condition for $\theta$.}

Let $\hat{\pi}_i(\theta)$ denote the implied probability weights at
$(\theta,\hat{\lambda}(\theta),\hat{\delta}(\theta))$.
The envelope first-order condition (FOC) for $\theta$ is

\begin{align*}
    0\overset{\text{FOC}}{=}\frac{\partial \mathcal{L}}{\partial \theta}\Bigg|_{\pi=\hat{\pi}(\theta),\lambda=\hat{\lambda}(\theta),\delta=\hat{\delta}(\theta)}
    &=\hat{\lambda}^{\prime}\sum_{i=1}^{n}\hat{\pi}_{i}(\hat{\theta})\frac{\partial g_{i}}{\partial \theta}\Bigg|_{\theta=\hat{\theta}}\\
    &=\sum_{i=1}^{n}\hat{\pi}_{i}(\hat{\theta})\frac{\partial g_{i}(\theta)}{\partial \theta}\Bigg|_{\theta=\hat{\theta}}^{\prime}\hat{\lambda}\\
    &=\sum_{i=1}^{n}\hat{\pi}_{i}(\hat{\theta})G_{i}(\hat{\theta})^{\prime}\hat{\lambda}.
\end{align*}

Since $\hat{\pi}_i(\hat{\theta})=\frac1n w_i(\hat{\theta},\hat{\lambda},\hat{\delta})$,
this condition can be written
$
0
=
\frac1n\sum_{i=1}^n
w_i(\hat{\theta},\hat{\lambda},\hat{\delta})
\,G_i(\hat{\theta})^{\prime}\hat{\lambda}.
$ Define
$
\bar G(\theta)=\frac1n\sum_{i=1}^n G_i(\theta).
$ Then the FOC becomes
$0
=
\bar G(\hat{\theta})^{\prime}\hat{\lambda}
+
R_{w,G},
$ where
$
R_{w,G}
=
\frac1n\sum_{i=1}^n
\big(w_i(\hat{\theta},\hat{\lambda},\hat{\delta})-1\big)
G_i(\hat{\theta})^{\prime}\hat{\lambda}.
$
\\

\textbf{Step 2: Order of the weight remainder.}

From the multiplier expansion we have $\hat{\lambda}=O_p(n^{-1/2}).$ Let $
\psi_\gamma(u):=\left(\frac{1}{\gamma+1}-\gamma u\right)^{1/\gamma}$, with $u=\delta+\lambda^{\prime}g_i(\theta)$. Under the smoothness conditions for the CRPD weights,
$
\frac1n\sum_{i=1}^n |w_i-1| = O_p(n^{-1/2}).
$ To see why, because the CRPD weights are smooth in $(\lambda,\delta)$, a first-order Taylor expansion around the population solution $(\lambda,\delta)=(0,\delta_0)$ yields $
w_i(\theta_0,\lambda,\delta)
=
w_i(\theta_0,0,\delta_0)
+
\psi_\gamma^{\prime}(u_0)\,\lambda^{\prime} g_i(\theta_0)
+
O\!\left(\|\lambda\|^2\right).
$\footnote{Fix $\theta=\theta_0$ and write the CRPD weight in terms of the scalar index
\[
u_i(\lambda,\delta):=\delta+\lambda^{\prime} g_i(\theta_0),
\qquad
w_i(\theta_0,\lambda,\delta)=\psi_\gamma\!\big(u_i(\lambda,\delta)\big),
\]
where $\psi_\gamma(u):=\left(\frac{1}{\gamma+1}-\gamma u\right)^{1/\gamma}$.
Let $(\lambda_0,\delta_0)=(0,\delta_0)$ denote the population solution and set $u_0:=u_i(0,\delta_0)=\delta_0$.

\medskip
\noindent\textbf{Step 1 (Taylor expansion in the scalar index).} Since $\psi_\gamma(\cdot)$ is differentiable at $u_0$, Taylor's theorem implies
$
\psi_\gamma(u)
=
\psi_\gamma(u_0)
+
\psi_\gamma^{\prime}(u_0)(u-u_0)
+
O\!\big((u-u_0)^2\big),
\; u\to u_0.
$

Applying this with $u=u_i(\lambda,\delta)$ gives

\[
w_i(\theta_0,\lambda,\delta)
=
w_i(\theta_0,0,\delta_0)
+
\psi_\gamma^{\prime}(u_0)\Big[(\delta-\delta_0)+\lambda^{\prime} g_i(\theta_0)\Big]
+
O\!\left(\Big|(\delta-\delta_0)+\lambda^{\prime} g_i(\theta_0)\Big|^2\right).
\]

\medskip
\noindent\textbf{Step 2 (Control the remainder).} Using $(a+b)^2\le 2a^2+2b^2$,
$
\Big((\delta-\delta_0)+\lambda^{\prime} g_i(\theta_0)\Big)^2
\le
2(\delta-\delta_0)^2
+
2\big(\lambda^{\prime} g_i(\theta_0)\big)^2.
$
Under the maintained rates $\lambda=O_p(n^{-1/2})$ and $\delta-\delta_0=O_p(n^{-1})$, and assuming
$E\|g(Z,\theta_0)\|^2<\infty$, we have
$
(\delta-\delta_0)^2=O_p(n^{-2})=o_p(\|\lambda\|^2),
\;
\big(\lambda^{\prime} g_i(\theta_0)\big)^2=O_p(\|\lambda\|^2).
$ Hence
$
O\!\left(\Big|(\delta-\delta_0)+\lambda^{\prime} g_i(\theta_0)\Big|^2\right)
=
O_p(\|\lambda\|^2).
$

\medskip
\noindent\textbf{Step 3 (Absorb the $\delta$ term).} Moreover,
$
\psi_\gamma^{\prime}(u_0)(\delta-\delta_0)=O_p(n^{-1})=o_p(\|\lambda\|),
$
so this term is of smaller order than the $\lambda^{\prime} g_i(\theta_0)$ term and can be absorbed into the
$O_p(\|\lambda\|^2)$ remainder at the scale of interest.

\medskip
\noindent\textbf{Step 4 (Simplified expansion).} Combining the above yields
$
w_i(\theta_0,\lambda,\delta)
=
w_i(\theta_0,0,\delta_0)
+
\psi_\gamma^{\prime}(u_0)\,\lambda^{\prime} g_i(\theta_0)
+
O_p(\|\lambda\|^2),
$
which is the desired form.}

Under correct specification, $\lambda_0=0$ and $\delta_0=u_0=-1/(\gamma+1)$, so $w_i(\theta_0,0,\delta_0)=\psi_\gamma(u_0)=1$ and $\psi_\gamma^{\prime}(u_0)=-1$. Hence
$
w_i(\theta_0,\lambda,\delta)
=
1
-\lambda^{\prime} g_i(\theta_0)
+
O\!\left(\|\lambda\|^2\right).
$ Hence
$
|w_i-1|
\le
C\,|\lambda^{\prime} g_i(\theta_0)| + O(\|\lambda\|^2)
$
for some constant $C>0$.

Since $\lambda=O_p(n^{-1/2})$, 
$
|w_i-1| = O_p(n^{-1/2})
$
for each $i$. Averaging across $i$ and using $E\|g(Z,\theta_0)\|<\infty$ gives
$
\frac1n\sum_{i=1}^n |w_i-1|
=
O_p(n^{-1/2}).
$ Moreover
$
\frac1n\sum_{i=1}^n \|G_i(\hat{\theta})\| = O_p(1).
$
Therefore
$
\|R_{w,G}\|
\le
\Big(\frac1n\sum_{i=1}^n |w_i-1|\,\|G_i(\hat{\theta})\|\Big)
\|\hat{\lambda}\|
=
O_p(n^{-1}).
$
Thus
$
0=\bar G(\hat{\theta})^{\prime}\hat{\lambda}+O_p(n^{-1}).
$ \\

\textbf{Step 3: Taylor expansion of $\bar G(\hat{\theta})$.}

Expand $\bar G(\hat{\theta})$ around $\theta_0$:
$
\bar G(\hat{\theta})
=
\bar G(\theta_0)
+
\bar H(\theta_0)(\hat{\theta}-\theta_0)
+
R_G,
$ where $
\bar H(\theta_0)=\frac1n\sum_{i=1}^n H_i(\theta_0),
\;
\|R_G\|=O_p(\|\hat{\theta}-\theta_0\|^2).
$ Since $\hat{\theta}-\theta_0=O_p(n^{-1/2})$,
the remainder satisfies
$
R_G=O_p(n^{-1}).
$ Substituting into the first-order condition gives
$
0
=
\Big(
\bar G(\theta_0)
+
\bar H(\theta_0)(\hat{\theta}-\theta_0)
\Big)^{\prime}\hat{\lambda}
+
O_p(n^{-1}).
$\\

\textbf{Step 4: Expansion of $\hat{\lambda}$.}

From Theorem \ref{thm:lambda_second},
$
\hat{\lambda}
=
\hat{\Omega}(\theta_0)^{-1}\bar g(\theta_0)
+
\frac1n B_{\lambda,n}(\gamma)
+
o_p(n^{-1}),
$ with $B_{\lambda,n}(\gamma)=O_p(1)$. Substituting this expansion gives
\begin{equation} \label{eq:expansion_multiplier}
0
=
\bar G(\theta_0)^{\prime}\hat{\Omega}(\theta_0)^{-1}\bar g(\theta_0)
+
\bar H(\theta_0)^{\prime}(\hat{\theta}-\theta_0)
\hat{\Omega}(\theta_0)^{-1}\bar g(\theta_0)
+
\frac1n\bar G(\theta_0)^{\prime}B_{\lambda,n}(\gamma)
+
o_p(n^{-1}),
\end{equation}
with $\frac{1}{n}\bar{H}(\theta_{0})^{\prime}(\hat{\theta}-\theta_{0})B_{\lambda,n}=O_p(n^{-3/2})=o_p(n^{-1})$.\\

\textbf{Step 5: Replace sample matrices with population limits.}

Since $\hat\theta-\theta_0=O_p(n^{-1/2})$, $\bar g(\theta_0)=O_p(n^{-1/2})$, $\hat\Omega(\theta_0)^{-1}=O_p(1)$,
and $\bar H(\theta_0)=O_p(1)$, we have
$
\bar H(\theta_0)^{\prime}(\hat\theta-\theta_0)\hat\Omega(\theta_0)^{-1}\bar g(\theta_0)
=
O_p(n^{-1}).
$ Hence we can absorb this term into the order-$n^{-1}$ remainder by defining a new random scalar
$\widetilde\Delta_{H,n}$ such that
$
\bar H(\theta_0)^{\prime}(\hat\theta-\theta_0)\hat\Omega(\theta_0)^{-1}\bar g(\theta_0)
=
\frac{1}{n}\widetilde\Delta_{H,n}
+
o_p(n^{-1}),
\; \widetilde\Delta_{H,n}=O_p(1).
$
Then \eqref{eq:expansion_multiplier} becomes
\begin{equation}\label{eq:envelope_FOC}
0
=
\bar G(\theta_0)^{\prime}\hat\Omega(\theta_0)^{-1}\bar g(\theta_0)
+
\frac{1}{n}\bar G(\theta_0)^{\prime}B_{\lambda,n}(\gamma)
+
\frac{1}{n}\widetilde\Delta_{H,n}
+
o_p(n^{-1}).
\end{equation}

Moreover, a first-order Taylor expansion of the sample moment yields
\begin{equation}\label{eq:expansion_g}
\bar g(\hat\theta)
=
\bar g(\theta_0)
+
\bar G(\theta_0)(\hat\theta-\theta_0)
+
O_p(\|\hat\theta-\theta_0\|^2)
=
\bar g(\theta_0)
+
\bar G(\theta_0)(\hat\theta-\theta_0)
+
O_p(n^{-1}).
\end{equation}
Premultiplying by $\bar G(\theta_0)^{\prime}\hat\Omega(\theta_0)^{-1}$ gives
\begin{align}\label{eq:premult}
\bar G(\theta_0)^{\prime}\hat\Omega(\theta_0)^{-1}\bar g(\hat\theta)
&=
\bar G(\theta_0)^{\prime}\hat\Omega(\theta_0)^{-1}\bar g(\theta_0)
+
\bar G(\theta_0)^{\prime}\hat\Omega(\theta_0)^{-1}\bar G(\theta_0)(\hat\theta-\theta_0)
+
O_p(n^{-1}).
\end{align}
Rearranging \eqref{eq:premult} yields the identity
\begin{equation}\label{eq:add-subtract}
\bar G(\theta_0)^{\prime}\hat\Omega(\theta_0)^{-1}\bar g(\theta_0)
=
\bar G(\theta_0)^{\prime}\hat\Omega(\theta_0)^{-1}\bar g(\hat\theta)
-
\bar G(\theta_0)^{\prime}\hat\Omega(\theta_0)^{-1}\bar G(\theta_0)(\hat\theta-\theta_0)
+
O_p(n^{-1}).
\end{equation}

Substituting \eqref{eq:add-subtract} into \eqref{eq:envelope_FOC} gives
\begin{align}\label{eq:envelope_FOC2}
0
&=
\bar G(\theta_0)^{\prime}\hat\Omega(\theta_0)^{-1}\bar g(\hat\theta)
-
\bar G(\theta_0)^{\prime}\hat\Omega(\theta_0)^{-1}\bar G(\theta_0)(\hat\theta-\theta_0)
+\frac{1}{n}\bar G(\theta_0)^{\prime}B_{\lambda,n}(\gamma)
+\frac{1}{n}\widetilde\Delta_{H,n}
+o_p(n^{-1}).
\end{align}

\medskip
\noindent\textbf{(i) First term.}
Since $\bar G(\theta_0)=G_0+O_p(n^{-1/2})$ and $\hat\Omega(\theta_0)^{-1}=\Omega_0^{-1}+O_p(n^{-1/2})$, we have
$
\bar G(\theta_0)^{\prime}\hat\Omega(\theta_0)^{-1}
=
(G_0+O_p(n^{-1/2}))^{\prime}(\Omega_0^{-1}+O_p(n^{-1/2}))
=
G_0^{\prime}\Omega_0^{-1}+O_p(n^{-1/2}).
$
Moreover, $\bar g(\hat\theta)=O_p(n^{-1/2})$ under the maintained assumptions. 
Thus,
\begin{equation}\label{eq:first_term_repl}
\bar G(\theta_0)^{\prime}\hat\Omega(\theta_0)^{-1}\bar g(\hat\theta)
=
G_0^{\prime}\Omega_0^{-1}\bar g(\hat\theta)
+
O_p(n^{-1}).
\end{equation}

\medskip
\noindent\textbf{(ii) Second term.}
Using $\bar G(\theta_0)=G_0+O_p(n^{-1/2})$ and $\hat\Omega(\theta_0)^{-1}=\Omega_0^{-1}+O_p(n^{-1/2})$,
\begin{align*}
\bar G(\theta_0)^{\prime}\hat\Omega(\theta_0)^{-1}\bar G(\theta_0)
&=
(G_0+O_p(n^{-1/2}))^{\prime}(\Omega_0^{-1}+O_p(n^{-1/2}))(G_0+O_p(n^{-1/2})) \\
&=
G_0^{\prime}\Omega_0^{-1}G_0
+O_p(n^{-1/2}),
\end{align*}
and therefore, since $\hat\theta-\theta_0=O_p(n^{-1/2})$,
\begin{equation}\label{eq:second_term_repl}
\bar G(\theta_0)^{\prime}\hat\Omega(\theta_0)^{-1}\bar G(\theta_0)(\hat\theta-\theta_0)
=
G_0^{\prime}\Omega_0^{-1}G_0(\hat\theta-\theta_0)
+
O_p(n^{-1}).
\end{equation}

\medskip
\noindent\textbf{(iii) Third term.}
Since $\bar G(\theta_0)=G_0+O_p(n^{-1/2})$ and $B_{\lambda,n}(\gamma)=O_p(1)$,
\begin{align}
\frac{1}{n}\bar G(\theta_0)^{\prime}B_{\lambda,n}(\gamma)
&=
\frac{1}{n}G_0^{\prime}B_{\lambda,n}(\gamma)
+
\frac{1}{n}O_p(n^{-1/2})O_p(1) \notag\\
&=
\frac{1}{n}G_0^{\prime}B_{\lambda,n}(\gamma)
+
o_p(n^{-1}).
\label{eq:third_term_repl}
\end{align}

\medskip
\noindent\textbf{(iv) Fourth term.}
By construction, $\widetilde\Delta_{H,n}=O_p(1)$, hence
\begin{equation}\label{eq:fourth_term_repl}
\frac{1}{n}\widetilde\Delta_{H,n}=O_p(n^{-1}).
\end{equation}

\medskip
\noindent\textbf{Substitution into \eqref{eq:envelope_FOC2}.}
Substituting \eqref{eq:first_term_repl}--\eqref{eq:fourth_term_repl} into \eqref{eq:envelope_FOC2} yields
$
0
=
G_0^{\prime}\Omega_0^{-1}\bar g(\hat\theta)
-
G_0^{\prime}\Omega_0^{-1}G_0(\hat\theta-\theta_0)
+
\frac{1}{n}G_0^{\prime}B_{\lambda,n}(\gamma)
+
\frac{1}{n}\widetilde\Delta_{H,n}
+
O_p(n^{-1})
+
o_p(n^{-1}).
$
Absorbing the $O_p(n^{-1})$ terms into a single $\frac{1}{n}\Delta_n(\gamma)$ term gives
$
0
=
G_0^{\prime}\Omega_0^{-1}\bar g(\hat\theta)
-
G_0^{\prime}\Omega_0^{-1}G_0(\hat\theta-\theta_0)
+
\frac{1}{n}\Delta_n(\gamma)
+
o_p(n^{-1}),
$
where $\Delta_n(\gamma)=O_p(1)$ collects all order-$n^{-1}$ terms, including 
$G_0^{\prime}B_{\lambda,n}(\gamma)$, $\widetilde\Delta_{H,n}$, and additional plug-in
terms arising from replacing the sample matrices $\bar G(\theta_0)$ and 
$\hat{\Omega}(\theta_0)^{-1}$ by their population limits $G_0$ and 
$\Omega_0^{-1}$.\footnote{These plug-in terms are generated by the deviations 
$\bar G(\theta_0)-G_0$ and $\hat{\Omega}(\theta_0)^{-1}-\Omega_0^{-1}$, both of which are 
$O_p(n^{-1/2})$, and therefore contribute $O_p(n^{-1})$ when multiplied by 
$O_p(n^{-1/2})$ quantities such as $\bar g(\theta_0)$ or $(\hat{\theta}-\theta_0)$.}

Since $G_0^{\prime}\Omega_0^{-1}G_0$ is nonsingular, we can rearrange the preceding equation to obtain
$
G_0^{\prime}\Omega_0^{-1}G_0(\hat\theta-\theta_0)
=
G_0^{\prime}\Omega_0^{-1}\bar g(\hat\theta)
+
\frac{1}{n}\Delta_n(\gamma)
+
o_p(n^{-1}).
$
Premultiplying both sides by $\big(G_0^{\prime}\Omega_0^{-1}G_0\big)^{-1}$ yields
$
\hat\theta-\theta_0
=
\big(G_0^{\prime}\Omega_0^{-1}G_0\big)^{-1}G_0^{\prime}\Omega_0^{-1}\bar g(\hat\theta)
+
\frac{1}{n}\big(G_0^{\prime}\Omega_0^{-1}G_0\big)^{-1}\Delta_n(\gamma)
+
o_p(n^{-1}).
$ The dependence of $\Delta_n(\gamma)$ on $\gamma$ arises
through the second-order expansion of the Lagrange multiplier
$B_{\lambda,n}(\gamma)$, which depends on $\gamma$,
as well as through the curvature term $\widetilde\Delta_{H,n}$ that
interacts with the multiplier expansion. Consequently, $\Delta_n(\gamma)$
captures how $\gamma$ affects the $n^{-1}$ (second-order)
behavior of the estimator.
\qed

\subsection{Proof for Lemma \ref{lem:system_risk_uniform_convergence}} \label{subsec:pf_uniform-conv-system-risk}
For $j\in{\theta,\lambda}$, define
\begin{equation}
\widehat R_{j,n}
=
\max_{\gamma\in\Gamma}\widehat C_{j,\mathrm{raw}}(\gamma)
-
\min_{\gamma\in\Gamma}\widehat C_{j,\mathrm{raw}}(\gamma),
\end{equation}
and
\begin{equation}
R_{j}
=
\max_{\gamma\in\Gamma}C_{j}^{\mathrm{raw}}(\gamma)
-
\min_{\gamma\in\Gamma}C_{j}^{\mathrm{raw}}(\gamma).
\end{equation}
Uniform convergence of the raw components implies
\begin{equation}
\widehat R_{j,n}
\xrightarrow{p}
R_{j}.
\end{equation}
Since $R_{j}>0$, the denominator in the normalized sample criterion is bounded
away from zero with probability approaching one.

Next, uniform convergence of
$\widehat C_{j,\mathrm{raw}}(\gamma)$ to $C_{j}^{\mathrm{raw}}(\gamma)$ implies
uniform convergence of the corresponding minima and maxima. In particular,
\begin{equation}
\left|
\min_{\gamma\in\Gamma}\widehat C_{j,\mathrm{raw}}(\gamma)
-
\min_{\gamma\in\Gamma}C_{j}^{\mathrm{raw}}(\gamma)
\right|
\leq
\sup_{\gamma\in\Gamma}
\left|
\widehat C_{j,\mathrm{raw}}(\gamma)
-
C_{j}^{\mathrm{raw}}(\gamma)
\right|,
\end{equation}
and the same argument applies to the maxima. Therefore,
\begin{equation}
\sup_{\gamma\in\Gamma}
\left|
\widehat{\widetilde C}*{j}(\gamma)
-
\widetilde C*{j}(\gamma)
\right|
\xrightarrow{p}0,
\qquad
j\in{\theta,\lambda}.
\end{equation}
Since
\begin{equation}
\widehat Q_{n}(\gamma;\tau)
=
\tau \widehat{\widetilde C}*{\theta}(\gamma)
+
(1-\tau)\widehat{\widetilde C}*{\lambda}(\gamma),
\end{equation}
and
\begin{equation}
Q(\gamma;\tau)
=
\tau \widetilde C_{\theta}(\gamma)
+
(1-\tau)\widetilde C_{\lambda}(\gamma),
\end{equation}
we obtain
\begin{align}
\sup_{\gamma\in\Gamma}
\left|
\widehat Q_{n}(\gamma;\tau)
-
Q(\gamma;\tau)
\right|
&\leq
\tau
\sup_{\gamma\in\Gamma}
\left|
\widehat{\widetilde C}*{\theta}(\gamma)
-
\widetilde C*{\theta}(\gamma)
\right|
\nonumber\\
&\quad
+
(1-\tau)
\sup_{\gamma\in\Gamma}
\left|
\widehat{\widetilde C}*{\lambda}(\gamma)
-
\widetilde C*{\lambda}(\gamma)
\right|
\nonumber\\
&=
o_{p}(1).
\end{align}
This proves \eqref{eq:uniform_convergence_gamma}. \qed

\subsection{Proof for Theorem \ref{thm:gamma_hat_consistency}} \label{subsec:thm_consistency_gammahat}

By Lemma \ref{lem:system_risk_uniform_convergence},
\begin{equation}
\sup_{\gamma\in\Gamma}
\left|
\widehat Q_{n}(\gamma;\tau)
-
Q(\gamma;\tau)
\right|
\xrightarrow{p}0.
\end{equation}
The result then follows from the argmin theorem. For completeness, fix
$\varepsilon>0$ and define
\begin{equation}
\Gamma_{\varepsilon}(\tau)
=
\left\{
\gamma\in\Gamma:
|\gamma-\gamma^{\star}(\tau)|\geq \varepsilon
\right\}.
\end{equation}
By compactness, continuity, and uniqueness of the minimizer,
\begin{equation}
\eta_{\varepsilon}
=
\inf_{\gamma\in\Gamma_{\varepsilon}(\tau)}
\left[
Q(\gamma;\tau)
-
Q(\gamma^{\star}(\tau);\tau)
\right]
>
0.
\end{equation}
On the event
\begin{equation}
\sup_{\gamma\in\Gamma}
\left|
\widehat Q_{n}(\gamma;\tau)
-
Q(\gamma;\tau)
\right|
<
\frac{\eta_{\varepsilon}}{3},
\end{equation}
any $\gamma\in\Gamma_{\varepsilon}(\tau)$ satisfies
\begin{align}
\widehat Q_{n}(\gamma;\tau)
&\geq
Q(\gamma;\tau)
-
\frac{\eta_{\varepsilon}}{3}
\nonumber\\
&\geq
Q(\gamma^{\star}(\tau);\tau)
+
\frac{2\eta_{\varepsilon}}{3}
\nonumber\\
&>
\widehat Q_{n}(\gamma^{\star}(\tau);\tau).
\end{align}
Hence no minimizer of $\widehat Q_{n}(\gamma;\tau)$ can lie outside the
$\varepsilon$-neighborhood of $\gamma^{\star}(\tau)$ on this event. Since the
event has probability approaching one,
\begin{equation}
P\left(
|\hat\gamma(\tau)-\gamma^{\star}(\tau)|\geq \varepsilon
\right)
\to 0.
\end{equation}
Therefore,
\begin{equation}
\hat\gamma(\tau)\xrightarrow{p}\gamma^{\star}(\tau).
\end{equation} \qed

\subsection{Proof for Proposition \ref{prop:post_selection_first_order_validity}} \label{subsec:prop_post-selection-first-order-validity}
By \eqref{eq:post_selection_rootn_expansion},
\begin{equation}
\sqrt n
\left(
\hat\theta_{\hat\gamma(\tau)}-\theta_{0}
\right)
=
\sqrt n A_{0}\bar g(\theta_{0})
+
o_{p}(1).
\end{equation}
By the central limit theorem,
\begin{equation}
\sqrt n A_{0}\bar g(\theta_{0})
\xrightarrow{d}
N(0,V_{\theta}).
\end{equation}
Therefore,
\begin{equation}
\sqrt n
\left(
\hat\theta_{\hat\gamma(\tau)}-\theta_{0}
\right)
\xrightarrow{d}
N(0,V_{\theta}).
\end{equation}
Consistency of $\widehat V_{\theta}$ and Slutsky's theorem imply the stated
coverage result. \qed

\subsection{Proof for Proposition \ref{prop:post_selection_bias_mse}} \label{subsec:pf_prop-post-selection-bias-mse}
Rewrite \eqref{eq:post_selection_expansion_with_remainder} as
\begin{equation}
\hat\theta_{\hat\gamma(\tau)}-\theta_{0}
=
\frac{1}{\sqrt n}Z_{n}
+
\frac{1}{n}B_{\theta,n}(\gamma^{\star}(\tau))
+
r_{\theta,n}.
\end{equation}
Taking expectations and using $E[\bar g(\theta_{0})]=0$ gives
\begin{equation}
E[Z_{n}]=0.
\end{equation}
Together with $E[r_{\theta,n}]=o(n^{-1})$, this gives
\eqref{eq:post_selection_bias}.

For the mean squared error matrix, expand the outer product:
\begin{align}
&\left(
\hat\theta_{\hat\gamma(\tau)}-\theta_{0}
\right)
\left(
\hat\theta_{\hat\gamma(\tau)}-\theta_{0}
\right)^{\prime}
\nonumber\\
&\quad =
\frac{1}{n}Z_{n}Z_{n}^{\prime}
+
\frac{1}{n^{3/2}}
\left[
Z_{n}B_{\theta,n}(\gamma^{\star}(\tau))^{\prime}
+
B_{\theta,n}(\gamma^{\star}(\tau))Z_{n}^{\prime}
\right]
\nonumber\\
&\qquad
+
\frac{1}{n^{2}}
B_{\theta,n}(\gamma^{\star}(\tau))
B_{\theta,n}(\gamma^{\star}(\tau))^{\prime}
+
\text{terms involving } r_{\theta,n}.
\end{align}
The uniform integrability and mean-square negligibility assumptions on
$r_{\theta,n}$ imply that the terms involving $r_{\theta,n}$ are $o(n^{-2})$ after
taking expectations. This yields \eqref{eq:post_selection_mse_matrix}. \qed

\subsection{Proof for Proposition \ref{prop:coverage_error_representation}}
\label{subsec:pf_prop-coverage-error-representation}
Let
\begin{equation}
F_{n}(t)
=
P
\left(
T_{n}(\hat\gamma(\tau))\leq t
\right).
\end{equation}
Then
\begin{equation}
P
\left(
|T_{n}(\hat\gamma(\tau))|\leq z_{1-\alpha/2}
\right)
=
F_{n}(z_{1-\alpha/2})
-
F_{n}(-z_{1-\alpha/2}).
\end{equation}
Using Assumption \ref{ass:edgeworth_post_selection} at
$t=z_{1-\alpha/2}$ and $t=-z_{1-\alpha/2}$ gives
\begin{align}
F_{n}(z_{1-\alpha/2})
-
F_{n}(-z_{1-\alpha/2})
&=
\Phi(z_{1-\alpha/2})
-
\Phi(-z_{1-\alpha/2})
\nonumber\\
&\quad
+
n^{-1/2}
\left[
p_{1}(z_{1-\alpha/2};\gamma^{\star}(\tau))
\phi(z_{1-\alpha/2})
\right.
\nonumber\\
&\qquad
\left.
-
p_{1}(-z_{1-\alpha/2};\gamma^{\star}(\tau))
\phi(-z_{1-\alpha/2})
\right]
+
o(n^{-1/2}).
\end{align}
Since
\begin{equation}
\Phi(z_{1-\alpha/2})-\Phi(-z_{1-\alpha/2})
=
1-\alpha,
\end{equation}
the result follows. \qed

\subsection{Boundary behavior and multiplier stability}\label{subsec:boundary_problem}

The boundary issue arises because the implied CRPD weights must lie in the probability simplex,
\begin{equation}
\Delta_{n}
=
\left\{
\pi=(\pi_{1},\ldots,\pi_{n})^{\prime}:
\pi_{i}\geq 0,\;
\sum_{i=1}^{n}\pi_{i}=1
\right\}.
\end{equation}
An interior solution satisfies $\hat\pi_{i,\gamma}>0$ for all $i$, preferably with weights not too far from the uniform benchmark $1/n$. A boundary solution arises when some weights are zero or close to zero. In the CRPD problem, the weights are determined by
\begin{equation}
\hat\pi_{i,\gamma}
=
\frac{1}{n}
\left[
\frac{1}{\gamma+1}
-
\gamma \hat\delta_{\gamma}
-
\gamma \hat\lambda_{\gamma}^{\prime}
g_{i}(\hat\theta_{\gamma})
\right]^{1/\gamma}.
\end{equation}
Hence, instability in $\hat\lambda_{\gamma}$ directly translates into instability in the implied weights. If the scalar adjustment term $\hat\lambda_{\gamma}^{\prime}g_{i}(\hat\theta_{\gamma})$ is large for some observations, the bracketed term may move close to zero. As a result, some implied weights may become very small, and the solution may approach the boundary of $\Delta_{n}$.

This boundary behavior matters because the weighted moment restriction
\begin{equation}
\sum_{i=1}^{n}
\hat\pi_{i,\gamma}g_{i}(\hat\theta_{\gamma})
=
0
\end{equation}
may then be satisfied only by effectively suppressing part of the sample. Observations with $\hat\pi_{i,\gamma}\approx 0$ contribute little to the empirical moment equation, while observations with relatively large weights may become disproportionately influential. Thus, even if the resulting $\hat\theta_{\gamma}$ appears stable, the underlying moment enforcement mechanism may be fragile because it relies on highly uneven empirical reweighting.

Near-boundary weights also reduce the effective amount of information used by the estimator. A useful diagnostic is the effective sample size,
\begin{equation}
ESS(\gamma)
=
\frac{1}{
\sum_{i=1}^{n}\hat\pi_{i,\gamma}^{2}
}.
\end{equation}
If the weights are uniform, then $ESS(\gamma)=n$. If the weights are concentrated on a small subset of observations, then $ESS(\gamma)$ can be much smaller than $n$. Therefore, large or unstable multipliers can make the estimator behave as if it were based on a smaller effective sample, even when the nominal sample size is unchanged.

The boundary issue is also important for asymptotic analysis. The higher-order expansions of the estimator and the multipliers rely on smoothness of the implied-weight mapping around an interior solution. If some weights approach zero, the mapping from $(\theta,\lambda,\delta)$ to $\pi$ can become highly nonlinear or nonsmooth near the boundary. This can invalidate the local Taylor expansions and implicit-function arguments used to obtain the second-order representations. Therefore, the multiplier component in the selection criterion is not a penalty on a nuisance parameter. It is a way to discourage values of $\gamma$ for which the moment restrictions are satisfied only through unstable, near-boundary reweighting.

This does not mean that downweighting observations is always undesirable. Robustness often requires reducing the influence of observations associated with large moment residuals. The concern is instead excessive or unstable downweighting. A moderate multiplier correction indicates that the moments can be enforced with limited distortion of the empirical distribution. In contrast, an unstable multiplier indicates a high finite-sample cost of moment stabilization. This provides a direct motivation for including a multiplier-stability component when selecting the CRPD power parameter.

\clearpage
\bibliographystyle{apalike}
\bibliography{main}

@article{CressieandRead1984,
author = {Cressie, Noel and Timothy RC Read},
title = {Multinomial goodness-of-fit tests},
journal = {Journal of the Royal Statistical Society Series B: Statistical Methodology},
volume = {46},
number = {3},
pages = {440-464},
year = {1984}
}

@article{HansenHeatonandYaron1996,
author = {Hansen, L. and Heaton, J. and Yaron, A.},
title = {Finite-Sample Properties of Some
tive GMM Estimators},
journal = {Journal of Business and Economic Statistics},
volume = {14},
number = {},
pages = {262-280},
year = {1996}
}

@article{Huber1964,
author = {Huber, P.},
title = {Robust Estimation of a Location Parameter},
journal = {The Annals of Mathematical Statistics},
volume = {35},
number = {1},
pages = {73-101},
year = {1964}
}

@book{HuberandRonchetti2009,
title={Robust Statistics},
author={Huber, P. and Ronchetti, E.},
volume={},
pages={},
year={2009},
publisher={Wiley}
}

@article{ImbensSpadyandJohnson1998,
author = {Imbens, G. W. and Spady, R. H. and Johnson, P. },
title = {Information theoretic approaches to inference in moment conditions models},
journal = {Econometrica},
volume = {66},
number = {2},
pages = {333-357},
year = {1998}
}

@article{KitamuraandStutzer1997,
author = {Kitamura, Y. and Stutzer, M.},
title = {An information-theoretic alternative to generalized method of moments estimation},
journal = {Econometrica},
volume = {65},
number = {},
pages = {861-874},
year = {1997}
}

@article{NeweyandSmith2004,
author = {Newey, W. K. and Smith, R. J. },
title = {Higher order properties of GMM and generalized empirical likelihood estimators},
journal = {Econometrica},
volume = {72},
number = {1},
pages = {219-255},
year = {2004}
}

@article{Owen1988,
author = {Owen, A. B.},
title = {Empirical likelihood ratio confidence intervals for a single functional},
journal = {Biometrika},
volume = {75},
number = {2},
pages = {237-249},
year = {1988}
}

@book{Owen2001,
title={Empirical likelihood},
author={Owen, A. B.},
volume={},
pages={},
year={2001},
publisher={Chapman and Hall/CRC}
}

@article{QinandLawless1994,
author = {Qin, J. and Lawless, J.},
title = {Empirical likelihood and general estimating equations},
journal = {The Annals of Statistics},
volume = {22},
number = {1},
pages = {300-3235},
year = {1994}
}

@book{Rothenberg1984,
title={Handbook of econometrics (Approximating the distributions of econometric estimators and test statistics)},
author={Rothenberg, T. J.},
volume={2},
pages={881-935},
year={1984},
publisher={}
}

@article{Schennach2007,
author = {Schennach, S.},
title = {Point estimation with exponentially tilted empirical likelihood},
journal = {The Annals of Statistics},
volume = {35},
number = {2},
pages = {634-672},
year = {2007}
}

\end{document}